\documentclass[12pt]{article}
\usepackage{authblk}
\usepackage{geometry}
\usepackage{amsmath}
\usepackage{amssymb}
\usepackage{amsthm}
\usepackage{mathrsfs}
\usepackage{subfigure}
\usepackage{customcommands}
\usepackage{bm}
\usepackage{Nettastyle}

\title{Surface Theory: the Classical, the Quantum, and the Holographic}
\author[1,2]{Netta Engelhardt}
\author[3]{and Sebastian Fischetti}

\affiliation[1]{Department of Physics, Princeton University, Princeton, NJ 08544, USA}
\affiliation[2]{Gravity Initiative, Princeton University, Princeton NJ 08544, USA}
\affiliation[3]{Department of Physics, McGill University, Montr\'eal, QC, H3A 2T8, Canada}

\emailAdd{nengelhardt@princeton.edu}
\emailAdd{fischetti@physics.mcgill.ca}

\abstract{
Motivated by the power of subregion/subregion duality for constraining the bulk geometry in gauge/gravity duality, we pursue a comprehensive and systematic approach to the behavior of extremal surfaces under perturbations.  Specifically, we consider modifications to their boundary conditions, to the bulk metric, and to bulk quantum matter fields.  We present a unified framework for treating such perturbations for classical extremal surfaces, classify some of their stability properties, and develop new technology to extend our treatment to quantum extremal surfaces, culminating in an ``equation of quantum extremal deviation''.  Part of the power of this formalism is due to its ability to map geometric statements into the language of elliptic operators; to illustrate, we show that various \textit{a priori} disparate bulk constraints all follow from basic consistency of subregion/subregion duality.  These include familiar properties such as (smeared) versions of the quantum focusing conjecture and the generalized second law, as well as new constraints on \textit{(i)} metric and matter perturbations in spacetimes close to vacuum and \textit{(ii)} the bulk stress tensor in generic (not necessary close to vacuum) spacetimes. This latter constraint is highly reminiscent of a quantum energy inequality.
}

\begin{document}

\maketitle

\section{Introduction}
\label{sec:intro}

To what do basic aspects of QFT map in an emergent semiclassical bulk in the context of gauge/gravity duality~\cite{Mal97, Wit98a, GubKle98}?  Guided by this question, here we focus on constraints imposed by the consistency of subregion/subregion duality, which translates tautological aspects of algebraic QFT, such as the inclusion of operator algebras of nested causal diamonds, into novel statements about bulk geometry.  Since such properties are so basic as to be axiomatic, it is generally expected that the resulting bulk constraints are equally as fundamental. 


Let us briefly remind the reader that subregion/subregion duality states that for any (globally hyperbolic) subregion~$R$ of the boundary, the operator algebra of~$R$ in the boundary CFT is dual to the so-called entanglement wedge~$W_E[R]$ in the bulk (subregion/subregion duality has technically only been proven within the code subspace~\cite{DonHar16,FauLew17}, i.e.~for quantum bulk fields on a fixed background spacetime, but here we assume it holds for the geometry as well).  The entanglement wedge~$W_E[R]$, in turn, is defined from the entanglement structure of the region~$R$.  First, recall that the HRT formula~\cite{RyuTak06,HubRan07} and its quantum generalization~\cite{FauLew13,EngWal14} states that the von Neumann entropy~$S_\mathrm{vN}[R]$ of~$R$ is given by
\be
S_\mathrm{vN}[R] = S_\mathrm{gen}[X_R],
\ee
where~$X_R$ is a surface homologous to~$R$ and a stationary point of the generalized entropy functional~$S_\mathrm{gen}$~\cite{Bek72}
\be
S_\mathrm{gen}[\Sigma] = \frac{\Area[\Sigma]}{4G_N \hbar} + S_\mathrm{out}[\Sigma],
\ee
with~$S_\mathrm{out}[\Sigma]$ the entropy of any quantum fields ``outside''~$\Sigma$ (if more than one such surface exists,~$X_R$ is the one with smallest generalized entropy).  Note that since in the classical~$\hbar \to 0$ limit~$S_\mathrm{gen}$ is dominated by the area functional, in this limit~$X_R$ is just a classical extremal surface.  For this reason, for surfaces which are stationary points of~$S_\mathrm{gen}$ are referred to as \textit{quantum} extremal surfaces.  Finally, the homology constraint requires the existence of an achronal hypersurface~$H_R$ with boundary~$\partial H_R = X_R \cup R$, from which the entanglement wedge is defined as the domain of dependence of~$H_R$:~$W_E[R] = D[H_R]$.

Because subregion/subregion duality relies so heavily on the HRT surface~$X_R$, the first half of this paper is devoted to studying perturbations of such surfaces; that such perturbations must behave in a way which is consistent with subregion/subregion duality imposes nontrivial constraints on the bulk. As a derivation of such constraints involves highly nonlocal and delicate control over the behavior of surfaces, any such analysis requires a sophisticated set of geometric tools for making progress. We therefore provide a toolkit for systematically studying the variations of objects defined on surfaces under small perturbations of said surfaces; these objects include, for instance, geometric tensors like extrinsic curvatures, but also nonlocal constructs such as the generalized entropy and variations thereof.  Consequently, if we are interested in surfaces defined by some ``equations of motion'' -- of particular interest are quantum extremal surfaces -- this formalism provides a way of analyzing how such surfaces vary under modifications (including changes in boundary conditions, in the ambient geometry, and in the state of bulk matter fields).

To that end, we begin in Section~\ref{sec:classical} with a comprehensive, unified review focused on \textit{classical} extremal surfaces that combines assorted aspects of minimal surface theory~\cite{ColMin}, cosmic branes and strings~\cite{LarFro93,Guv93,VisPar96,Car92,Car92b,Car93,BatCar95,BatCar00}, and classical extremal surfaces in AdS/CFT~\cite{Mos17,GhoMis17,LewPar18}.  We focus on non-null surfaces of arbitrary dimension and signature, though we ultimately specialize to codimension-two spacelike extremal surfaces in Lorentzian geometries.  The utility of this presentation stems partly from the provision of a link between the geometric problem of perturbations of surfaces and elliptic equations, which are well-studied; it is this connection that powers many of our later results.  In Section~\ref{sec:quantum} we then derive equations governing the behavior of \textit{quantum} extremal surfaces under perturbations of the state and of their boundary conditions.  To do this, we develop a covariant treatment of distribution-valued tensor functionals on surfaces, including both functional covariant and Lie derivatives.  Finally, using the aforementioned connection between extremal surface perturbations and elliptic equations, we complete our analysis in Section~\ref{sec:stability} by describing different notions of stability of extremal surfaces.  This permits, for instance, the prescription of a rigorous mathematical criterion for the existence of extremal surfaces under spacetime perturbations.

We then proceed to use this formalism towards its described purpose: deriving bulk constraints from subregion/subregion duality.  This duality manifests itself in two key ways.  The first follows from the observation that the causal wedge~$W_C[R]$ of~$R$ -- defined as the intersection of the past and future of~$R$ in the bulk -- is dual to one-point functions on~$R$ (since it can be recovered by essentially ``integrating in'' the equations of motion from the boundary).  Because these one-point functions are a proper subset of the full information accessible from the state and operator algebra of~$R$, subregion/subregion duality implies that~$W_C[R]$ must be contained within~$W_E[R]$; we refer to this as \textit{causal wedge inclusion} (CWI).  More formally, if~$\Acal_R$ is the operator algebra of~$R$, then the set of operators that compute one-point functions in~$R$ is a proper subset of~$\Acal_R$.  It follows, in particular, that the causal surface~$C_R$, which is essentially the ``rim'' of~$W_C[R]$, must be achronally-separated from the HRT surface~$X_R$, as shown in Figure~\ref{subfig:CWI}.

\begin{figure}[t]
\centering
\subfigure[]{
\includegraphics[page=1,width=0.3\textwidth]{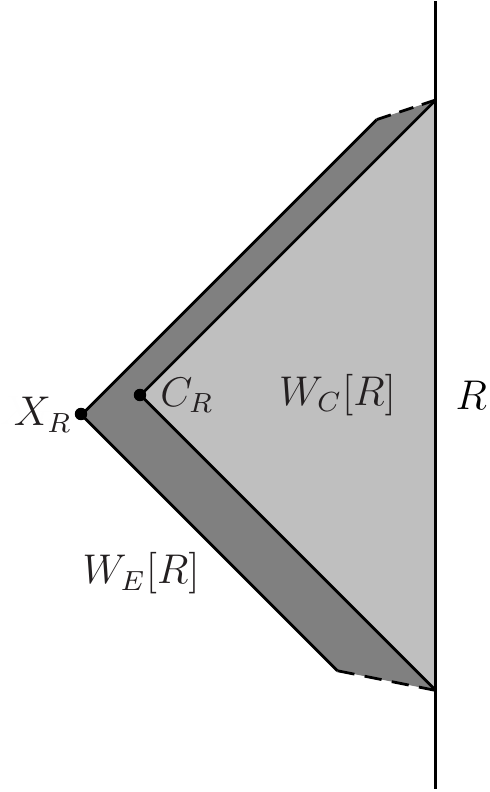}
\label{subfig:CWI}
}
\hspace{2cm}
\subfigure[]{
\includegraphics[page=2,width=0.3\textwidth]{Figures-pics}
\label{subfig:EWN}
}
\caption{\subref{subfig:CWI}: CWI requires the causal wedge~$W_C[R]$ to lie inside the entanglement wedge~$W_E[R]$, and therefore the causal surface~$C_R$ must be achronally separated and to the outside of the HRT surface~$X_R$.  \subref{subfig:EWN}: EWN requires the entanglement wedge~$W_E[R]$ to shrink into itself as the boundary region~$R$ is shrunk; this implies that the HRT surface~$X_R$ must move in an achronal direction towards~$R$.  (The dashed lines indicate caustics and intersections of generators of~$\partial W_E[R]$.)}
\label{figs:nesting}
\end{figure}

The second key manifestation of subregion/subregion duality consistency  is \textit{entanglement wedge nesting} (EWN). This follows from the inclusion of operator algebras of nested causal diamonds on the boundary: if $D[R_1] \subset D[R_2]$, then the algebras nest as well, i.e.~${\cal A}_{R_1}\subset {\cal A}_{R_2}$.  Since subregion/subregion duality requires that the bulk duals to the algebras~$\Acal_{R_1}$ and~$\Acal_{R_2}$ be the entanglement wedges~$W_E[R_1]$ and~$W_E[R_2]$, the nesting of~$\Acal_{R_1}$ and~$\Acal_{R_2}$ implies the nesting of the entanglement wedges:~$W_E[R_1] \subset W_E[R_2]$.  This property is sketched in Figure~\ref{subfig:EWN}. (In fact, under appropriate implicit regularity assumptions~\cite{AkeKoe16} showed that EWN implies CWI; here we treat the two as separate constraints both for pedagogical and computational clarity and also to minimize our assumptions on the bulk.) {\color{white} Entanglement wedges nest like vampires.}

Because CWI and EWN are constraints on how the bulk surfaces~$X_R$ and~$C_R$ must behave, our formalism is precisely the necessary tool for a systematic investigation of constraints that they impose on the bulk geometry.  Specifically, in Section~\ref{sec:causal} we show that, for spacetimes which are a linear perturbation away from pure AdS, classical CWI enforces a highly nontrivial condition on the metric perturbation~$\delta g_{ab}$.  This condition is essentially a refined version of our boundary causality condition (BCC)~\cite{EngFis16}, which constrained the averaged ``tilting'' of a light cone along a complete null geodesic due to the perturbation~$\delta g_{ab}$.  The refined condition that we obtain here instead relates the averaged tilting of light cones along incomplete (null) generators of Rindler horizons to the perturbation of their bifurcation surface.  While we do not give a physical interpretation of this constraint, we note that~\cite{AfkHar17} found that the BCC is intimately related to the chaos bound; conceivably our new constraint may be related to a refined version thereof.  Next, by treating quantum fields in a \textit{fixed} pure AdS spacetime, we also show that CWI enforces a smeared generalized second law (GSL) -- the increase of $S_{gen}$ along slices of a causal horizon -- along Poincar\'e horizons of the bulk.  This latter result is quite pleasantly consistent with the fact that the GSL enforces CWI in this context~\cite{EngWal14}; hence we find that an ``averaged'' version of the converse is true.

Finally, in Section~\ref{sec:EWN} we use maximum principles in elliptic operator theory to deduce more constraints on the bulk from EWN.  First we illustrate the power of the formalism by rederiving the known result that the NEC implies EWN (at leading order in~$1/G_N \hbar$) in a novel method that requires fewer assumptions than the proof of~\cite{EngWal13} and different (incomparable) assumptions from those used in~\cite{Wal12}.  We then derive a general constraint on any classical bulk spacetime: that on any HRT surface~$X_R$ with a null normal~$k^a$, the quantity~$\sigma_k^2 + R_{ab} k^a k^b$ cannot be everywhere-negative, where~$\sigma_k^2$ is the shear of the null congruence generated by~$k^a$ and~$R_{ab}$ is the Ricci tensor.  This combination of terms is what causes nearby null geodesics to ``focus'', and is therefore what makes gravity ``attractive''; it is in this very heuristic sense that we may interpret our result as an energy inequality.  This is strongly reminiscent of spatial quantum energy inequalities, which require negative local energy densities to be accompanied by compensating positive energies elsewhere.  Here we find that ( assuming the Einstein equation), an HRT surface cannot sustain a region of negative $T_{ab}k^{a}k^{b}+\sigma_{k}^{2}/(8\pi G_N)$ without this quantity being positive elsewhere on it.  We emphasize that this result is spacetime-independent, in the sense that we make no assumptions about our spacetime being perturbatively away from the vacuum or weakly curved.  However, when we do restrict to perturbations of the vacuum, we can get a more quantitative constraint: we find that for classical perturbations of pure AdS, EWN imposes that~$R_{ab} k^a k^b$ must be non-negative when integrated over the HRT surface of any ball-shaped region of the boundary (for an appropriate choice of~$k^a$); this is closely related to positive energy theorems obtained from entropic inequalities~\cite{LasRab14,LasLin16,NeuSar18}.  In this perturbative context, we can moreover include quantum corrections to show that EWN enforces a smeared version of the quantum focusing conjecture~\cite{BouFis15} -- which requires appropriate second functional derivatives of~$S_\mathrm{gen}$ in a null direction to be non-positive -- on these HRT surfaces; this last result can be obtained (in a non-holographic context) under a different set of assumptions from the quantum null energy condition on Killing horizons~\cite{BouFis15b}.

Let us make some brief comments.  First, the applications we present in Section~\ref{sec:EWN} are just an example of how the formalism that we present, which relates perturbations of extremal surfaces to elliptic operators, can be used to deduce new information about the bulk; they are far from an exhaustive study of the applications of elliptic operator theory in this context.  For instance, the first application (that we are aware of) of elliptic operator theory to classical extremal surfaces via subregion/subregion duality may be found in~\cite{EngWal17b}, which gave a holographic account of dynamical black hole entropy; a more recent application of elliptic operator theory via the equation of \textit{classical} extremal deviation to bulk reconstruction may be found in~\cite{BaoCao19}. Since this article is the first presentation of the equation of \textit{quantum} extremal deviation, the results discussed in Section~\ref{sec:EWN} constitute the first applications of elliptic operator theory to subregion/subregion duality in the semiclassical regime.  Second, it is worth remarking on the inverse investigation, which assumes consistency of subregion/subregion duality in the bulk and derives constraints on the boundary theory (see e.g.~\cite{KoeLei15,AkeKoe16,KoeLei17,AkeCha17}). This has been used to derive, for instance, the boundary quantum null energy condition~\cite{KoeLei15} before it was broadly derived for quantum field theories in flat space~\cite{BalFau17}. While our formalism is developed with a view towards constraining the bulk, we see no reason why it could not also be used to further the investigation of the boundary physics as well.

\subsection{Surface Theory: A User's Manual}
\label{subsec:manual}

Here we give a streamlined survey of the results that we review and develop in Sections~\ref{sec:classical},~\ref{sec:quantum}, and~\ref{sec:stability}, which essentially answer the following questions: how does an extremal surface -- either classical or quantum -- behave under perturbations to its boundary conditions (if it has a boundary) or to the geometry in which it is embedded?  Under what conditions is this question well-defined?  And under what conditions can a classical extremal surface sensibly be said to be ``minimal'' or ``maximal''?

Although our presentation in the first half of this paper is completely coordinate-independent, here let us begin with a coordinate-based description of surfaces, which we expect is more familiar to most readers.  Consider some~$n$-dimensional surface~$\Sigma$ in an ambient~$d$-dimensional geometry~$(M,g_{ab})$, and introduce a coordinate system~$\{y^\alpha\}$,~$\alpha = 1, \ldots, n$ on~$\Sigma$ and a coordinate system~$\{x^\mu\}$,~$\mu = 1, \ldots, d$ on~$M$.  The surface~$\Sigma$ is given explicitly by specifying~$d$ embedding functions~$X^\mu(y)$ of the coordinates~$y^\alpha$.  Perturbations of~$\Sigma$ are then made precise by introducing a continuous one-parameter family of surfaces~$\Sigma(\lambda)$ with~$\Sigma(\lambda = 0) = \Sigma$, given by a one-parameter family~$X^\mu(\lambda; y)$ of embedding functions that are continuous in~$\lambda$.  The corresponding ``infinitesimal perturbation'' of~$\Sigma$ is captured by the objects~$dX^\mu/d\lambda|_{\lambda = 0}$, which are the components~$\eta^\mu$ (in the coordinate system~$\{x^\mu\}$) of a \textit{deviation vector field}~$\eta^a$ on~$\Sigma$.   Explicitly,
\be
\label{eq:embeddingvariation}
X^\mu(\lambda; y) = X^\mu(y) + \lambda \, \eta^\mu(y) + \Ocal(\lambda^2).
\ee
Understanding the behavior of small perturbations of~$\Sigma$ is therefore tantamount to understanding the behavor of the deviation vector~$\eta^a$.  Of course, if the family~$\Sigma(\lambda)$ is completely arbitrary, then the deviation vector~$\eta^a$ is as well.  But if each surface in the family~$\Sigma(\lambda)$ is constrained somehow -- say, if they are all required to be extremal -- then~$\eta^a$ must also be constrained.

For a familiar example of such a constraint, consider the case where the surfaces~$\Sigma(\lambda)$ are all required to be geodesics; then the component~$\eta^a_\perp$ of~$\eta^a$ normal to~$\Sigma$ obeys the equation of geodesic deviation
\be
\label{eq:geodeviation}
u^c \grad_c(u^b\grad_b \eta^a_\perp) + {R_{bcd}}^a u^b u^d \eta^c_\perp = 0,
\ee
where~$u^a$ is an affinely-parametrized tangent to~$\Sigma$ and~$R_{abcd}$ is the Riemann tensor of~$g_{ab}$.  This equation can be interpreted as either governing the relative acceleration of nearby geodesics due to tidal forces, or alternatively as describing how a particular geodesic deforms in response to a small deformation of its boundary conditions, as shown in Figure~\ref{fig:geodeviation}.  It is this latter interpretation that we will adopt here, though of course the two are completely equivalent.

\begin{figure}[t]
\centering
\includegraphics[page=3,width=0.25\textwidth]{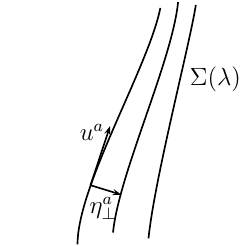}
\caption{The equation of geodesic deviation, which constrains the deviation vector~$\eta^a$ along a one-parameter family~$\Sigma(\lambda)$ of geodesics, can be interpreted as describing the relative acceleration of nearby geodesics in a congruence or alternatively as the perturbation of a geodesic as its boundary conditions are changed.}
\label{fig:geodeviation}
\end{figure}

We would like a generalization of~\eqref{eq:geodeviation} to higher-dimensional (classical or quantum) extremal surfaces\footnote{The terminology ``extremal surface'' is a misnomer, as it refers to a surface that is merely a stationary point of the area functional (or, in the quantum case, of~$S_\mathrm{gen}$), rather than a local extremum of it.  Thus such surfaces should more correctly be called ``stationary'', as argued in~\cite{Bou18}.  Unfortunately, the terminology has stuck, which as far as we can tell originated with the statement of~\cite{HubRan07} that spacelike geodesics in Lorentzian spacetimes extremize proper length.  This is incorrect: for instance, the proper length of a spacelike geodesic in Minkowski space can be either increased or decreased by ``wiggling'' it in a spacelike or timelike direction, respectively.}, and moreover we would like a generalization that includes not just the response of an extremal surface to a perturbation of its boundary conditions, but also of the ambient geometry~$g_{ab}$ (and entropy functional~$S_\mathrm{out}[\Sigma]$, in the case of quantum extremal surfaces).  Noting that the induced metric on a (non-null) geodesic with tangent~$u^a$ is~$h^{ab} = u^a u^b/u^2$, the two terms in~\eqref{eq:geodeviation} can be interpreted as a Laplacian~$D^2 \eta^a_\perp = h^{bc} D_b D_c \eta^a_\perp$ (which will be defined precisely in Section~\ref{subsec:submanifold}) and~$h^{bd} {R_{bcd}}^a \eta_{\perp}^c$, so we might guess that these terms should appear in the generalization of~\eqref{eq:geodeviation} to higher-dimensional surfaces.  This expectation is correct: for a non-null extremal surface~$\Sigma$ with induced metric~$h_{ab}$, the deviation vector field~$\eta^a$ describing a deformation through a family of extremal surfaces must obey what we call here the \textit{equation of extremal deviation} (sometimes also called the Jacobi equation in the literature)
\be
\label{eq:extremaldeviation}
J (\eta_\perp)^a \equiv -D^2 \eta^a_\perp - {S^a}_b \eta^b_\perp - {R_{ced}}^b {P^a}_b h^{cd} \eta^e_\perp = 0,
\ee
where~${P_a}^b \equiv {\delta_a}^b - {h_a}^b$ is the orthogonal projector to~$\Sigma$;~$S^{ab}$ is Simons' tensor\footnote{We issue a disclaimer that this name has nothing to do with our funding sources.}, given explicitly in terms of the extrinsic curvature~${K^a}_{bc} \equiv -{h_b}^d {h_c}^e \grad_d {h_e}^a$ as~$S_{ab} \equiv {K^a}_{cd} K^{bcd}$; and the differential operator~$D^2$ (which we will define more explicitly below) is the Laplacian on the normal bundle of~$\Sigma$.  Equation~\eqref{eq:extremaldeviation} is a homogeneous PDE that governs the behavior of an extremal surface under deformations of its boundary conditions; if we are additionally interested in the behavior of~$\Sigma$ under a perturbation~$\delta g_{ab}$ to the ambient geometry, then~\eqref{eq:extremaldeviation} is sourced by the perturbation~$\delta g_{ab}$:
\be
\label{eq:sourceextremaldeviation}
J \eta_\perp^a = K^{abc} \delta g_{bc} + {P^a}_b h^{cd} \, \delta{\Gamma^b}_{cd},
\ee
where~$\delta{\Gamma^a}_{bc}$ is the perturbation of the Christoffel symbols due to the metric perturbation~$\delta g_{ab}$.  We briefly note that~$S_{ab}$ does not appear in~\eqref{eq:geodeviation} because geodesics have vanishing extrinsic curvature; on the other hand, higher-dimensional extremal surfaces only have vanishing \textit{mean} curvature~$K^a \equiv h^{bc} {K^a}_{bc}$.  Likewise, note that the last term of~\eqref{eq:geodeviation} is normal to~$\Sigma$ due to the symmetries of the Riemann tensor, so it agrees with the last term of~\eqref{eq:extremaldeviation}.

To exploit these equations, it is convenient to decompose the operator~$J$ in a basis of the vectors normal to $\Sigma$.  For simplicity, here let us give the expressions for the case in which we are most interested: namely, when~$\Sigma$ is a codimension-two spacelike surface in a Lorentzian geometry (though we emphasize that this formalism applies to non-null surfaces of arbitrary dimension, codimension, and signature).  Then we may introduce a null basis~$\{k^a, \ell^a\}$ of its normal bundle satisfying~$k \cdot \ell = 1$, and hence decomposing~$\eta^a_\perp = \alpha k^a + \beta \ell^a$, we find that the components of~$J \eta_\perp^a$ are
\begin{subequations}
\label{eqs:Jdecomp}
\begin{align}
k_a J \eta_\perp^a &= -D^2 \beta + 2\chi^a D_a \beta - \left(|\chi|^2 - D_a \chi^a + Q_{k\ell} \right)\beta - Q_{kk} \alpha, \\
\ell_a J \eta_\perp^a &= -D^2 \alpha - 2\chi^a D_a \alpha - \left(|\chi|^2 + D_a \chi^a + Q_{k\ell} \right)\alpha - Q_{\ell\ell} \beta,
\end{align}
\end{subequations}
where~$D^2$ is now the usual scalar Laplacian on~$\Sigma$ and we defined~$\chi^a \equiv \ell^b h^{ac} \grad_c k_b$ and~$Q_{ab} \equiv S_{ab} + h^{cd} {P_a}^e {P_b}^f R_{cedf}$.  Roughly speaking, $\chi^{a}$ is related to frame dragging,~$Q_{kk}$ corresponds to focusing of light rays, and $Q_{ab}k^{a}\ell^{b}$ is related to ``cross-focusing'' of light rays.  The equation of extremal deviation~$J \eta_\perp^a = 0$ is consequently an elliptic system of PDEs.

More generally,~$J$ is an elliptic operator whenever the induced metric~$h_{ab}$ on~$\Sigma$ has fixed sign.  This implies, in particular, that the spectrum of~$J$ is bounded; we use this feature in Section~\ref{sec:stability} to classify two notions of stability of extremal surfaces.  First, what we term \textit{strong stability} is the requirement that an extremal surface be a \textit{bona fide} local extremum (i.e.~maximum or minimum) of the area functional; this notion of stability only makes sense when~$P_{ab}$ also has definite sign, in which case it imposes that the spectrum of~$J$ be bounded by zero.  Second, what we term \textit{weak stability} is the requirement that small perturbations of either the boundary conditions of~$\Sigma$ or of its ambient geometry must correspondingly induce small perturbations of~$\Sigma$; this imposes that the spectrum of~$J$ not contain zero.

In AdS/CFT, the extremal surfaces in which we are most interested are spacelike and codimension-two, since these compute the leading-order (in~$1/N^2 \sim G_N \hbar$) contribution to the entanglement entropy of the dual CFT.  To compute subleading corrections, we must make use of the quantum extremal surfaces defined above, which can be thought of as ``quantum-corrected'' versions of classical extremal surfaces.  We therefore desire a generalization of the equation of extremal deviation~\eqref{eq:sourceextremaldeviation} to include these quantum corrections.  This task is nontrivial due to the fact that unlike the area functional,~$S_\mathrm{gen}$ cannot be expressed as an integral over~$\Sigma$ of local quantities, and therefore variations of~$S_\mathrm{gen}$ will be nonlocal.  To deal with this issue, in Section~\ref{sec:quantum} we develop a \textit{covariant functional derivative}~$\Dcal/\Dcal \Sigma^a$ which computes the variation in tensorial multi-local functionals on a surface under small deformations.  In terms of this operator, quantum extremal surfaces are defined by the condition~$\Dcal S_\mathrm{gen}/\Dcal \Sigma^a = 0$.  The full quantum-corrected version of the sourced equation of extremal deviation~\eqref{eq:sourceextremaldeviation} is presented in equation~\eqref{eq:quantumextremaldeviation}, though here let us focus on two special cases.  First, the quantum analogue of the \textit{unsourced} equation~\eqref{eq:extremaldeviation} is
\be
J \eta_\perp^a(p) + 4 G_N \hbar \int_\Sigma P^{ab}(p) \frac{\Dcal^2 S_\mathrm{out}}{\Dcal \Sigma^c(p') \Dcal \Sigma^b(p)} \, \eta^c(p') = 0,
\ee
which governs the perturbation to a quantum extremal surface under a perturbation of its boundary conditions (this equation holds for each point~$p$ on~$\Sigma$ with the integral taken over~$p'$ with~$p$ fixed).  Second, we also obtain an equation for computing how a classical extremal surface is corrected to a quantum extremal surface:
\be
J\eta_\perp^a = -4G_N \hbar \, \frac{\Dcal S_\mathrm{out}}{\Dcal \Sigma^a}.
\ee
Note that if the classical extremal surface is weakly stable, then this quantum-corrected surface always exists. We note that this equation (and in particular the sign of spacelike components of the right hand side) can be used to determine whether or not quantum effects result in the entanglement wedge moving deeper into the bulk.

\subsubsection*{An Explicit Example}

To illustrate a simple use of the formalism above, let us quickly reproduce a well-known result for the perturbation to the extremal surface anchored to a ball-shaped region on the boundary of pure AdS.  We work with AdS$_d$ spacetime in the coordinates
\be
\label{eq:AdS}
ds^2 = l^2 \left[\frac{\cosh^2\chi}{\rho^2} (-dt^2 + d\rho^2) + d\chi^2 + \sinh^2 \chi d\Omega_{d-3}\right],
\ee
which can be obtained from the usual Poincar\'e coordinates by taking~$z = \rho \sech\chi$ and~$r = \rho \tanh\chi$, so the AdS boundary is at~$\chi \to \infty$ and~$(\rho,\Omega)$ are spherical coordinates on boundary slices of constant~$t$.  For the boundary region given by the sphere with radius~$\rho = \rho_0$ (on any slice~$t = $~const.), the RT surface is just given by~$\rho = \rho_0$ everywhere.  The induced metric on this surface is the hyperbolic ball
\be
\label{eq:hyperbolicball}
ds^2_\mathrm{RT} = l^2 \left(d\chi^2 + \sinh^2 \chi \, d\Omega_{d-3}\right),
\ee
and moreover this surface is totally geodesic, i.e.~${K^a}_{bc} = 0$, and hence~$S_{ab} = 0$.  The null basis
\be
\label{eq:AdSbasis}
k^a = \frac{\rho\sech\chi}{\sqrt{2} \, l} \left[(\partial_t)^a - (\partial_\rho)^a\right], \qquad \ell^a = -\frac{\rho\sech\chi}{\sqrt{2} \, l} \left[(\partial_t)^a + (\partial_\rho)^a\right]
\ee
is normal to these surfaces and satisfies~$k \cdot \ell = 1$, and moreover sets~$\chi^a = 0$.  It is easy to check that since~$S_{ab} = 0$,~$Q_{kk} = Q_{\ell\ell} = 0$ and~$Q_{k\ell} = -(d-2)/l^2$, hence using~\eqref{eqs:Jdecomp} we find that the components of the equation of extremal deviation~$J \eta_\perp^a = 0$ become simply
\be
\label{eqs:pureAdSJacobi}
D^2 \alpha - \frac{d-2}{l^2} \alpha = 0, \qquad D^2 \beta - \frac{d-2}{l^2} \beta = 0.
\ee

Now consider a deformation~$\delta \rho|_{\chi \to \infty}$ to the boundary ball~$\rho|_{\chi \to \infty} = \rho_0$ on a slice of constant~$t$, which we decompose in spherical harmonics as
\be
\label{eq:deltarhobndry}
\delta \rho|_{\chi \to \infty} = \sum_{\ell, m_i} a_{\ell,m_i} Y_{\ell,m_i}(\Omega).
\ee
The corresponding deformation to the extremal surface anchored to this boundary region is governed by~\eqref{eqs:pureAdSJacobi}; since the bulk is static, the RT surface must remain on the same time slice, implying that~$\delta t = \eta^t = 0$ and hence~$\alpha = \beta$.  The deformation to the~$\rho$-embedding of the RT surface is thus given by~$\delta \rho = \eta^\rho = -(\sqrt{2}/l) \alpha \rho \sech \chi$.  Solving~\eqref{eqs:pureAdSJacobi} subject to the boundary condition~\eqref{eq:deltarhobndry}, we thus find the regular solution
\bea
\delta \rho(\chi,\Omega) &= \sum_{\ell, m_i} a_{\ell,m_i} C_\ell \tanh^\ell \chi \, _2 F_1 \left(\frac{\ell-1}{2}, \frac{\ell}{2}, \frac{d}{2} + \ell - 1; \tanh^2\chi\right) Y_{\ell m_i}(\Omega), \\
C_\ell &= -\frac{\cos\left(\frac{\pi d}{2}\right)\Gamma\left(\frac{3-d}{2}\right)\Gamma\left(d+\ell-2\right)}{2^{d+\ell-3} \sqrt{\pi} \, \Gamma\left(\frac{d}{2} + \ell - 1\right)}.
\eea
This expression matches precisely that of~\cite{Mez14} (see also~\cite{AllMez14,NozNum13,Hub12}) under the appropriate change of coordinates~$\tanh\chi \to \sin\theta$ (and under the substitution~$d \to d+1$ necessary due to our differing conventions for~$d$).

\section{Theory of Classical Surface Deformations}
\label{sec:classical}

In this section, we give a unified treatment of classical surface theory~\cite{ColMin,LarFro93,Guv93,VisPar96,Car92,Car92b,Car93,BatCar95,BatCar00,Mos17,GhoMis17,LewPar18}, presented to maximize ease in generalizations to quantum extremal surfaces (of these references, we would highlight~\cite{Car92b} for a very pleasant introduction to the topic of embedded surfaces). To begin, let us review some basic definitions and properties of surfaces with a particular emphasis on formalism that will be useful to later deriving properties of perturbations of extremal surfaces.  Section~\ref{subsec:submanifold} provides an introduction to the topology and geometry of embedded surfaces; readers familiar with these are welcome to skip ahead to Section~\ref{subsec:families}, where surface deformations are discussed.  Our notation and conventions will otherwise follow~\cite{Wald}.

\subsection{Embedded Surfaces}
\label{subsec:submanifold}

We begin with a review of embedded submanifolds.  Intuitively, we think of a submanifold~$\Sigma$ of some geometry~$(M, g_{ab})$ as a subset of~$M$ with some topological and geometric properties inherited from~$(M, g_{ab})$.  However, in order to develop the formalism necessary for our ultimate treatment of quantum extremal surfaces and nonlocal functionals in Section~\ref{sec:quantum}, we will need to exploit the precise definition of an embedded submanifold in order to identify some properties that will be crucial to our later derivation.  To provide a complete story, the purpose of the present section is to first give a pedagogical review of embedded submanifolds from a purely topological perspective (i.e.~without invoking a notion of a metric), followed by a review of their geometry once a metric is invoked.

\subsubsection*{Topology of Embedded Surfaces}

The crucial ingredient in the definition of an embedded submanifold from a topological perspective is the notion of the embedding map and the pullback and pushforward that it defines.  If~$M$ and~$\Sigma$ are manifolds of arbitrary dimensions~$d = \dim(M)$,~$n = \dim(\Sigma)$, with~$n < d$, then the embedding map~$\psi$ is a map~$\psi: \Sigma \to M$ which is injective (that is, no two points in~$\Sigma$ are mapped to the same point in~$M$).  The image~$\psi(\Sigma)$ of~$\Sigma$ in~$M$ defines a submanifold of $M$ which we refer to as a \textit{surface} of codimension~$d-n$ in~$M$; see Figure~\ref{fig:surface}.  The requirement that~$\psi$ be injective is simply the statement that this surface does not self-intersect, and~$\psi$ is thus said to provide an \textit{embedding} of~$\Sigma$ in~$M$\footnote{The removal of the restriction that~$\psi$ be injective instead merely gives an \textit{immersion} of~$\Sigma$ in~$M$.}.  (In terms of coordinate systems on~$\Sigma$ and~$M$, the map~$\psi$ just corresponds to the embedding functions~$X^\mu(y)$ discussed in Section~\ref{subsec:manual}.)

\begin{figure}[t]
\centering
\includegraphics[page=4]{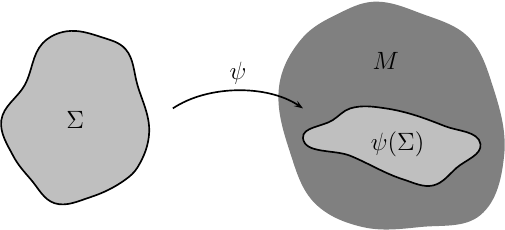}
\caption{A surface in some manifold~$M$ is the image of a lower-dimensional manifold~$\Sigma$ under an embedding map~$\psi$.  The map~$\psi$ can be used to push forward the tangent bundle~$T\Sigma$ to a subset of~$TM$ or to pull back the cotangent bundle~$T_{\psi(\Sigma)}^* M$ to~$T^*\Sigma$.  (Here we depict~$\Sigma$ with a boundary, but whether or not this is the case is immaterial to the discussion.)}
\label{fig:surface}
\end{figure}

For each point~$p \in \Sigma$, the tangent spaces~$T_p \Sigma$ and~$T_{\psi(p)} M$ are related by~$\psi$.  Specifically,~$T_p\Sigma$ can be pushed forward to a subspace of~$T_{\psi(p)} M$ via the pushforward~$\psi^*$, and the image of~$T_p\Sigma$ under this map defines the ``tangent subspace'' to~$\psi(\Sigma)$ at~$\psi(p)$ (see Figure~\ref{fig:surface}):
\be
T^\parallel_{\psi(p)} M \equiv \psi^* T_p\Sigma \subset T_{\psi(p)} M;
\ee
a vector~$v^a \in T^\parallel_{\psi(p)} M$ is thus said to be tangent to~$\psi(\Sigma)$ at~$\psi(p)$.  In turn, this allows us to define the ``normal subspace'' to~$\psi(\Sigma)$ at~$\psi(p)$ as the space of dual vectors normal to all elements of~$T^\parallel_{\psi(p)} M$:
\be 
\left(T^\perp_{\psi(p)}\right)^* M \equiv \{n_a \in T^*_{\psi(p)} M \, | \, n_a v^a = 0 \, \, \forall v^a \in T^\parallel_{\psi(p)} M\};
\ee
a dual vector~$n_a \in (T^\perp_{\psi(p)})^*$ is said to be normal to~$\psi(\Sigma)$ at~$\psi(p)$.  Conversely, the dual vector space~$T_{\psi(p)}^* M$ can be pulled back to~$T^*_p \Sigma$ via the pullback~$\psi_*$.  We may then define assorted tangent bundles in the usual way by taking disjoint unions of the various tangent spaces associated to each~$p$:
\bea
T\Sigma &\equiv \{(p,v^A) \, | \, p \in \Sigma, \, v^A \in T_p \Sigma \}, \\
TM &\equiv \{(p,v^a) \, | \, p \in M, \, v^a \in T_p M \}, \\
T_{\psi(\Sigma)} M &\equiv \{(p,v^a) \, | \, p \in \psi(\Sigma), \, v^a \in T_p M \}, \\
T^\parallel_{\psi(\Sigma)} M &\equiv \{(p,v^a) \, | \, p \in \psi(\Sigma), \, v^a \in T_p^\parallel M \}, \\
\left(T^\perp_{\psi(\Sigma)}\right)^* M &\equiv \{(p,n_a) \, | \, p \in \psi(\Sigma), \, n_a \in \left(T^\perp_p\right)^* M \}, 
\eea
where we use abstract indices~$a, b, \ldots$ for vectors on~$M$ and~$A, B, \ldots$ for vectors on~$\Sigma$.  We should interpret~$T_{\psi(\Sigma)} M$ as the space of all vector fields in~$TM$ ``living on''~$\psi(\Sigma)$, while~$T^\parallel_{\psi(\Sigma)} M$ is the subspace of these which are everywhere tangent to~$\psi(\Sigma)$.  A key point we would like to emphasize here is that the tangent bundle~$T^\parallel_{\psi(\Sigma)} M$ of vector fields tangent to~$\psi(\Sigma)$ and the normal bundle~$(T^\perp_{\psi(\Sigma)})^* M$ of dual vectors normal to~$\Sigma$ are defined \textit{without reference to a metric}; thus objects tangent and normal to~$\psi(\Sigma)$ should be more primitively thought of as upper- or lower-index objects, respectively.

(In terms of the coordinate embeddings~$X^\mu(y)$, the pullback~$\psi_*$ and pushforward~$\psi^*$ are induced by the objects~$\Psi^\mu_\alpha \equiv \partial X^\mu/\partial y^\alpha$, so for example a vector~$v^A \in T\Sigma$ with components~$v^\alpha$ gets pushed forward to a vector~$v^a \in T^\parallel_{\psi(\Sigma)} M$ with components~$v^\mu = \sum_{\alpha = 1}^n \Psi^\mu_\alpha v^\alpha$, while a dual vector~$w_a \in T^*_{\psi(\Sigma)} M$ with components~$w_\mu$ gets pulled back to a dual vector~$w_A \in T^* \Sigma$ with components~$w_\alpha = \sum_{\mu = 1}^d \Psi^\mu_\alpha w_\mu$.)

\subsubsection*{Geometry of Embedded Surfaces}

Now let us suppose that~$M$ comes equipped with a metric~$g_{ab}$ (of arbitrary signature).  The metric uniquely maps vectors to dual vectors and vice versa (i.e.~it lets us raise and lower indices), and it therefore permits the definition of objects like the normal bundle~$T^\perp_{\psi(\Sigma)} M$ as the dual space of~$(T^\perp_{\psi(\Sigma)})^* M$, interpreted as the set of \textit{vector} fields on~$\psi(\Sigma)$ that are normal to it.  A surface is said to have trivial normal bundle if there exists a global orthonormal basis of~$T^\perp_{\psi(\Sigma)} M$ (for intuition, saying a codimension-one surface has trivial normal bundle is equivalent to saying that it is two-sided); when we introduce a basis of the normal bundle in Section~\ref{subsec:codimtwo} we will take the normal bundle to be trivial, but otherwise we do not require this restriction.

The metric on~$M$ gives rise to an induced metric on~$\Sigma$, defined as the pullback of~$g_{ab}$:~$h_{AB} = \psi_* g_{ab}$.  If~$h_{AB}$ is degenerate (that is, if there exists some vector~$k^A \in T\Sigma$ such that~$h_{AB} k^B = 0$), we say that~$\psi(\Sigma)$ is null in~$(M,g_{ab})$; we will not consider this case further unless otherwise stated.  On the other hand, if~$h_{AB}$ is not degenerate, it has an inverse~$h^{AB} \in T \Sigma$ which we may push forward to a tensor~$h^{ab} = \psi^* h^{AB}$.  We may then lower indices as usual using~$g_{ab}$ to define the tensor $h_{ab}$ as well as~${h^a}_b$, which projects from~$T_{\psi(\Sigma)} M$ to~$T^\parallel_{\psi(\Sigma)} M$ (and from~$T^*_{\psi(\Sigma)} M$ to~$(T^\parallel_{\psi(\Sigma)})^* M$).  More generally, any tensor field on~$\Sigma$ can be mapped to a tensor field in~$T^\parallel_{\psi(\Sigma)} M$ by raising all its indices with~$h^{AB}$ and then pushing it forward with~$\psi^*$; we can then lower the indices of this pushforward using~$g_{ab}$.  Consequently, we may work entirely with tensor fields in~$T M$ (and its various subspaces~$T_{\psi(\Sigma)} M$, $T^\parallel_{\psi(\Sigma)} M$, $T^\perp_{\psi(\Sigma)} M$) without reference to the original manifold~$\Sigma$ or its tangent space~$T \Sigma$ at all.  Proceeding in this manner allows us to simplify notation: from here on, we will always refer to the image~$\psi(\Sigma)$ as simply~$\Sigma$.  In particular, the spaces~$T_{\psi(\Sigma)} M$, $T^\parallel_{\psi(\Sigma)} M$, and $T^\perp_{\psi(\Sigma)} M$ will now be called~$T_\Sigma M$, $T^\parallel_\Sigma M$, and $T^\perp_\Sigma M$.  (It is also for this reason that discussions of surfaces often do not introduce the map~$\psi$ at all, and just start by considering some~$\Sigma \subset M$.)

Next, the induced metric~$h_{ab}$ gives rise to a covariant derivative~$^\parallel \! D_a$ on~$\Sigma$ via the usual projection onto~$\Sigma$: for any~$v^a \in T^\parallel_\Sigma M$, we define
\be
\label{eq:Dpardef}
^\parallel \! D_a v^b = {h^c}_a {h^b}_d \grad_c v^d,
\ee
and likewise for higher-rank tensors all of whose indices are tangent to~$\Sigma$, where~$\grad_a$ is the covariant derivative on~$M$ compatible with~$g_{ab}$.  It is straightforward to show that~$^\parallel \! D_a$ inherits this metric-compatibility in the sense that~$^\parallel \! D_a h_{bc} = 0$.  Now note that both the objects~$^\parallel \! D_a$ and~${h^a}_b \grad_a$ are derivative operators on~$\Sigma$, and therefore they must be related by a connection.  Indeed, it is easy to see that for any~$v^a, u^a \in T^\parallel_\Sigma M$,
\be
\label{eq:Dpargrad}
v^b \, ^\parallel \! D_b u^a = v^b \grad_b u^a + v^b u^c {K^a}_{bc},
\ee
where
\be
\label{eq:Kabcdef}
{K^a}_{bc} \equiv - {h_b}^d {h_c}^e \grad_d {h_e}^a
\ee
defines the \textit{extrinsic curvature} of~$\Sigma$ in~$M$ (so called because by~\eqref{eq:Dpargrad}, if~$u^a$ is parallel-transported along~$v^b$ with respect to the intrinsic geometry on~$\Sigma$, then~${K^a}_{bc} u^b v^c$ measures the failure of~$u^a$ to be parallel-transported along~$v^a$ with respect to the ambient geometry~$g_{ab}$, so~${K^a}_{bc}$ quantifies the ``bending'' of~$\Sigma$ in~$M$).  Note that~${P_a}^b {K^c}_{bd} = {P_a}^d {K^c}_{bd} = 0$, with
\be
{P_a}^b \equiv {\delta_a}^b - {h_a}^b
\ee
the projector from~$T_\Sigma M$ to~$T^\perp_\Sigma M$; likewise, from the definition of~$\, ^\parallel \! D_a$ it is clear that~${h_b}^a {K^b}_{cd} = 0$ as well.  Thus~${K^a}_{bc}$ is normal to~$\Sigma$ in its first index and tangent to~$\Sigma$ in its last two indices.  Moreover, Frobenius' theorem says that for any~$v^a$,~$u^a$ tangent to~$\Sigma$, the commutator~$[v,u]^a$ must be tangent to~$\Sigma$ as well.  Expressing the commutator in terms of~$\grad_a$ and using~\eqref{eq:Dpargrad}, we find~$0 = {P^a}_b [v,u]^b = 2u^b v^c {K^a}_{[bc]}$, concluding that~${K^a}_{bc}$ is symmetric in its last two indices.  Finally, it also follows from~\eqref{eq:Kabcdef} that for any dual vector~$n_a$ normal to~$\Sigma$,
\be
\label{eq:Kn}
n_a {K^a}_{bc} = {h_b}^a {h_c}^d \grad_a n_d;
\ee
for this reason, the literature often treats the extrinsic curvature of surfaces by introducing a basis~$\{(n^i)_a\}$,~$i = 1, \ldots, d-n$ of the normal bundle~$(T^\perp_\Sigma)^* M$ and then defining a separate extrinsic curvature~$K^i_{ab} = {h_a}^c {h_b}^d \grad_c (n^i)_d$ for each~$i$.  The definition~\eqref{eq:Kabcdef} is preferable to us, however, as it is manifestly basis-independent.

Besides the covariant derivative~$^\parallel \! D_a$, which acts on tensors tangent to~$\Sigma$ in all their indices, we may also define a covariant derivative on the normal bundle which acts on tensors \textit{normal} to~$\Sigma$ in all their indices: for any~$n^a \in T^\perp_\Sigma M$, we define
\be
^\perp \! D_a n^b = {h^c}_a {P^b}_d \grad_c n^d,
\ee
and likewise for higher-rank tensors.  Again,~$^\perp \! D_a$ and~${h^a}_b \grad_a$ are related by the extrinsic curvature: for any~$v^a \in T^\parallel_\Sigma M$ and~$n^a \in T^\perp_\Sigma M$,
\be
v^b \, ^\perp \! D_b n^a = v^b \grad_b n^a - v^b n^c {K_{cb}}^a.
\ee
More generally, let~${T_{a_1 \cdots a_k}}^{b_1 \cdots b_l}$ be any tensor whose indices are each strictly tangent to or normal to~$\Sigma$.  We define the covariant derivative~$D_a {T_{b_1 \cdots b_k}}^{c_1 \cdots c_l}$ on~$\Sigma$ by first computing~${h_a}^d \grad_d {T_{b_1 \cdots b_k}}^{c_1 \cdots c_l}$ and then projecting each of the ~$b$ and~$c$ indices back to being tangent or normal to~$\Sigma$.  For instance, if~${T_a}^{bc}$ is normal to~$\Sigma$ in its first two indices and tangent to~$\Sigma$ in its last index, we have
\be
D_a {T_b}^{cd} = {h_a}^e {P_b}^f {P^c}_g {h^d}_h \grad_e {T_f}^{gh}.
\ee
Finally, let us note that the fact that~${K^a}_{bc}$ is symmetric in its last two indices implies that~$D_a$ is torsion-free: that is, for any scalar~$f$ on~$\Sigma$,~$D_{[a} D_{b]} f = 0$.

The covariant derivative~$D_a$ can be used to derive the Gauss, Codazzi, and Ricci equations, which relate the intrinsic and extrinsic curvatures of~$\Sigma$ in~$(M,g_{ab})$ to the intrinsic curvature of~$(M,g_{ab})$.  Recall that the Riemann curvature of the latter is defined by
\be
\label{eq:Riemanndef}
2\grad_{[a} \grad_{b]} v_c = {R_{abc}}^d v_d \quad \forall v^a \in TM,
\ee
and that the existence of a tensor~${R_{abc}}^d$ defined in this way is guaranteed by the fact that~$\grad_a$ is torsion-free.  Since~$D_a$ is also torsion-free, there must also exist objects~$^\parallel \! R_{abcd}$ and~$^\perp \! R_{abcd}$, interpreted as the curvatures of~$(\Sigma,h_{ab})$ and of the normal bundle, respectively, that obey
\begin{subequations}
\label{eqs:Rparperp}
\begin{align}
2 D_{[a} D_{b]} v_c &= \, ^\parallel \! {R_{abc}}^d v_d \quad \forall v^a \in T^\parallel_\Sigma M, \\
2 D_{[a} D_{b]} n_c &= \, ^\perp \! {R_{abc}}^d n_d \quad \forall n^a \in T^\perp_\Sigma M.
\end{align}
\end{subequations}
Expressing~$D_a$ in terms of~$\grad_a$ in~\eqref{eqs:Rparperp} and using~\eqref{eq:Kabcdef} and~\eqref{eq:Riemanndef} to rearrange, we obtain the Gauss and Ricci equations\footnote{Equation~\eqref{eq:Ricci} is often not treated in the physics literature because it is trivial unless both the dimension and codimension of~$\Sigma$ are greater than one.  In most of the math literature it goes by the name Ricci equation, a convention that we will follow here, though see~\cite{Car92b} for an argument that it should go by some linear combination of the names Voss, Ricci, Walker, and Schouten.}
\begin{subequations}
\begin{align}
^\parallel \! R_{abcd} &= {h_a}^e {h_b}^f {h_c}^g {h_d}^h R_{efgh} + 2K_{ec[a} {K^e}_{b]d}, \label{eq:Gauss} \\
^\perp \! R_{abcd} &= {h_a}^e {h_b}^f {P_c}^g {P_d}^h R_{efgh} + 2K_{ce[a|} {K_{d|b]}}^e. \label{eq:Ricci}
\end{align}
Likewise, using the definition of~$D_a$ in terms of~$\grad_a$ along with~\eqref{eq:Riemanndef}, it is straightforward to directly obtain the Codazzi equation
\be
\label{eq:Codazzi}
D_a K_{dbc} - D_b K_{dac} = {h_a}^e {h_b}^f {h_c}^g {P_d}^h R_{efgh}.
\ee
\end{subequations}

It will be useful to note that~$^\perp \! R_{abcd}$ measures the path-dependence of parallel transport along~$\Sigma$ of vectors in the normal bundle, and thus the vanishing of~$^\perp \! R_{abcd}$ is therefore a necessary and sufficient condition for the existence of a basis~$\{(n^i)_a\}$ of the normal bundle which satisfies~$D_a (n^i)_b = 0$; such a frame is called a Fermi-Walker frame.  Since~$^\perp \! R_{abcd}$ is antisymmetric and tangent to~$\Sigma$ in its first two indices and antisymmetric and normal to~$\Sigma$ in its last two indices, it vanishes identically when~$\Sigma$ is a curve or a hypersurface (i.e.~a codimension-one surface), in which case a Fermi-Walker frame always exists.  More generally, it is possible to show that~\eqref{eq:Ricci} can be rewritten as
\be
\label{eq:conformalRicci}
^\perp \! R_{abcd} = {h_a}^e {h_b}^f {P_c}^g {P_d}^h W_{efgh} + 2\widetilde{K}_{ce[a|} \widetilde{K}^e_{\phantom{e}d|b]},
\ee
where~$W_{abcd}$ is the Weyl tensor of~$(M, g_{ab})$ and~$\widetilde{K}^a_{\phantom{a}bc}$ is the trace-free extrinsic curvature
\be
\widetilde{K}^a_{\phantom{a}bc} \equiv {K^a}_{bc} - \frac{1}{n} \, K^a h_{bc}.
\ee
It then follows that a sufficient condition to ensure the vanishing of~$^\perp \! R_{abcd}$ is that~$g_{ab}$ be conformally flat and that~$\widetilde{K}^a_{\phantom{a}bc} = 0$.

\subsection{Families of Surfaces}
\label{subsec:families}

We are now equipped to study perturbations of surfaces.  As reviewed in Section~\ref{subsec:manual}, such perturbations are often treated by introducing an explicit coordinate embedding~$X^\mu(y)$ and then considering a perturbation~$\delta X^\mu(y)$ to the embedding functions.  While this brute-force approach does work for deriving the equation of extremal deviation, we instead exploit an equivalent but more abstract formulation which, as we will see, makes it almost effortless to derive the relevant perturbation equations we desire, and also facilitates the inclusion of quantum corrections.  To introduce this formalism, we must specify precisely what we mean by a ``perturbation'' of a given surface~$\Sigma$.  Here we take the following approach: we consider a continuous one-parameter family of surfaces~$\Sigma(\lambda)$ with~$\Sigma(\lambda = 0) = \Sigma$; studying ``perturbations'' of~$\Sigma$ then corresponds to studying the behavior of this family of surfaces around~$\lambda = 0$.  To perform such an analysis, note that at each value of~$\lambda$, the surface~$\Sigma(\lambda)$ can be obtained from the original surface~$\Sigma$ by some diffeomorphism~$\phi_\lambda$, that is,~$\Sigma(\lambda) = \phi_\lambda (\Sigma)$.  Thus the one-parameter continuous family~$\Sigma(\lambda)$ is obtained by ``evolving''~$\Sigma$ along a one-parameter group of diffeomorphisms~$\phi_\lambda$.  The generator of~$\phi_\lambda$ will be denoted by~$\eta^a$, and its restruction to~$\Sigma$ is said to be a deviation vector field on~$\Sigma$ along the family~$\Sigma(\lambda)$, as shown in Figure~\ref{fig:surfacefamily}.  This deviation vector field encodes ``infinitesimal deformations'' of~$\Sigma$ (which can be seen explicitly from the fact that for any coordinate system~$\{x^\mu\}$ on~$M$,~$\phi_\lambda$ generates a coordinate transformation~$x^\mu \to x^\mu + \lambda \eta^\mu + \Ocal(\lambda^2)$, and hence the embedding functions of~$\Sigma$ are modified as in~\eqref{eq:embeddingvariation}).

\begin{figure}[t]
\centering
\includegraphics[page=5]{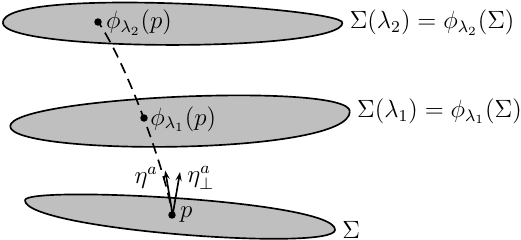}
\caption{A one-parameter family of surfaces~$\Sigma(\lambda)$ can be obtained from a starting surface~$\Sigma$ by evolving it along a one-parameter family of diffeomorphisms~$\phi_\lambda$.  For each point~$p \in \Sigma$,~$\phi_\lambda(p)$ is a curve parametrized by~$\lambda$ (shown as a dashed line); the tangent to such curves at all points on~$\Sigma$ is a deviation vector field~$\eta^a$ along this family of surfaces.  There is some freedom in the choice of~$\eta^a$ on~$\Sigma$ (corresponding to the freedom in how each point~$p \in \Sigma$ is mapped to subsequent surfaces), but the evolution of the geometry of the family~$\Sigma(\lambda)$ is captured by the component~$\eta^a_\perp$ normal to~$\Sigma$.}
\label{fig:surfacefamily}
\end{figure}

Note that for each~$\lambda$,~$\phi_\lambda$ is only defined by the requirement that it map~$\Sigma$ to~$\Sigma(\lambda)$, and therefore the group of diffeomorphisms~$\phi_\lambda$ is highly non-unique.  One source of this non-uniqueness is that the action of~$\phi_\lambda$ on any points not lying on~$\Sigma$ is completely arbitrary, and thus the generator~$\eta^a$ is completely arbitrary off of~$\Sigma$.  More importantly, if~$\varphi$ is any diffeomorphism that maps~$\Sigma$ to itself, then the composed diffeomorphism~$\phi_\lambda \circ \varphi$ also maps~$\Sigma$ to~$\Sigma(\lambda)$ (this is just the observation that the family~$\Sigma(\lambda)$ is unchanged if points within each~$\Sigma(\lambda)$ are ``moved around''); this implies that even on~$\Sigma$, the component of~$\eta^a$ tangent to~$\Sigma$ is arbitrary\footnote{Here we are temporarily ignoring boundary terms because they are irrelevant to the present discussion, but as shown in Appendix~\ref{app:variations}, the tangential component~$\eta^a_\parallel = {h^a}_b \eta^b$ does play a role in the evolution of the boundary~$\partial \Sigma$.}.  We therefore conclude that geometric information about the family~$\Sigma(\lambda)$ near~$\Sigma$ must be captured by the normal component~$\eta^a_\perp = {P^a}_b \eta^b$.

To proceed further, let us recall that the \textit{position} of a surface~$\Sigma$ in a geometry~$(M, g_{ab})$ is ``gauge-dependent'' in the sense that for any diffeomorphism~$\phi$, the surface~$\phi(\Sigma)$ in the geometry~$(M, \phi^* g_{ab})$ is geometrically equivalent to the original surface~$\Sigma$ in the original geometry~$(M, g_{ab})$.  What makes the one-parameter family~$\Sigma(\lambda)$ we have just introduced nontrivial is that the diffeomorphisms~$\phi_\lambda$ act only on~$\Sigma$, and \textit{not} on the ambient geometry~$g_{ab}$.  In other words, the nontrivial evolution of the~$\Sigma(\lambda)$ is due to a \textit{relative} diffeomorphism between~$\Sigma(\lambda)$ and the ambient metric~$g_{ab}$.  This observation leads to a natural alternative formulation of the evolution of surfaces: rather than considering a family of surfaces~$\Sigma(\lambda)$ evolving through a fixed metric~$g_{ab}$, as in Figure~\ref{subfig:active}, we may instead fix the surface~$\Sigma$ and evolve the metric~$g_{ab}$ ``back'' to~$\Sigma$, as in Figure~\ref{subfig:passive}.  We will call the former formulation (in which the~$\Sigma(\lambda)$ are evolving) the ``active'' picture, and we will refer to the latter formulation (in which~$\Sigma$ is left fixed but the metric is evolved) the ``passive'' picture\footnote{Our active and passive pictures are essentially the Eulerian and Lagrangian schemes used in~\cite{Car93,BatCar95,BatCar00}.}.  The metric in the active picture will be denoted by~$g^\mathrm{act}_{ab}$, while the metric in the passive picture is~$g^\mathrm{pas}_{ab} = (\phi^{-1}_{\lambda})_* g^\mathrm{act}_{ab} = \phi^*_{-\lambda} g^\mathrm{act}_{ab}$ (the latter equality follows from the fact that~$\phi_\lambda$ is a one-parameter group of diffeomorphisms and therefore~$\phi^{-1}_\lambda = \phi_{-\lambda}$, since~$\phi_{\lambda_1} \circ \phi_{\lambda_2} = \phi_{\lambda_1 + \lambda_2}$ with~$\phi_0$ the identity).  As we will see, switching to the passive picture offers two substantial advantages: first, in the passive picture we may exploit the fact that the tangent vector bundle~$T^\parallel_\Sigma M$ and normal dual vector bundle~$(T^\perp_\Sigma)^* M$ of the \textit{fixed} surface~$\Sigma$ are independent of the ambient metric, and therefore are unaffected by its evolution; second, all~$\lambda$-dependence is contained in the passive metric~$g^\mathrm{pas}_{ab}(\lambda)$, and therefore it is quite easy to investigate the behavior of families of surfaces even when the \textit{active} metric is varying arbitrarily (this is relevant, for instance, to a situation in which the boundary conditions of an extremal surface and its ambient geometry are deformed simultaneously).

\begin{figure}[t]
\centering
\subfigure[]{
\includegraphics[page=6,width=0.45\textwidth]{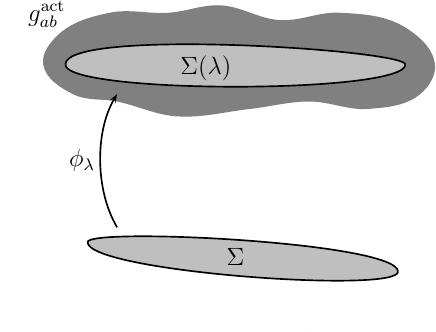}
\label{subfig:active}
}%
\hspace{1cm}%
\subfigure[]{
\includegraphics[page=7,width=0.45\textwidth]{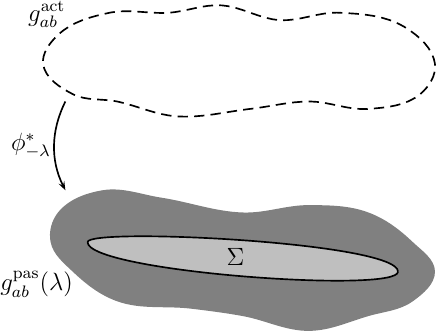}
\label{subfig:passive}
}
\caption{\subref{subfig:active}: in the active picture, the initial surface~$\Sigma$ is evolved to a one-parameter family of surfaces~$\Sigma(\lambda)$ by a one-parameter group of diffeomorphisms~$\phi_\lambda$.  Each~$\Sigma(\lambda)$ is sensitive only to the ambient geometry~$g^\mathrm{act}_{ab}$ in its neighborhood, illustrated schematically as the dark gray shading.  \subref{subfig:passive}: equivalently, in the passive picture~$\Sigma$ is left unchanged, but the ambient geometry is pulled back to~$\Sigma$ by~$(\phi^{-1}_\lambda)_* = \phi_{-\lambda}^*$.  $\Sigma$ is then sensitive to a one-parameter family of passive metrics~$g^\mathrm{pas}_{ab}(\lambda) = \phi^*_{-\lambda} g^\mathrm{act}_{ab}$.}
\label{fig:pictures}
\end{figure}

\subsubsection*{Warmup: Geodesics}

To highlight the advantages of switching to the passive picture, let us warm up by re-deriving the equation of geodesic deviation (recall that this equation can be thought of as governing the infinitesimal perturbation to a geodesic under a deformation of its boundary conditions).  In fact, we will be more general: we will derive the \textit{sourced} equation of geodesic deviation, which describes how a geodesic varies in response to a perturbation of the active metric (as well as of its boundary conditions).  In other words, consider a geodesic~$\Sigma$ in a geometry with metric~$g_{ab}$, and let us begin in the active picture by perturbing both~$\Sigma$ and the metric to one-parameter families~$\Sigma(\lambda)$ and~$g^\mathrm{act}_{ab}(\lambda)$ such that for each~$\lambda$,~$\Sigma(\lambda)$ is a geodesic with respect to the geometry~$g^\mathrm{act}_{ab}(\lambda)$.  Switching to the passive picture, we keep~$\Sigma$ fixed and only vary the passive metric~$g^\mathrm{pas}_{ab}(\lambda) = \phi^*_{-\lambda} g^\mathrm{act}_{ab}(\lambda)$.

Let~$u^a$ be an affinely-parametrized tangent to~$\Sigma$ with respect to~$g_{ab}$; i.e.~$u^b \grad_b u^a = 0$.  Because the tangent space of~$\Sigma$ is independent of the metric,~$u^a$ is always tangent to~$\Sigma$ for any~$\lambda$.  For general~$g^\mathrm{pas}_{ab}(\lambda)$,~$u^a$ will not necessarily remain an affinely-parametrized tangent, but we can always gauge-fix by performing a~$\lambda$-dependent diffeomorphism within~$\Sigma$ to ensure that it does (such a diffeomorphism essentially corresponds to ``reparametrizing'' the curve).  Taking therefore~$u^a$ to be an affinely-parametrized tangent to~$\Sigma$ for all~$\lambda$, the requirement that~$\Sigma$ remain a geodesic under the perturbation is simply the geodesic equation 
\be
\label{eq:geodesic}
u^b \grad_b^{(\lambda)} u^a = 0,
\ee
where~$\grad_a^{(\lambda)}$ is the covariant derivative compatible with~$g^\mathrm{pas}_{ab}(\lambda)$.  Since both~$\grad_a$ and~$\grad_a^{(\lambda)}$ are derivative operators, they are again related by a connection:
\begin{subequations}
\label{eqs:gradlambda}
\be
\grad_b^{(\lambda)} u^a = \grad_b u^a + {C^a}_{bc}(\lambda) u^c,
\ee
where (see e.g.~Section~7.5 of~\cite{Wald})
\be
\label{subeq:Cdef}
{C^a}_{bc}(\lambda) = \frac{1}{2} (g^\mathrm{pas})^{ad}(\lambda) \left[ \grad_b g^\mathrm{pas}_{cd}(\lambda) + \grad_c g^\mathrm{pas}_{bd}(\lambda) - \grad_d g^\mathrm{pas}_{bc}(\lambda) \right].
\ee
\end{subequations}
Since~$u^a$ is affinely-parametrized with respect to~$g_{ab}$, the geodesic equation~\eqref{eq:geodesic} becomes~${C^a}_{bc}(\lambda) u^b u^c = 0 $, and in particular the derivative of this equation at~$\lambda = 0$ yields~$\dot{C}^a_{\phantom{a}bc} u^b u^c = 0$, where the dot denotes a~$\lambda$-derivative at~$\lambda = 0$.  Note that from~\eqref{subeq:Cdef},
\be
\label{eq:Cdot}
\dot{C}^a_{\phantom{a}bc} = \frac{1}{2} g^{ad} \left(\grad_b \dot{g}^\mathrm{pas}_{cd} + \grad_c \dot{g}^\mathrm{pas}_{bd} - \grad_d \dot{g}^\mathrm{pas}_{bc}\right).
\ee

Now let us return to the active picture: since~$g^\mathrm{pas}_{ab}(\lambda) = \phi^*_{-\lambda} g^\mathrm{act}_{ab}(\lambda)$, we have that on~$\Sigma$,
\be
\label{eq:gpassivedot}
\dot{g}^\mathrm{pas}_{ab} = \lim_{\lambda \to 0} \frac{\phi^*_{-\lambda} g^\mathrm{act}_{ab}(\lambda) - g^\mathrm{act}_{ab}(0)}{\lambda} = \pounds_\eta g_{ab} + \delta g_{ab} = 2\grad_{(a} \eta_{b)} + \delta g_{ab},
\ee
where~$\delta g_{ab} \equiv dg^\mathrm{act}_{ab}/d\lambda|_{\lambda = 0}$ is the linear perturbation to the active metric.  Inserting this expression for~$\dot{g}^\mathrm{pas}_{ab}$ into the equation~$\dot{C}^a_{\phantom{a}bc} u^b u^c = 0$, using the fact that~$u^b \grad_b u^a = 0$, and using the definition~\eqref{eq:Riemanndef} of the Riemann tensor, we quickly obtain
\be
\label{eq:sourcedJacobi}
u^b \grad_b \left(u^c \grad_c \eta^a\right) + {R_{cbd}}^a u^c u^d \eta^b = - {\delta \Gamma^a}_{bc} u^b u^c,
\ee
where
\be
\label{eq:deltagamma}
{\delta \Gamma^a}_{bc} \equiv \frac{1}{2} g^{ad} \left( \grad_b \delta g_{cd} + \grad_c \delta g_{bd} - \grad_d \delta g_{bc}\right)
\ee
is the variation in the connection due to the variation in the active metric.  When~$\delta g_{ab} = 0$, we immediately recognize~\eqref{eq:sourcedJacobi} as just the usual equation of geodesic deviation.  More generally,~\eqref{eq:sourcedJacobi} is a \textit{sourced} equation of geodesic equation, describing how a geodesic ``moves'' in response to simultaneous perturbations of the spacetime and of its boundary conditions.

Before moving on to surfaces of general dimension, let us make two remarks.  First, geodesics are sufficiently simple that here we never needed to assume that~$\Sigma$ was nondegenerate; thus~\eqref{eq:sourcedJacobi} is valid for geodesics of any signature, including null.  Indeed, the sourced equation of geodesic deviation was derived via more brute-force calculations in~\cite{PynBir93} with the goal of applying it to null geodesics.  Second,~\eqref{eq:sourcedJacobi} constrains \textit{all} of the components of~$\eta^a$ on~$\Sigma$, even though only the components of~$\eta^a$ which are transverse (that is, not tangent) to~$\Sigma$ affect how $\Sigma(\lambda)$ changes \textit{as a surface} with $\lambda$. The extra constraint on the component of~$\eta^a$ tangent to~$\Sigma$ comes from our earlier gauge-fixing requirement enforcing that~$u^a$ be an affinely-parametrized tangent to~$\Sigma$ for all~$\lambda$, since it is the component of~$\eta^a$ tangent to~$\Sigma$ which encodes how~$\Sigma$ is to be reparametrized to ensure that~$u^a$ remain an affinely-parametrized tangent.  If~$\Sigma$ is not null, we may isolate the transverse parts of~\eqref{eq:sourcedJacobi} by projecting onto the normal bundle~$T^\perp_\Sigma M$ with~${P^a}_b = {\delta^a}_b - u^a u_b/u^2$, obtaining
\be
\label{eq:sourcedJacobinormal}
u^b \grad_b \left(u^c \grad_c \eta^a_\perp \right) + {R_{cbd}}^a u^c u^d \eta_\perp^b = - {P^a}_d {\delta \Gamma^d}_{bc} u^b u^c.
\ee
The ``gauge'' part of~\eqref{eq:sourcedJacobi}, which keeps~$u^a$ affinely parametrized, is just the contraction with~$u^a$, which yields
\be
u^a \grad_a \left(u^b \grad_b (u \cdot \eta)\right) = - {\delta \Gamma^a}_{bc} u_a u^b u^c;
\ee
for a given metric perturbation~$\delta g_{ab}$, this equation can easily be integrated twice to obtain the component~$u \cdot \eta$.

\subsubsection*{Surfaces of General Dimension}

Now that we have demonstrated how to derive the sourced equation of geodesic deviation in the passive picture, it is relatively straightforward to generalize to surfaces of arbitrary dimension by working with the induced metric~$h_{ab}$ rather than the tangent vector~$u^a$.  We begin as we did for geodesics: consider a surface~$\Sigma$ (of arbitrary dimension) in a geometry with metric~$g_{ab}$, and perturb both the surface and the metric to the one-parameter families~$\Sigma(\lambda)$,~$g^\mathrm{act}_{ab}(\lambda)$.  Again we switch to the passive picture, considering instead a fixed surface~$\Sigma$ in the one-parameter family of ambient metrics~$g^\mathrm{pas}_{ab}(\lambda) = \phi^*_{-\lambda} g^\mathrm{act}_{ab}(\lambda)$.

Now we note that the induced metric on~$\Sigma$ varies as
\be
\dot{h}^{ab} = \frac{d}{d\lambda}\left({h^a}_c {h^b}_d (g^\mathrm{pas})^{cd}\right) = {h^a}_c {h^b}_d (\dot{g}^\mathrm{pas})^{cd} + \dot{h}^a_{\phantom{a}c} h^{cb} + \dot{h}^b_{\phantom{b}c} h^{ca}.
\ee
Recall, however, that~${h^a}_b$ is the identity on the tangent space~$T^\parallel_\Sigma M$, which in this picture is~$\lambda$-independent; thus~${h^a}_b(\lambda) v^b = v^a$ for all~$\lambda$ and for any~$v^a \in T^\parallel_\Sigma M$.  It then follows that~$\dot{h}^a_{\phantom{a}b}v^b = 0$, implying~$\dot{h}^a_{\phantom{a}c} h^{cb} = 0$ and thus
\begin{subequations}
\be
\dot{h}^{ab} = {h^a}_c {h^b}_d (\dot{g}^\mathrm{pas})^{cd} = -h^{ac} h^{bd} \dot{g}^\mathrm{pas}_{cd},
\ee
where in the second equality we used the fact that~$(\dot{g}^\mathrm{pas})^{ab} = -g^{ac} g^{bd} \dot{g}^\mathrm{pas}_{cd}$.  We then straightforwardly obtain
\begin{align}
\dot{h}_a^{\phantom{a}b} &= \frac{d}{d\lambda}(h^{bc} g^\mathrm{pas}_{ac}) = {P_a}^c h^{bd} \, \dot{g}^\mathrm{pas}_{cd}, \\
\dot{h}_{ab} &= \frac{d}{d\lambda}({h_a}^c g^\mathrm{pas}_{cb}) = \left({\delta_a}^c {\delta_b}^d - {P_a}^c {P_b}^d \right) \dot{g}^\mathrm{pas}_{cd}.
\end{align}
\end{subequations}
These equations, in addition to the variation~\eqref{eq:Cdot} of the covariant derivative, are sufficient to compute the variation~$\dot{K}^a_{\phantom{a}bc}$ of the extrinsic curvature of~$\Sigma$.  In fact, we will ultimately only be interested in the mean curvature~$K^a \equiv h^{bc} {K^a}_{bc}$, whose variation can be simplified to
\be
\label{eq:Kdotpassive}
\dot{K}_a = P_{ab}\left[(K^c g^{bd} - K^{bcd}) \dot{g}^\mathrm{pas}_{cd} - h^{cd} \dot{C}^b_{\phantom{b}cd}\right].
\ee
The fact that~$\dot{K}_a \in (T^\perp_\Sigma)^* M$ is as expected, since the normal bundle~$(T^\perp_\Sigma)^* M$ is~$\lambda$-independent.

Now we may use~\eqref{eq:gpassivedot} to express~$\dot{K}_a$ in terms of~$\eta^a$ and the perturbation~$\delta g_{ab}$ to the active metric: using~\eqref{eq:Cdot}, decomposing~$\eta^a$ into parts normal and tangent to~$\Sigma$, and making use of the Codazzi equation~\eqref{eq:Codazzi}, we obtain
\begin{subequations}
\label{eq:Kdotgeneral}
\be
\dot{K}_a = J (\eta_\perp)_a - s_a + K_c {P_a}^b \grad_b \eta^c_\perp + \pounds_{\eta_\parallel} K_a + {P_a}^b K^c \delta g_{bc},
\ee
where~$J$ is a second-order differential operator on the normal bundle, given explicitly as
\be
J (\eta_\perp)_a \equiv -D^2 (\eta_\perp)_a - Q_{ab} \eta^b_\perp \label{subeq:Ldef}
\ee
with
\begin{align}
D^2 \eta_\perp^a &= h^{bc} D_b D_c \eta_\perp^a = h^{bc} {P^a}_d \grad_b( {h_c}^e {P^d}_f \grad_e \eta_\perp^f), \label{eq:Laplaciannormalbundle} \\
Q_{ab} &\equiv S_{ab} + h^{cd} {P_a}^e {P_b}^f R_{cedf}, \label{subeq:Qdef} \\
S_{ab} &\equiv K_{acd} {K_b}^{cd},
\end{align}
and the source term~$s_a(\delta g)$ is given by
\be
\label{subeq:sdef}
s_a(\delta g) \equiv {K_a}^{bc} \delta g_{bc} + h^{cd} P_{ab} \delta {\Gamma^b}_{cd}.
\ee
\end{subequations}
The tensor~$S_{ab}$, which is symmetric and normal to~$\Sigma$, is often called Simons' operator, while~$D^2 (\eta_\perp)_a$ is often called the Laplacian of~$(\eta_\perp)_a$ on the normal bundle.  Let us also note that in the context of minimal surfaces in Riemannian manifolds,~$J$ is called the stability operator of~$\Sigma$~\cite{ColMin}; the reason for this nomenclature will become clear in Section~\ref{sec:stability}.  In other contexts,~$J$ is sometimes referred to as the Jacobi operator.

Equation~\eqref{eq:Kdotgeneral} governs perturbations of surfaces in broad generality, and we will make use of several special cases of it in the rest of this paper.  It is therefore worth pausing here to make some remarks and to highlight these special cases.  First, note that as expected, the component~$\eta^a_\parallel$ of~$\eta^a$ tangent to~$\Sigma$ simply transforms~$K_a$ by a diffeomorphism within~$\Sigma$; all the geometric information about the ``flow'' of surfaces is contained within the normal component~$\eta_\perp^a$.  Second, recall that~$\eta^a$ is arbitrary off of~$\Sigma$, and therefore the normal derivative~${P_a}^b \grad_b \eta_\perp^c$ is as well.  The appearance of this term as well as of the arbitrary component~$\eta_\parallel^a$ is an artifact of the fact that~\eqref{eq:Kdotgeneral} is a ``mixed-picture'' expression:~$\dot{K}_a$ should be understood as a passive-picture object, and is perfectly well-defined in terms of~$\dot{g}^\mathrm{pas}_{ab}$ via~\eqref{eq:Kdotpassive}; on the other hand, the objects on the right-hand side of~\eqref{eq:Kdotgeneral} are active-picture quantities.  To convert entirely to the active picture, note that the extrinsic curvature~$K_a$ is defined on each surface~$\Sigma(\lambda)$, and therefore can be thought of as a field on the~$(n+1)$-dimensional surface~$\Xi$ swept out by~$\Sigma(\lambda)$ as~$\lambda$ is varied\footnote{For simplicity, here we assume that~$\eta^a$ is nowhere vanishing, so that for sufficiently small range of~$\lambda$, none of the~$\Sigma(\lambda)$ intersect any of the others.  The case where they do intersect can be treated as a limiting case, and all the expressions are unaffected.}, as shown in Figure~\ref{fig:Kfield}.  The derivative~$\dot{K}_a$ can then be interpreted as a Lie derivative, so we may write
\be
\dot{K}_a = {P_a}^b \dot{K}_b = {P_a}^b \pounds_\eta K_b = {P_a}^b \left(\eta^c_\perp \grad_c K_b + K_c \grad_b \eta^c_\perp\right) + \pounds_{\eta_\parallel} K_b,
\ee
where we have used the fact that~$\dot{K}_a$ is normal to~$\Sigma$, and thus we can freely move projectors~${P_a}^b$ into or out of the Lie derivatives.  Inserting this expression into~\eqref{eq:Kdotgeneral}, both~$\eta_\parallel^a$ and~${P_a}^b \grad_b \eta^c_\perp$ drop out, and we are left with
\be
\label{eq:etagradK}
{P_a}^b \eta^c_\perp \grad_c K_b = J(\eta_\perp)_a - s_a + {P_a}^b K^c \delta g_{bc}.
\ee
The left-hand side is just a derivative of~$K_a$ along~$\eta^a_\perp$ (which is now well-defined), and the right-hand side depends only on~$\eta_\perp^a$ on~$\Sigma$.

\begin{figure}[t]
\centering
\includegraphics[page=8]{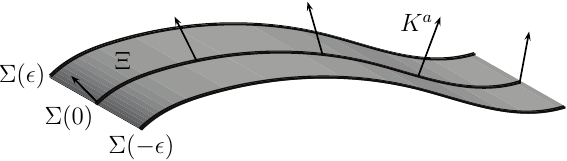}
\caption{Given a one-parameter family of~$n$-dimensional surfaces~$\Sigma(\lambda)$, we may define the mean curvature~$K_a$ on each surface, which can therefore be thought of as a field near~$\Sigma = \Sigma(0)$ on the~$(n+1)$-dimensional surface~$\Xi$ (shaded in gray) swept out by~$\Sigma(\lambda)$ as~$\lambda$ is varied around zero (wherever these surfaces do not intersect one another for sufficiently small range of~$\lambda$).  This picture allows us to define the directional derivative of~$K^a$ along~$\eta^a$, in addition to just the directions tangent to~$\Sigma$. }
\label{fig:Kfield}
\end{figure}

As a final observation, recall that the condition that a surface be extremal is simply~$K_a = 0$.  If we require that the surface~$\Sigma(\lambda)$ be extremal with respect to~$g^\mathrm{act}_{ab}(\lambda)$ for each~$\lambda$, then we have~$\dot{K}_a = 0$, which from~\eqref{eq:Kdotgeneral} imposes a constraint on the deviation vector along such a family of extremal surfaces:
\be
\label{eq:sourcedgeneralJacobi}
J \eta_\perp^a = s^a.
\ee
This is the analogue of the sourced equation of geodesic deviation~\eqref{eq:sourcedJacobinormal} for higher-dimensional extremal surfaces.  Indeed, when~$\Sigma$ is a (non-null) geodesic with tangent~$u^a$, we have~$h^{ab} = u^a u^b/u^2$ and~${K^a}_{bc} = 0$, and~\eqref{eq:sourcedgeneralJacobi} reduces to~\eqref{eq:sourcedJacobinormal}.  Finally, if we require the family of surfaces~$\Sigma(\lambda)$ to all be extremal in the \textit{same} ambient geometry~$g_{ab}$, then~$\delta g_{ab} = 0$, and we find that the deviation vector obeys
\be
\label{eq:unsourcedJacobi}
J \eta_\perp^a = 0.
\ee
In the mathematics literature, the equations~\eqref{eq:sourcedgeneralJacobi} and~\eqref{eq:unsourcedJacobi} are often referred to as the (inhomogeneous and homogeneous) Jacobi equation.  Here we will instead give them the more descriptive name of the sourced and unsourced equations of extremal deviation.

\subsection{Codimension-Two Spacelike Surfaces}
\label{subsec:codimtwo}

So far we have focused on surfaces of general codimension and signature.  Since our ultimate goal is to apply this formalism to subregion/subregion duality, let us now explicitly restrict to the case of spacelike codimension two surfaces in Lorentzian geometries.  Such surfaces have two independent null normal vectors~$k^a$ and~$\ell^a$ with corresponding null expansions
\be
\theta^{(k)} = K_a k^a, \qquad \theta^{(\ell)} = K_a \ell^a,
\ee
so when a spacelike codimension-two \textit{extremal} surface~$\Sigma$ is perturbed by a deviation vector~$\eta^a$, we may interpret~\eqref{eq:Kdotgeneral} (or~\eqref{eq:etagradK}) as computing the perturbation to its expansions:
\be
\label{eq:thetadot}
\dot{\theta}^{(k)} = \dot{K}_a k^a = k^a J (\eta_\perp)_a - s_k, \qquad \dot{\theta}^{(\ell)} = \dot{K}_a \ell^a = \ell^a J (\eta_\perp)_a - s_\ell,
\ee
where we use the notation~$s_k \equiv s \cdot k$,~$s_\ell \equiv s \cdot \ell$.  Our purpose is now to decompose the \textit{scalar} objects~$k^a J (\eta_\perp)_a$ and~$\ell^a J(\eta_\perp)_a$ in a useful choice of basis of the normal bundle of~$\Sigma$.  Such a decomposition has the advantage that the resulting equations are more immediately amenable to a treatment using elliptic operator theory, which is typically formulated in terms of (systems of) scalar elliptic differential equations.

First, note that given a basis~$\{(n_i)^a\}$ of the normal bundle of any (non-null) surface~$\Sigma$ of \textit{arbitrary} codimension, we may define the scalar differential operators~$J_{i,j}$ (the ``components'' of~$J$) via
\be
\label{eq:Lij}
J_{i,j} f \equiv (n_i)^a J \left(f (n_j)_a\right)
\ee
for any scalar~$f$ on~$\Sigma$.  To evaluate these objects, it is convenient to decompose the covariant derivative~$D_a$ on the normal bundle as
\be
\label{eq:normaldecomp}
(n_i)^b D_a u_b = D_a u_i - \sum_{j = 1}^{d-n} {\omega_{a i}}^j u_j \mbox{ with } {\omega_{a i}}^j = (n^j)_b D_a (n_i)^b,
\ee
where~$u^a$ is any vector field normal to~$\Sigma$,~$u_i \equiv u \cdot n_i$ are its components in this basis, the~${\omega_{ai}}^j$ are connection one-forms, and the normal index on~$(n^j)_b$ is raised using~$P^{ij}$, the matrix inverse of the metric on the normal bundle~$P_{ij} \equiv n_i \cdot n_j$.  Restricting now to the case of codimension two, let us take the basis~$\{(n_i)^a\}$ to consist of the null vector field~$k^a$ (which we take to be future-pointing) and another arbitrarily specified vector field~$m^a$ normalized such that~$k \cdot m = 1$.  Then the connection one-forms in this basis can be straightforwardly computed:
\be
\label{eq:connections}
{\omega_{a k}}^k = - {\omega_{a m}}^m = m^b D_a k_b \equiv \chi_a, \qquad {\omega_{ak}}^m = 0, \qquad {\omega_{a m}}^k = (m^b - m^2 k^b) D_a m_b.
\ee
In this context, the object~$\chi_a$ is often called the twist potnetial (not to be confused with the twist~$\omega_{ab}$, also called the vorticity, of a geodesic congruence).  Note that~$\chi_a$ is independent of the choice of~$m^a$, since it is unchanged under the transformation~$m^a \to m^a + f k^a$ for any scalar~$f$; however,~$\chi_a$ still depends on the normalization of~$k^a$, as the transformation~$k^a \to e^f k^a$ sends~$\chi_a \to \chi_a + D_a f$.  In general there is no choice of normalization that sets~$\chi_a = 0$, but note that per the discussion around~\eqref{eq:conformalRicci}, it follows that a sufficient condition for the existence of such a normalization is that~$g_{ab}$ be conformally flat and~$\widetilde{K}^a_{\phantom{a}bc} = 0$.

Second, note that in this basis the induced metric on~$\Sigma$ can be written as
\be
h_{ab} = g_{ab} + m^2 k_a k_b - 2k_{(a} m_{b)},
\ee
and hence, using the definition~\eqref{subeq:Qdef} of~$Q_{ab}$ and the symmetries of the Riemann tensor, we have
\be
\label{eq:Qkk}
Q_{kk} = S_{kk} + R_{kk}, \qquad Q_{km} = S_{km} + R_{km} + R_{kmkm}.
\ee
Restricting now to the case where~$\Sigma$ is extremal, we note that from the definition of Simons' tensor and of the extrinsic curvature that~$S_{kk} = h^{ab} h^{cd} (\grad_a k_c)(\grad_b k_d)$, which is the square of the shear of the null geodesic congruence fired from~$\Sigma$ in the~$k^a$ direction; in particular, since~$S_{kk} \geq 0$, the null curvature condition (NCC) -- which is the statement that~$R_{ab} k^a k^b \geq 0$ for any null vector~$k^a$ -- implies~$Q_{kk} \geq 0$.  We may also re-express the Riemann tensor component~$R_{kmkm}$ using the Gauss equation~\eqref{eq:Gauss} and the fact that~$\Sigma$ is extremal; doing so yields
\be
\label{eq:Qkm}
Q_{km} = -G_{km} + \frac{m^2}{2}\left(2R_{kk} + S_{kk}\right) - \frac{1}{2} \, ^\parallel \! R,
\ee
where~$G_{ab}$ is the Einstein tensor of~$g_{ab}$ and~$^\parallel \! R$ is the Ricci scalar of~$h_{ab}$.

We may now combine these results to compute the scalar differential operators~$J_{k,k}$ and~$J_{k,m}$ defined by~\eqref{eq:Lij}.  From~\eqref{subeq:Ldef}, some simple computations using the definition~\eqref{eq:normaldecomp} with the connection coefficients~\eqref{eq:connections} yields
\bea
J_{k,k} f &= -f Q_{kk}, \label{eq:Raychaudhuri} \\
J_{k,m} f &= -D^2 f + 2\chi^a D_a f + \left(D_a \chi^a - |\chi|^2 - Q_{km} \right)f. \label{eq:MOTSoperator}
\eea
Per~\eqref{eq:thetadot}, when~$\delta g_{ab} = 0$ then~$J_{k,k} f$ computes the change in the expansion~$\theta^{(k)}$ when~$\Sigma$ is perturbed in the direction~$f k^a$; using~\eqref{eq:Qkk} we thus recognize~\eqref{eq:Raychaudhuri} as just the Raychaudhuri equation.  On the other hand,~\eqref{eq:MOTSoperator} computes the change in~$\theta^{(k)}$ when~$\Sigma$ is perturbed in the direction~$f m^a$; using~\eqref{eq:Qkm}, it reproduces the known formula of~\cite{AndMar05} used in analyses of marginally outer trapped surfaces.

Now let us take~$m^a = \ell^a$; writing~$\eta^a = \alpha k^a + \beta \ell^a$ we find that the equation of extremal deviation~\eqref{eq:sourcedgeneralJacobi} decomposes into~$J_{k,\ell} \beta = -J_{k,k} \alpha + s_k$ and~$J_{\ell,k} \alpha = -J_{\ell,\ell} \beta + s_\ell$, or
\begin{subequations}
\label{eqs:scalarcodimtwo}
\begin{align}
-D^2 \beta + 2\chi^a D_a \beta - \left(|\chi|^2 - D_a \chi^a + Q_{k\ell} \right)\beta &= \alpha Q_{kk} + s_k, \\
-D^2 \alpha - 2\chi^a D_a \alpha - \left(|\chi|^2 + D_a \chi^a + Q_{k\ell} \right)\alpha &= \beta Q_{\ell\ell} + s_\ell.
\end{align}
\end{subequations}
This is the desired decomposition into the null basis of the normal bundle for codimension-two spacelike extremal surfaces.

\section{Theory of Quantum Surface Deformations}
\label{sec:quantum}

So far, the only special kinds of surfaces we have considered are classical extremal surfaces.  In the context of perturbative quantum gravity, however, we are interested in a more general class of codimension-two surfaces which are obtained by adding appropriate quantum corrections.  Specifically, in a Lorentzian spacetime consider a codimension-two surface~$\Sigma$ which splits a Cauchy slice in two.  The area of a surface is ordinarily corrected to a ``quantum area'', which incorporates contributions from the entanglement entropy of quantum fields across the surface~\cite{Bek72}. This replacement comes from a rich history grounded in black hole thermodynamics, but it is now understood to be relevant in a broader context.  We refer the reader to~\cite{Wal18} and references therein.  

The ``quantum-corrected'' area is Bekenstein's generalized entropy,
\be
\label{eq:Sgen}
S_\mathrm{gen}[\Sigma] = \frac{A[\Sigma]}{4G_N \hbar} + S_\mathrm{out}[\Sigma],
\ee
where~$S_\mathrm{out}[\Sigma]$ is the von Neumann entropy of any (quantum) matter fields living on the portion of the Cauchy slice ``outside'' of~$\Sigma$, as illustrated in Figure~\ref{fig:Sout} (in more general theories of gravity, the area term will be replaced by some other geometric object~\cite{Don13,Cam13,JacMye93,Wal93,IyeWal94,IyeWal95}; here we just focus on quantum corrections to classical Einstein-Hilbert gravity).  Just as classical extremal surfaces are defined as stationary points of the area functional~$A[\Sigma]$, \textit{quantum} extremal surfaces are defined as stationary points of the generalized entropy~$S_\mathrm{gen}[\Sigma]$.

\begin{figure}[t]
\centering
\includegraphics[page=9]{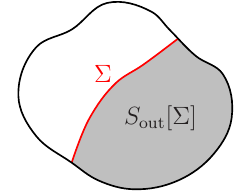}
\caption{A Cauchy-splitting surface~$\Sigma$ (which is necessarily codimension-two)  (red) divides a Cauchy slice into two regions, which allows us to define the entropy~$S_\mathrm{out}[\Sigma]$ of quantum fields to one side of~$\Sigma$.}
\label{fig:Sout}
\end{figure}

There is significant evidence that~$S_\mathrm{gen}$ is UV-finite: renormalization of $1/G_{N}$ cancels out divergences in $S_{\mathrm{out}}$ (see~\cite{BouFis15} and references therein).  It is therefore normally most physically relevant to work with $S_{\mathrm{gen}}$ as a complete quantity without dividing it into geometric and entropic components.  In the present section, however, our goal is to describe the geometric deformations of quantum extremal surfaces, which makes it natural to work with the terms~$A[\Sigma]$ and~$S_\mathrm{out}[\Sigma]$ separately.  To ensure that this is well-defined, we implicitly assume that we have imposed a UV cutoff and renormalization scheme that renders both geometric and entropic terms independently finite (the sum total will be independent of both the cutoff and the scheme). The formalism described in Section~\ref{sec:classical} above is of course well-suited to studying perturbations of the area, but we must take some care to treat the entropy term~$S_\mathrm{out}[\Sigma]$ properly.  In particular, because~$S_\mathrm{out}$ is nonlocal, we must make use of functional derivatives; we therefore pause here to set up the appropriate formalism for treating functional derivatives covariantly before proceeding.

\subsection{Distributional Tensors}
\label{subsec:multiloc}

Because functional derivatives involve global deformations, they often yield ``multilocal'' objects, which are typically distribution-valued; the purpose of this section is to give a precise definition of these non-local, distributional tensors.  To do so, first recall that ordinary tensors over the tangent space~$T_p M$ of a point~$p$ are defined by their action on vectors in~$T_p M$: that is, the tensor~${V_{a_1 \cdots a_k}}^{b_1 \cdots b_l}$ over the tangent space~$T_p M$ is a linear map
\be
V : (T_p M)^k \times (T^*_p M)^l \to \mathbb{R}
\ee
given by
\be
V(v_1, \ldots, v_k, u^1, \ldots, u^l) = {T_{a_1 \cdots a_k}}^{b_1 \cdots b_l} (v_1)^{a_1} \cdots (v_k)^{a_k} (u^1)_{b_1} \cdots (u^l)_{b_l}
\ee
for any~$(v_i)^a \in T_p M$ and~$(u^i)_b \in T^*_p M$.  We define \textit{distributional} tensors on a surface~$\Sigma$ in an analogous way as linear, integral maps from \textit{fields} on~$\Sigma$ to~$\mathbb{R}$.  Precisely, letting~$\Fcal(\Sigma)$ denote the space of scalar fields on~$\Sigma$, then on a given surface~$\Sigma$ we consider a map\footnote{In a slight abuse of notation, here we use~$T_\Sigma M$ and~$T^*_\Sigma M$ to refer to \textit{sections} of the bundles~$T_\Sigma M$ and~$T^*_\Sigma M$, i.e.~to vector and dual vector fields on~$\Sigma$.}
\be
\label{eq:Vmap}
V : (T_\Sigma M)^k \times (T^*_\Sigma M)^l \times \Fcal(\Sigma)^m \to \mathbb{R}
\ee
given explicitly as
\begin{multline}
\label{eq:multilocaldef}
V(v_1, \ldots, v_k, u^1, \ldots, u^l, f_1, \ldots, f_m) = \\ \int_{\Sigma^{k+l+m}} {V_{a_1 \cdots a_k}}^{b_1 \cdots b_l}(p_i, q_i, r_i) \left(\prod_{i = 1}^k (v_i)^{a_i}(p_i)\right) \left(\prod_{i = 1}^l (u^i)_{b_i}(q_i)\right) \left(\prod_{i = 1}^m f_i(r_i)\right)
\end{multline}
with the integral running over~$k+l+m$ copies of~$\Sigma$ labeled by the points~$p_i$,~$q_i$, and~$r_i$ with the natural volume element understood.  Note that each index of~${V_{a_1 \cdots a_k}}^{b_1 \cdots b_l}$ acts on the tangent space of a different point, and in addition~${V_{a_1 \cdots a_k}}^{b_1 \cdots b_l}$ also depends on the~$m$ points which are integrated against scalars in~\eqref{eq:multilocaldef}; the object~${V_{a_1 \cdots a_k}}^{b_1 \cdots b_l}$ is thus an example of what we mean by a ``multilocal'' distributional tensor field on~$\Sigma$.  More generally, we can of course consider maps of the form~\eqref{eq:multilocaldef} for which more than one index of~${V_{a_1 \cdots a_k}}^{b_1 \cdots b_l}$ acts on the tangent space of the \textit{same} point, e.g.~we may consider an object like~$V_{abc}(p,p')$, for which the first two indices act on~$T_p M$ and the last index acts on~$T_{p'} M$.  Indeed, any ordinary tensor field~${V_{a_1 \cdots a_k}}^{b_1 \cdots b_l}(p)$ (all of whose indices act on the tangent space of the point~$p$) can be thought of as a limiting case: given any such ordinary tensor field, we can always define the map
\be
V(v_1, \ldots, v_k, u^1, \ldots, u^l) = \int_\Sigma {V_{a_1 \cdots a_k}}^{b_1 \cdots b_l} (v_1)^{a_1} \cdots (v_k)^{a_k} (u^1)_{b_1} \cdots (u^l)_{b_l},
\ee
where the integral runs over a single copy of~$\Sigma$.  We will refer to any such tensor field over~$\Sigma$, and not just those defined as in~\eqref{eq:multilocaldef}, as a distributional tensor field on~$\Sigma$.  We will specifically be interested in the case of \textit{functional} multilocal tensor fields, which arise from a functional~$V[\Sigma]$ which yields a map~\eqref{eq:Vmap} (or a generalization theoreof as just discussed) for \textit{any} surface~$\Sigma$.  This includes the case of ordinary functionals of~$\Sigma$, which simply map any surface~$\Sigma$ to a real number (for instance, the area functional~$A[\Sigma]$ or the entropy~$S_\mathrm{out}[\Sigma]$).

The indices of a distributional tensor can be raised and lowered in the standard way by using the metric acting on the appropriate tangent space, but because of their distributional nature, we must be careful with contracting their indices.  When a contraction of~${V_{a_1 \cdots a_k}}^{b_1 \cdots b_l}$ \textit{is} defined -- say, between the indices~$a_1$ and~$b_1$ -- it can be evaluated by setting~$p_1 = q_1$ and then taking a standard contraction over the tangent space of this shared point\footnote{The reason such contractions may not always be well-defined is that they essentially require taking some of the~$(v_i)^a$ and~$(u^i)_a$ in the map~\eqref{eq:multilocaldef} to have delta-function support, but~$V$ is defined via its action on smooth vector fields.}.  On the other hand, outer products of distributional tensors are always well-defined, since they correspond to just multiplying the corresponding functionals~\eqref{eq:Vmap} together.

Finally, let us briefly note that since we will be exclusively interested in the case of functionals, for notational convenience we will often leave the argument~$\Sigma$ implied; i.e.~we will often write~$V$ in place of~$V[\Sigma]$ and~${V_{a_1 \cdots a_k}}^{b_1 \cdots b_l}$ in place of~${V_{a_1 \cdots a_k}}^{b_1 \cdots b_l}[\Sigma]$.  Similarly, we will also often forego explicitly calling objects ``functionals'' when it is clear that they are (in much the same way that ``tensor field'' is often colloquially shortened to just ``tensor'').  We will also sometimes refer to distributional tensors with no indices as distributional scalars.

\subsection{Functional Covariant and Lie Derivatives}
\label{subsec:functionalderiv}

Our purpose now is to generalize the notions of ordinary Lie and covariant derivatives to distributional tensor functionals, which will allow us to treat nonlocal variations of such objects covariantly.

\subsubsection*{Functional Covariant Derivatives}

Heuristically, an ordinary functional derivative captures how a functional varies under an infinitesimal variation.  We would like to generalize this notion to a covariant functional derivative, which should capture how a distributional tensor functional varies as the surface~$\Sigma$ on which it is defined is deformed.  To introduce such a derivative operator, we proceed in complete analogy with the logic via which the ordinary covariant derivative~$\grad_a$ is defined: a functional covariant derivative is a map from a distributional tensor functional~${V_{a_1 \cdots a_k}}^{b_1 \cdots b_l}$ to a distributional tensor functional denoted by~$\Dcal {V_{a_1 \cdots a_k}}^{b_1 \cdots b_l}/\Dcal \Sigma^c$, with~$\Dcal/\Dcal \Sigma^a$ obeying the following properties (which we write with all-lower indices for notational expedience, though the index structure may be general):
\begin{subequations}
\begin{enumerate}
	\item \label{cond:normal} Normal to~$\Sigma$: for any distributional tensor field~$V_{a_1 \cdots a_k}$, we have
		\be
		{h_c}^b \frac{\Dcal V_{a_1 \cdots a_k}}{\Dcal \Sigma^b} = 0.
		\ee
	\item \label{cond:linearity} Linearity: for any distributional tensor fields~$V_{a_1 \cdots a_k}$,~$U_{a_1 \cdots a_k}$ and any~$c_1, c_2 \in \mathbb{R}$, we have
		\be
		\frac{\Dcal}{\Dcal \Sigma^b} \left(c_1 V_{a_1 \cdots a_k} + c_2 U_{a_1 \cdots a_k}\right) = c_1 \frac{\Dcal V_{a_1 \cdots a_k}}{\Dcal \Sigma^b} + c_2 \frac{\Dcal U_{a_1 \cdots a_k}}{\Dcal \Sigma^b}.
		\ee
	\item \label{cond:leibnitz} The Leibnitz rule: for any distributional tensor fields~$V_{a_1 \cdots a_k}$,~$U_{a_1 \cdots a_{k'}}$, we have
		\be
		\frac{\Dcal}{\Dcal \Sigma^c}\left(V_{a_1 \cdots a_k} U_{a_1 \cdots a_{k'}}\right) = V_{a_1 \cdots a_k} \frac{\Dcal U_{a_1 \cdots a_{k'}}}{\Dcal \Sigma^c} + \frac{\Dcal V_{a_1 \cdots a_k}}{\Dcal \Sigma^c} \, U_{a_1 \cdots a_{k'}} .
		\ee
	\item \label{cond:contraction} Commutativity with contraction: for any distributional tensor field ${V^{a_1}}_{a_2 \cdots a_k}$ for which the contraction of the first two indices is well-defined, we have
		\be
		\frac{\Dcal}{\Dcal \Sigma^c} \left({V^b}_{ba_1 \cdots a_{k-1}} \right) = \frac{\Dcal {V^b}_{ba_1 \cdots a_{k-1}}}{\Dcal \Sigma^c}.
		\ee
	\item \label{cond:scalar} Variations of scalars: for any family of smooth surfaces~$\Sigma(\lambda)$ generated by a one-parameter group of diffeomorphisms~$\phi_\lambda$ and for any distributional scalar functional~$F$,
		\be
		\label{eq:scalarvar}
		\left. \frac{dF(p_i)}{d\lambda} \right|_{\lambda = 0} \equiv \left. \frac{d}{d\lambda} F[\Sigma(\lambda)](\phi_\lambda(p_i)) \right|_{\lambda = 0} = \int_\Sigma \frac{\Dcal F(p_i)}{\Dcal \Sigma^a(p')} \, \eta^a(p') \, \bm{\eps}(p'),
		\ee
		where the deviation vector~$\eta^a$ is taken to be normal to~$\Sigma$ and we have written the volume element as~$\bm{\eps}(p')$ to indicate explicitly that the integral is taken over~$p'$ (with the other arguments~$p_i$ kept fixed).
\end{enumerate}
The first condition just requires the functional derivative of any object on a surface~$\Sigma$ to be normal to~$\Sigma$ in its ``derivative'' index; this captures the notion that variations in the shape of~$\Sigma$ are contained only in the normal components~$\eta^a_\perp$ of a deviation vector.  The next three conditions are identical to their counterparts for the ordinary covariant derivative~$\grad_a$, and we will not discuss them further.  The fifth condition is the functional generalization of the requirement for ordinary covariant derivatives~$\grad_a$ that for any vector~$v^a$ and scalar field~$f$,~$v(f) = v^a \grad_a f$; it is the crucial property that differentiates functional derivatives from other kinds of derivatives.  In Appendix~\ref{subapp:functionalcovariant}, we show that covariant functional derivatives satisfying properties~\ref{cond:normal}-\ref{cond:scalar} exist by relating them to appropriately constructed ordinary functional derivatives associated to an arbitrary coordinate system.

There are many choices of functional derivative that satisfy the above properties; specifying a unique derivative operator involes imposing some additional constraint.  This is analogous to the freedom that we have in defining the ordinary covariant derivative $\nabla_{a}$. In that case, since the only tensor with which a geometry~$(M,g_{ab})$ comes equipped is the metric, the natural  requirements that uniquely fix~$\grad_a$ are that it be torsion-free and metric-compatible:~$\grad_{[a} \grad_{b]} f = 0$ and~$\grad_a g_{bc} = 0$.  We could generalize these conditions to uniquely fix a preferred functional covariant derivative, but it is much more intuitive to instead note that we can use the ordinary covariant derivative itself to uniquely fix~$\Dcal/\Dcal \Sigma^a$: we require
\begin{enumerate}
	\setcounter{enumi}{5}
	\item \label{cond:compatible} Compatibility with~$\grad_a$: for any \textit{ordinary} tensor field~${V_{a_1 \cdots a_k}}^{b_1 \cdots b_l}(p)$ on~$M$, we have that on any surface~$\Sigma$,
		\be
		\label{eq:compatible}
		\frac{\Dcal {V_{a_1 \cdots a_k}}^{b_1 \cdots b_l}(p)}{\Dcal \Sigma^c(p')} = \delta(p,p') {P_c}^d \grad_d {V_{a_1 \cdots a_k}}^{b_1 \cdots b_l},
		\ee
		where~$\grad_a$ is the preferred (torsion-free and metric-compatible) ordinary covariant derivative operator, and where~$\delta(p,p')$ is the covariant Dirac delta function on~$\Sigma$, defined by~$\int_\Sigma f(p') \delta(p,p') \bm{\eps}(p') = f(p)$ for any ordinary scalar field~$f$ on~$\Sigma$.
\end{enumerate}
\end{subequations}
In Appendix~\ref{subapp:functionalcovariant} we verify that this condition suffices to uniquely specify~$\Dcal/\Dcal\Sigma^a$ and we show that the connection relating it to coordinate functional derivatives is given by (a distributional version of) the usual Christoffel symbols.  We note that it follows that the functional covariant derivative is metric-compatible:~$\Dcal g_{ab}(p)/\Dcal \Sigma^c(p') = 0$.

\subsubsection*{Functional Lie Derivatives}

Our next task is the generalization of the notion of Lie derivatives to distributional tensor functionals.  Recall that for an ordinary tensor field~${V_{a_1 \cdots a_k}}^{b_1 \cdots b_l}$ on~$M$, the Lie derivative along a vector field $\eta^{a}$ roughly measures how ${V_{a_1 \cdots a_k}}^{b_1 \cdots b_l}$ changes with flow along $\eta^{a}$.  More precisely, we introduce a one-parameter group of diffeomorphisms~$\phi_\lambda$ with generator~$\eta^a$; the Lie derivative then computes the difference between~${V_{a_1 \cdots a_k}}^{b_1 \cdots b_l}$ and its pullback~$\phi^*_{-\lambda} {V_{a_1 \cdots a_k}}^{b_1 \cdots b_l}$:
\be
\label{eq:normalLie}
\pounds_\eta {V_{a_1 \cdots a_k}}^{b_1 \cdots b_l} = \lim_{\lambda \to 0} \frac{\phi^*_{-\lambda}{V_{a_1 \cdots a_k}}^{b_1 \cdots b_l} - {V_{a_1 \cdots a_k}}^{b_1 \cdots b_l}}{\lambda}.
\ee
This definition also applies to a distributional tensor functional~${V_{a_1 \cdots a_k}}^{b_1 \cdots b_l}[\Sigma]$, but we must be careful to specify how the pullback should be interpreted.  This is quite straightforward: given a family of surfaces~$\Sigma(\lambda) = \phi_\lambda(\Sigma)$, the functional~${V_{a_1 \cdots a_k}}^{b_1 \cdots b_l}$ defines a distributional tensor~${V_{a_1 \cdots a_k}}^{b_1 \cdots b_l}[\Sigma(\lambda)]$ on each surface, as shown in Figure~\ref{fig:functionalpullback}.  Then for each~$p \in \Sigma$, the tangent space of the ``evolved'' point~$\phi(p) \in \Sigma(\lambda)$ can be pulled back to~$T_p M$ in the usual way, allowing us to pull the distributional tensor field~${V_{a_1 \cdots a_k}}^{b_1 \cdots b_l}[\Sigma(\lambda)]$ back to~$\Sigma$ using~$\phi^*_{-\lambda}$ as usual.  We thus define the functional Lie derivative as
\be
\label{eq:functionalLie}
\pounds_\eta {V_{a_1 \cdots a_k}}^{b_1 \cdots b_l}[\Sigma] = \lim_{\lambda \to 0} \frac{\phi^*_{-\lambda}{V_{a_1 \cdots a_k}}^{b_1 \cdots b_l}[\phi_\lambda(\Sigma)] - {V_{a_1 \cdots a_k}}^{b_1 \cdots b_l}[\Sigma]}{\lambda},
\ee
where we emphasize that we are keeping the notation~$\pounds_\eta$ unchanged because this definition reproduces the conventional one~\eqref{eq:normalLie} when acting on ordinary tensor fields on~$M$.

\begin{figure}[t]
\centering
\includegraphics[page=10]{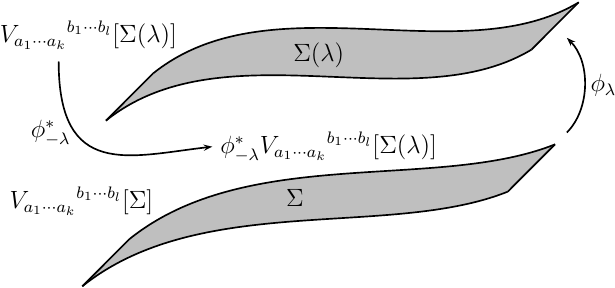}
\caption{For a family of surfaces~$\Sigma(\lambda) = \phi_\lambda(\Sigma)$ defined by a one-parameter group of diffeomorphisms~$\phi_\lambda$, a multilocal tensor functional~${V_{a_1 \cdots a_k}}^{b_1 \cdots b_l}$ yields a multilocal tensor~${V_{a_1 \cdots a_k}}^{b_1 \cdots b_l}[\Sigma(\lambda)]$ on each surface.  Each of these tensors can be pulled back to~$\Sigma$ in the usual way using the pullback~$\phi^*_{-\lambda}$; the difference of these tensors in the limit~$\lambda \to 0$ is what defines the functional Lie derivative.}
\label{fig:functionalpullback}
\end{figure}

We would now like to express the Lie derivative of functionals in terms of the functional covariant derivative.  It follows immediately from~\eqref{eq:scalarvar} that when~$\eta^a$ is normal to~$\Sigma$, the Lie derivative of any distributional scalar functional~$F$ is given by
\be
\label{eq:Liescalar}
\pounds_\eta F = \int_\Sigma \frac{\Dcal F}{\Dcal \Sigma^a(p')} \, \eta^a(p') \, \bm{\eps}(p').
\ee
The analogous expression for general distributional tensors is derived in Appendix~\ref{subapp:functionalLie}; the relevant result for our purposes is that for a dual vector field~$V_a$,
\be
\label{eq:LieVdown}
\pounds_\eta V_a = \int_\Sigma \frac{\Dcal V_a}{\Dcal \Sigma^b(p')} \eta^b(p') \, \bm{\eps}(p') + V_b \grad_a \eta^b.
\ee
Note that when~$V_a$ is an ordinary dual vector field on~$M$, this reproduces (using~\eqref{eq:compatible}) the usual expression for the Lie derivative, as it must.  (The generalization to the case where~$\eta^a$ is not normal to~$\Sigma$ is straightforward and can be found in Appendix~\ref{subapp:functionalLie}.)

\subsection{Equation of Quantum Extremal Deviation}

We have now prepared the technology necessary to derive the analogue of the equation of extremal deviation,~\eqref{eq:sourcedgeneralJacobi}, for quantum extremal surfaces (which must be Cauchy-splitting in order for~$S_\mathrm{out}$ to be defined).  A quantum extremal surface~$\Sigma$ is defined by the condition that for any one-parameter perturbation~$\Sigma(\lambda)$ of it,~$dS_\mathrm{gen}[\Sigma(\lambda)]/d\lambda|_{\lambda = 0} = 0$, and hence~$\Dcal S_\mathrm{gen}/\Dcal \Sigma^a = 0$~\cite{EngWal14}.  From the definition~\eqref{eq:Sgen}, the first area variation formula~\eqref{eq:firstareavar}, and~\eqref{eq:scalarvar}, we have for any~$\eta^a$ normal to~$\Sigma$,
\be
0 = \int_\Sigma \left(K_a + 4G_N \hbar \, \frac{\Dcal S_\mathrm{out}}{\Dcal \Sigma^a} \right) \eta^a \, \bm{\eps};
\ee
thus the condition of quantum extremality is
\be
\label{eq:quantumextremal}
K_a + 4G_N \hbar \, \frac{\Dcal S_\mathrm{out}}{\Dcal \Sigma^a} = 0.
\ee
For this reason it is natural to define the object~$4G_N \hbar \Dcal S_\mathrm{gen}/\Dcal \Sigma^a$ as the ``quantum'' mean curvature of a surface; its component in a null direction~$k^a$ normal to~$\Sigma$ is what~\cite{BouFis15} coined the quantum expansion of the surface in the~$k^a$ direction:
\be
\label{eq:Theta}
\Theta^{(k)} = 4 G_N \hbar \, k^a \frac{\Dcal S_\mathrm{gen}}{\Dcal \Sigma^a}.
\ee
Let us pause to preempt a point of potential perplexity.  In~\cite{BouFis15},~$\Theta^{(k)}$ was defined by considering perturbations of~$\Sigma$ along a null pencil\footnote{The functional derivative~$\delta/\delta V$ used in~\cite{BouFis15} to define~$\Theta^{(k)}$ is realted to our~$\Dcal/\Dcal\Sigma^{a}$ via~$k^a \Dcal S_\mathrm{gen}/\Dcal \Sigma^a = h^{-1/2} \delta S_\mathrm{gen}/\delta V$, where~$h^{1/2}$ is the volume element on~$\Sigma$.}; this construction is sufficient to obtain a null (or more generally, a spacelike) component of~$\Dcal S_\mathrm{gen}/\Dcal \Sigma^a$, since such perturbations keep~$\Sigma$ spacelike.  However, an analogous construction cannot be used to obtain timelike components of~$\Dcal S_\mathrm{gen}/\Dcal \Sigma^a$, since perturbations of~$\Sigma$ along a timelike pencil can make~$\Sigma$ timelike around the perturbation, rendering~$S_\mathrm{out}$ ill-defined.  In our formalism, however, the object~$\Dcal S_\mathrm{out}/\Dcal \Sigma^a$ is defined in a distributional sense by equation~\eqref{eq:scalarvar}; in particular, its definition involves perturbing~$\Sigma$ into one-parameter families of \textit{smooth} surfaces~$\Sigma(\lambda)$, which for sufficiently small~$\lambda$ must be spacelike and achronal as~$\Sigma$ is.  Indeed, that~$\Dcal S_\mathrm{gen}/\Dcal \Sigma^a$ is well-defined as a dual vector is clear from the fact that it can be decomposed in a null basis~$\{k^a, \ell^a\}$ (with~$k \cdot \ell = 1$) as
\be
\frac{\Dcal S_\mathrm{gen}}{\Dcal \Sigma^a} = \frac{1}{4 G_N \hbar} \left[\ell_a \Theta^{(k)} + k_a \Theta^{(\ell)}\right],
\ee
with each of~$\Theta^{(k)}$,~$\Theta^{(\ell)}$ well-defined.

Now, to use~\eqref{eq:quantumextremal} to derive the quantum analogue of~\eqref{eq:sourcedgeneralJacobi} let us again consider a one-parameter family of surfaces~$\Sigma(\lambda)$.  In Section~\ref{sec:classical}, we required each surface in this family to be extremal with respect to some varying metric~$g_{ab}(\lambda)$; then the family~$\Sigma(\lambda)$ encodes perturbations to a surface~$\Sigma$ as the state of the ambient geometry changes (or as the boundary conditions of~$\Sigma$ are varied).  This variation in the geometry should presumably be coupled to a variation in the state of any matter fields, and so the generalization to the present case is clear: we require that for each~$\lambda$,~$\Sigma(\lambda)$ be a quantum extremal surface in the state with geometry~$g_{ab}(\lambda)$ and with matter entropy functional~$S^{(\lambda)}_\mathrm{out}$.  Differentiating~\eqref{eq:quantumextremal} in~$\lambda$, we thus have that the variation in the entropy term gets a geometric contribution and a contribution from the explicit~$\lambda$-dependence of the one-parameter family of functionals~$S^{(\lambda)}_\mathrm{out}$:
\be
\label{eq:DSderiv}
\left. \frac{d}{d\lambda} \left(\frac{\Dcal S_\mathrm{out}}{\Dcal \Sigma^a}\right) \right|_{\lambda = 0} = \pounds_\eta \left(\frac{\Dcal S_\mathrm{out}}{\Dcal \Sigma^a}\right) + \frac{\Dcal \delta S_\mathrm{out}}{\Dcal \Sigma^a},
\ee
where the functional~$\delta S_\mathrm{out}$ is given explicitly on any fixed surface~$\Sigma$ as
\be
\delta S_\mathrm{out}[\Sigma] \equiv \left. \frac{dS_\mathrm{out}^{(\lambda)}[\Sigma]}{d\lambda} \right|_{\lambda = 0}.
\ee
The variation of~\eqref{eq:quantumextremal} is therefore
\be
\dot{K}_a + 4G_N \hbar \, {P_a}^b \pounds_\eta \left(\frac{\Dcal S_\mathrm{out}}{\Dcal \Sigma^b}\right) + 4G_N \hbar \frac{\Dcal \delta S_\mathrm{out}}{\Dcal \Sigma^a} = 0,
\ee
where we noted by the same arguments used in Section~\ref{sec:classical} that since~$\Dcal S_\mathrm{out}/\Dcal \Sigma^a$ is always normal to~$\Sigma$, its Lie derivative must be as well.  Then taking~$\eta^a = \eta_\perp^a$ to be normal to~$\Sigma$ and using~\eqref{eq:Kdotgeneral},~\eqref{eq:LieVdown}, and the fact that the quantum mean curvature of~$\Sigma$ vanishes, we obtain the desired equation
\begin{multline}
\label{eq:quantumextremaldeviation}
J(\eta_\perp)_a + 4 G_N \hbar \int_\Sigma {P_a}^b \frac{\Dcal^2 S_\mathrm{out}}{\Dcal \Sigma^c(p') \Dcal \Sigma^b} \, \eta_\perp^c(p') \bm{\eps}(p') = \\ s_a + 4G_N \hbar \left[{P_a}^b P^{cd} \frac{\Dcal S_\mathrm{out}}{\Dcal \Sigma^d} \, \delta g_{bc} - \frac{\Dcal \delta S_\mathrm{out}}{\Dcal \Sigma^a}\right].
\end{multline}
This is the sourced equation of quantum extremal deviation, describing how a quantum extremal surface varies in response to a change of the state (including both the geometry and matter entropy) and of its boundary conditions.  Two special cases are worth highlighting.  First, if we set all sources to vanish, we obtain the unsourced equation of extremal deviation, which describes how a quantum extremal surface may be perturbed (due to modifications of its boundary conditions) in a fixed geometry and state while maintaining its quantum extremality:
\be
\label{eq:unsourcedquantumextremaldeviation}
J(\eta_\perp)_a + 4 G_N \hbar \int_\Sigma {P_a}^b \frac{\Dcal^2 S_\mathrm{out}}{\Dcal \Sigma^c(p') \Dcal \Sigma^b} \, \eta_\perp^c(p') \bm{\eps}(p') = 0;
\ee
this is the quantum generalization of~\eqref{eq:unsourcedJacobi}.  Second, we can obtain the ``quantum correction'' to a classical extremal surface by requiring that~$\Sigma$ be a classical extremal surface and that quantum corrections be ``turned on'' with~$\lambda$ by taking~$S^{(\lambda)}_\mathrm{out} = \lambda S_\mathrm{out}$.\footnote{Of course, no matter state actually yields such an entropy functional; for the purposes of obtaining~\eqref{eq:quantumcorrection} we are just thinking of~$S^{(\lambda)}_\mathrm{out}$ as an arbitrary functional we may specify by hand. Since our derivation of~\eqref{eq:quantumextremaldeviation} is purely kinematical, this does not pose a problem.}  We thereby obtain
\be
\label{eq:quantumcorrection}
J(\eta_\perp)_a = - 4G_N \hbar \, \frac{\Dcal S_\mathrm{out}}{\Dcal \Sigma^a}.
\ee

\section{Stability of Extremal Surfaces}
\label{sec:stability}

Extremal surfaces are defined by the requirement that their first area variation vanish.  It is sometimes useful, however, to classify them by additional properties which allows us to give them a clearer geometric interpretation.  The purpose of this section is thus to motivate notions of stability for extremal surfaces; we will focus on the case in which~$h_{ab}$ has definite sign (i.e.~when the components of~$h_{ab}$ in an orthonormal frame are all~$+1$ or all~$-1$), though for completeness we will conclude with a brief mention of stability when~$h_{ab}$ is indeterminate.

When~$h_{ab}$ has fixed sign, the operator~$J$ is elliptic, and the notions of stability discussed here are constraints on the Dirichlet spectrum of~$J$ (that is, the spectrum of~$J$ on the space of perturbations vanishing at~$\partial \Sigma$) which stem from natural geometric considerations; it is for this reason that~$J$ is sometimes called the stability operator.  In short, we will review two notions of stability.  The first is \textit{strong stability}, the requirement that the Dirichlet spectrum of~$J$ be bounded by zero; in certain cases this corresponds to the requirement that an extremal surface be a local minimum or maximum of the area functional (recall that unfortunately here the word ``extremal'' does not necessarily mean an extremum of the area functional).  This notion of stability can be further divided into two sub-cases: an extremal surface is \textit{strictly} strongly stable (or just strictly stable for short) if the Dirichlet spectrum of~$J$ has definite sign, and \textit{marginally} strongly stable (or just marginally stable for short) if the Dirichlet spectrum of~$J$ has semidefinite sign and contains zero.  Our second notion of stability, which we term \textit{weak stability}, is the requirement that the Dirichlet spectrum of~$J$ simply not contain zero, which is related to the continued existence of a surface under perturbations.  Note that a strictly stable surface is also weakly stable, and a marginally stable surface is not weakly stable.

\subsection{Strong Stability}
\label{subsec:strong}

It is useful to begin by considering the case where~$(M,g_{ab})$ is Riemannian, in which the picture is most intuitive.  In this case, we often call an extremal surface~$\Sigma$ ``minimal'' based on the intuition that it should be a minimum of the area functional (with boundary conditions that fix~$\partial \Sigma$).  But not every extremal surface in a Riemannian manifold is minimal: for example, consider portions of great circles (i.e.~one-dimensional extremal surfaces) on the two-sphere.  If a segment of a great circle traverses less than halfway around the sphere, as shown in Figure~\ref{subfig:minimaltwosphere}, it is indeed the minimal-length curve connecting its endpoints.  However, if a segment of a great circle goes \textit{more} more than halfway around the sphere, as in Figure~\ref{subfig:unstabletwosphere}, then small deformations of it can shorten its length while keeping its endpoints fixed.  (This property stems from the fact that a geodesic starting at the North pole has a conjugate point at the South pole.)

\begin{figure}[t]
\centering
\subfigure[]{
\includegraphics[page=11,width=0.25\textwidth]{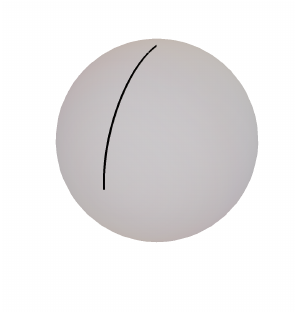}
\label{subfig:minimaltwosphere}
}
\hspace{0.5cm}
\subfigure[]{
\includegraphics[page=12,width=0.25\textwidth]{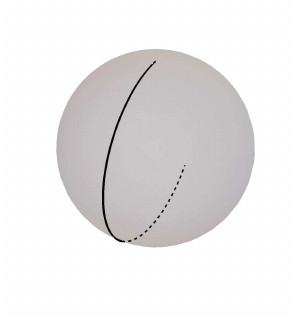}
\label{subfig:unstabletwosphere}
}
\hspace{0.5cm}
\subfigure[]{
\includegraphics[page=13,width=0.25\textwidth]{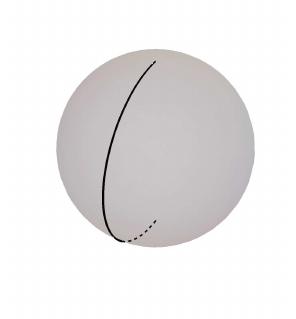}
\label{subfig:marginaltwosphere}
}
\caption{\subref{subfig:minimaltwosphere}: A geodesic on the two-sphere starting at the North pole is the minimal-length curve connecting its endpoints as long as it doesn't reach the South pole.  \subref{subfig:unstabletwosphere}: If it goes past the South pole, small deformations of it can decrease its length.  Case~\subref{subfig:minimaltwosphere} is said to be strictly stable; case~\subref{subfig:unstabletwosphere} is not strongly stable; and case~\subref{subfig:marginaltwosphere} (where a geodesic starts at the North pole and ends at the South pole) is marginally stable.}
\label{fig:twosphere}
\end{figure}

To distinguish between the case where~$\Sigma$ is locally minimal and when it isn't, we must look at second derivatives of the area: that is,~$\Sigma$ being a local minimum of the area functional requires that for any one-parameter deformation~$\Sigma(\lambda)$ with~$\partial \Sigma(\lambda) = \partial \Sigma$, the area~$A(\lambda)$ of these surfaces obeys~$d^2 A/d\lambda^2|_{\lambda = 0} \geq 0$.  In~\eqref{eq:extremalsecondareavar} we show that~$J$ is essentially the Hessian of the area functional, which in particular implies that
\be
\label{eq:inproddef}
\left. \frac{d^2A}{d\lambda^2}\right|_{\lambda = 0} = \int_\Sigma \eta_\perp^a J (\eta_\perp)_a \equiv \inprod{\eta_\perp}{J\eta_\perp},
\ee
and therefore the requirement that~$\Sigma$ be a local minimum of the area functional requires~$\inprod{\eta_\perp}{J\eta_\perp} \geq 0$ for any~$\eta^a_\perp$ which vanishes at~$\partial \Sigma$; in other words, the Dirichlet spectrum of~$J$ must be non-negative\footnote{The Dirichlet spectrum of~$J$ is guaranteed to be real by virtue of the fact that~$J$ is self-adjoint under the inner product~\eqref{eq:inproddef}, which is Hermitian in Riemannian signature.}.  This is a non-trivial constraint independent of the extremality condition~$K_a = 0$, and is precisely the notion of \textit{strong stability} mentioned above.  Moreover, note that the sub-case of strict stability (that is, the requirement that the Dirichlet spectrum of~$J$ be strictly positive) guarantees that~$\Sigma$ is minimal, since all perturbations of strictly stable extremal surfaces must increase their area.  On the other hand, the sub-case of marginal stability (in which zero is in the Dirichlet spectrum of~$J$) is agnostic about whether or not~$\Sigma$ is minimal: it just ensures that perturbations of~$\Sigma$ do not decrease its area to second order (higher derivatives of~$A(\lambda)$ would then be needed to determine whether or not~$\Sigma$ is actually minimal).

To develop this idea a little more explicitly, note that we may always (regardless of signature) use~\eqref{subeq:Ldef} to write
\be
\label{eq:Ldecomp}
\inprod{\eta_\perp}{J\eta_\perp} = \inprod{D\eta_\perp}{D\eta_\perp} - \inprod{\eta_\perp}{Q\eta_\perp},
\ee
where we used the fact that the adjoint of~$D_a$ is~$-D_a$ under the inner product~~\eqref{eq:inproddef} with Dirichlet boundary conditions (see Appendix~\ref{app:variations} for more details).  In Riemannian signature, the object~$\inprod{D\eta_\perp}{D\eta_\perp}$ is non-negative, and we have
\be
\inprod{\eta_\perp}{J\eta_\perp} \geq - \inprod{\eta_\perp}{Q\eta_\perp}.
\ee
Since at each point~$p \in \Sigma$,~${Q^a}_b$ (thought of as a map from~$T_p^\perp M$ to itself) has finitely many (finite) eigenvalues, we conclude that the Dirichlet spectrum of~$J$ must \textit{always} be bounded below (whether or not~$\Sigma$ is strongly stable).  It is this condition that guarantees that the criterion of strong stability is a reasonable one: strong stability simply requires that this lower bound on the Dirichlet spectrum of~$J$ be zero.

More generally, whenever~$h_{ab}$ is positive (negative) definite,~$J$ is elliptic, and so its Dirichlet spectrum must always be bounded below (above).  Strong stability requires this bound to be zero, from which it follows that if~$P_{ab}$ also has definite sign, then we immediately conclude that for a strongly stable extremal surface,
\bea
\left. \frac{d^2A}{d\lambda^2} \right|_{\lambda = 0} &\geq 0 \mbox{ if } h_{ab} \mbox{ and } P_{ab} \mbox{ have the same sign}, \\
\left. \frac{d^2A}{d\lambda^2} \right|_{\lambda = 0} &\leq 0 \mbox{ if } h_{ab} \mbox{ and } P_{ab} \mbox{ have opposite signs},
\eea
with the inequalities being obeyed strictly in the case of strict stability.  Thus a strictly stable extremal surface can sensibly be said to be ``minimal'' (``maximal'') if~$h_{ab}$ and~$P_{ab}$ have the same (opposite) sign.  Examples of minimal surfaces of course include the Riemannian context discussed above, while examples of maximal surfaces include timelike geodesics and spacelike hypersurfaces in a Lorentzian geometry.  Indeed, the case of timelike geodesics has been of crucial importance in the derivation of the Penrose-Hawking singularity theorems: a key result is that a timelike geodesic from~$p$ to~$q$ is strictly stable if and only if it has no points conjugate to~$p$ between~$p$ and~$q$~\cite{HawEll}.

When~$P_{ab}$ does not have definite sign (as in the case of codimension-two spacelike surface in Lorentzian spacetimes), it is clear from~\eqref{eq:Ldecomp} that even if the Dirichlet spectrum of~$J$ is bounded, the inner product~$\inprod{\eta_\perp}{J\eta_\perp}$ is not, and therefore there is no way to ensure that area variations have fixed sign.  Nevertheless, there is still a way to ascribe a physical interpretation to a constraint on the spectrum of~$J$; it is this physical interpretation that yields the notion of weak stability.

\subsection{Weak Stability}
\label{subsec:weak}

If area variations necessarily have indefinite sign when~$P_{ab}$ does, is there another physical notion of stability that we can impose on such surfaces?  A natural one is provided by the notion of stability under perturbations: roughly speaking, an extremal surface~$\Sigma$ is stable under perturbations if it does not ``cease to exist'' under arbitrary deformations of either its boundary~$\partial \Sigma$ or the ambient geometry~$g_{ab}$.  Formally, we require that the sourced equation of extremal deviation~\eqref{eq:sourcedgeneralJacobi} must have a solution for any inhomogeneous Dirichlet boundary conditions on~$\eta_\perp^a$ and for any metric perturbation~$\delta g_{ab}$, which is a necessary condition for there to exists a continuous deformation of~$\Sigma$ which keeps it extremal under arbitrary continuous deformations of~$\partial \Sigma$ and~$g_{ab}$\footnote{We emphasize that this is merely a \textit{necessary} condition: the existence of a solution to~\eqref{eq:sourcedgeneralJacobi} for arbitrary~$\delta g_{ab}$ and~$\eta^a_\perp|_{\partial \Sigma}$ is not in general sufficient to conclude the existence of a one-parameter family of surfaces~$\Sigma(\lambda)$ corresponding to a one-parameter family of metrics~$g_{ab}(\lambda)$ with~$\dot{g}_{ab}(\lambda = 0) = \delta g_{ab}$ and boundary conditions~$\partial \Sigma(\lambda)$ with deviation vector~$\eta^a_\perp|_{\partial \Sigma}$.  See e.g.~\cite{Cam18} for more details.}.  For perturbations of the geometry~$g_{ab}$ but for which the boundary~$\partial \Sigma$ is left fixed, this condition is just the requirement that~$J$ be invertible on the space of vector fields normal to~$\Sigma$ which vanish at~$\partial \Sigma$.  To include deformations of~$\partial \Sigma$, we invoke a more refined version of this statement in the form of the Fredholm alternative.

First, consider an arbitrary elliptic operator~$L$ on~$\Sigma$ which acts on \textit{scalar fields}; then the Fredholm alternative (see~\cite{evans10} for an introductory discussion) is an exclusive alternative which states that precisely one of the following must be true:
\begin{enumerate}
	\item The homogeneous boundary-value problem~$L f = 0$,~$f|_{\partial \Sigma} = 0$ has nontrivial solutions; or
	\item The inhomogeneous boundary-value problem~$L f = s$,~$f|_{\partial \Sigma} = v$ has a unique solution for arbitrary functions~$s$,~$v$.
\end{enumerate}
It is in general not entirely trivial to generalize the Fredholm alternative to elliptic operators that act on a vector bundle over~$\Sigma$, as our operator~$J$ does.  However, since~$J$ is a linear operator, it is expected to satisfy the Fredholm alternative as well~\cite{Igor}.  Thus invoking the Fredholm alternative on the operator~$J$, we conclude that the requirement that the sourced equation of extremal deviation have a solution for any perturbation of the metric or of~$\partial \Sigma$ is equivalent\footnote{We have been a little fast here: strictly speaking, stability under perturbations merely requires that~$J \eta^a = s^a$ have a solution for any source~$s^a$ constructed from a metric pertrbation~$\delta g_{ab}$ as given by~\eqref{subeq:sdef}, for any \textit{arbitrary}~$s^a$.  But it is relatively easy to show that any desired~$s^a$ can be generated from some metric perturbation~$\delta g_{ab}$ in this way, so the two statements are equivalent.}  to the requirement that~$J \eta^a = 0$ have no nontrivial solutions with~$\eta^a|_{\partial \Sigma} = 0$; in other words, that the Dirichlet spectrum of~$J$ cannot contain zero.  We therefore define \textit{weak stability} as the requirement that the Dirichlet spectrum of~$J$ does not contain a zero eigenvalue.

It is useful to revisit the simple example of geodesics on the two-sphere, Figure~\ref{fig:twosphere}, in this context.  A great circle that traverses less than half the sphere, as in Figure~\ref{subfig:minimaltwosphere}, is strictly stable, and therefore must also be weakly stable.  Indeed, it is easy enough to see that small perturbations of one of the endpoints of the geodesics in Figure~\ref{subfig:minimaltwosphere} must cause a small perturbation of the geodesic itself (the same must also be true of small perturbations of the geometry, though perhaps this is less easy to intuit).  Interestingly, a geodesic which traverses more than half the sphere, Figure~\ref{subfig:unstabletwosphere}, is weakly stable even though it's not strongly stable: small perturbations of the endpoints will induce a small perturbation of the geodesic, even if there are deformations that locally decrease its length.  Finally, if the endpoints of the geodesic lie precisely at the North and South poles, as in Figure~\ref{subfig:marginaltwosphere}, it is \textit{not} weakly stable even though it is marginally stable: certain (in fact, generic) small perturbations of one of the endpoints will induce a discontinuous global change in the geodesic (for instance, if the geodesic orginally lies along the line~$\phi = 0$ in the usual spherical coordinates, perturbing one endpoint an arbitrarily small amount onto the line~$\phi = \pi/2$ requires the geodesic to abruptly ``jump'' to this line everywhere).

Finally, it is worth commenting that in the special case of codimension-two spacelike surfaces in Lorentzian spacetimes -- the case of primary interest in the context of entanglement entropy in AdS/CFT -- the notions of strong and weak stability just discussed appear quite naturally.  For instance, the Lewkowykcz-Maldacena proof of RT~\cite{LewMal13} makes use of the replica trick, in which the RT surface is the limit of a family of surfaces in an analytically continued family of geometries.  The requirement that the RT surface be stable ensures that (linearized) such deformations of the surface exist.  Similarly, the maximin construction of HRT surfaces imposes a notion of ``spatial'' strict stability in the following sense.  The maximin surface is found by first finding the minimal surface~$X_\mathrm{min}[H]$ on every possible Cauchy slice~$H$ (containing~$\partial X$), and then maximizing over the area of all such surfaces; the resulting surface~$X_*$ will be a minimal surface lying on some Cauchy slice~$H_*$, and therefore must be strongly stable with respect to the geometry on~$H_*$.  Moreover, the stability requirement imposed in~\cite{Wal12,MarWal19} is that if~$H_*$ is slightly perturbed to some nearby Cauchy slice~$H$, the corresponding minimal surface~$X_\mathrm{min}[H]$ must be a small deformation of~$X_*$; but since a deformation of~$H_*$ perturbs the intrinsic geometry of this Cauchy slice, weak stability of~$X_*$ within~$H_*$ is sufficient to guarantee that linearized variations of~$X_*$ exist for small deformations of~$H_*$.  Thus requiring that~$X_*$ be strictly stable in~$H_*$ is sufficient to ensure a linearized version of the notion of stability enforced in~\cite{Wal12,MarWal19}.

\subsection{Surfaces with $h_{ab}$ of Indefinite Sign}
\label{subsec:indefinitehab}

Before concluding, let us briefly say a few words on the situation in which~$h_{ab}$ does not have fixed sign (for example, a timelike surface of dimension greater than one in a Lorentzian spacetime).  In such a case,~$J$ is not elliptic, and we cannot sensibly refer to its Dirichlet spectrum, let alone make any statements about its boundedness.  Is there a nevertheless a notion of stability we may impose in this case?

The answer, perhaps surprisingly, is in the affirmative, but requires generalizing the notion of an extremal surface.  Specifically, here we have been exclusively focused on the case of surfaces that are stationary points of the area functional; however, we may equally well consider surfaces that are stationary points of some more general geometric functional~$F[\Sigma]$.  Perturbations of such a surface~$\Sigma$ will obey some linear equation~$\widetilde{J} \eta^a = \tilde{s}^a$ (for some differential operator~$\widetilde{J}$) which is the analogue of the equation of extremal deviation~\eqref{eq:sourcedgeneralJacobi}.  If~$h_{ab}$ does not have definite sign, a natural notion of stability of~$\Sigma$ is the requirement that~$\widetilde{J}$ be hyperbolic, so that any ``initial perturbation'' $\eta^a_0$ defined on a spacelike slice of~$\Sigma$ propagates causally to all of~$\Sigma$~\cite{Car92}.  Note that this physical perspective is very similar in spirit to that of weak stability, which requires an extremal surface (with~$h_{ab}$ of definite sign) to continue to exist under arbitrary perturbations of its boundary and of the ambient geometry.

In the case of timelike \textit{extremal} surfaces, however, the operator~$\widetilde{J}$ is just the stability operator~$J$, which is manifestly hyperbolic.  Thus for timelike extremal surfaces, this \textit{dynamical} notion of stability is always satisfied.

\section{Causal Wedge Inclusion}
\label{sec:causal}

We now transition to reaping the benefits of the formalism developed in the first half of the paper by studying some applications in AdS/CFT. As discussed in Section~\ref{sec:intro}, the consistency of subregion/subregion duality manifests in significant constraints on the behavior of surfaces via subregion/subregion duality. In this section we focus on causal wedge inclusion, which we remind the reader requires that for any boundary region~$R$, the HRT surface~$X_R$ and the holographic causal information surface~$C_R$ are nowhere timelike-separated and with~$C_R$ nowhere to the outside of~$X_R$, as shown in Figure~\ref{fig:CWI}.  For arbitrary boundary region~$R$ and bulk geometry obeying the null curvature condition to leading order in~$1/N$, it is known that~$W_C[R]$ is typically a proper subset of~$W_E[R]$, and thus~$X_R$ and~$C_R$ are non-perturbatively separated, a property that cannot be violated by arbitrarily small deformations (as we consider here).  Consequently, in these cases our local variational formalism does not yield any constraints on the bulk geometry from CWI.

\begin{figure}[t]
\centering
\includegraphics[page=1,width=0.3\textwidth]{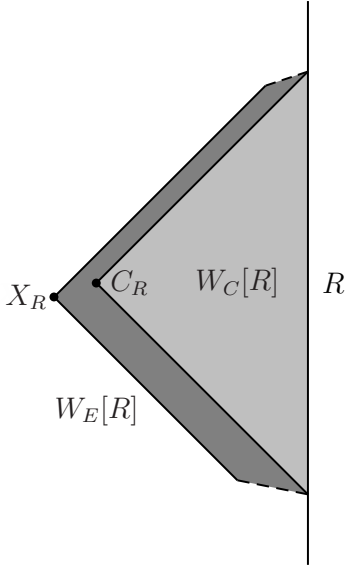}
\caption{Causal wedge inclusion is the requirement that the causal wedge~$W_C[R]$ (light gray) of some boundary region~$R$ must always lie within the entanglement wedge~$W_E[R]$ (dark gray).  This is equivalent to the requirement that the causal information surface~$C_R$ must be spacelike-separated (or marginally, null-separated) towards~$R$ of the HRT surface~$X_R$.}
\label{fig:CWI}
\end{figure}

However, in the non-generic case that CWI is saturated, i.e.~$W_C[R] = W_E[R]$, then arbitrarily small perturbations of either the spacetime or of the entropy~$S_\mathrm{out}[\Sigma]$ run into the danger of violating it; in such cases, our variational formalism constrains the allowed behavior of such perturbations.  Let us therefore consider perturbations to spacetimes and boundary regions that saturate CWI.  To saturate CWI at leading order in~$1/N$, the null boundaries of the causal/entanglement wedge must have zero expansion everywhere, since if we assume the NEC holds at leading order then the outgoing expansion of~$\partial W_C[R]$~($\partial W_E[R]$) must be non-negative (non-positive), and thus if~$\partial W_C[R] = \partial W_E[R]$ this expansion must vanish.  This in turn implies that this boundary is a local Killing horizon in the sense that the generator~$k^a$ of any piecewise null piece of~$\partial W_C[R]$ can be chosen such that~$\pounds_k g_{ab} = 2\grad_{(a} k_{b)} = 0$ on~$\partial W_C[R]$.  Examples of this saturation therefore include the case where~$R$ is (the boundary causal development of) a ball-shaped region and the bulk is vacuum AdS (in which case the causal/entanglement wedge is a Rindler wedge of pure AdS) or the case where~$R$ is a complete connected component of the boundary of a stationary black hole.  Of these cases, the former is more interesting, as pure AdS can be foliated by Rindler horizons and therefore CWI would constrain linear perturbations of the metric everywhere.

Let us therefore consider a ball-shaped region~$R$ on the boundary of pure AdS and examine how the surfaces~$X_R$ and~$C_R$ are perturbed under a perturbation of the state, both in the sense of a perturbation to the spacetime or of the entropy functional~$S_\mathrm{out}[\Sigma]$.

\subsection{Perturbation of $X_R$}

In general, a linearized perturbation to the quantum extremal surface under a deformation of the metric is governed by the sourced equation of quantum extremal deviation~\eqref{eq:quantumextremaldeviation}.  In vacuum and when~$R$ is a ball-shaped region,~$X_R$ is the bifurcation surface~$\Hcal$ of an AdS Rindler horizon, and thus its perturbations are governed by the simpler equation~\eqref{eq:quantumcorrection} which describes how a \textit{classical} extremal surface responds to quantum corrections and a perturbation of the geometry.  Note that~$\Hcal$ is totally geodesic (that is,~${K^a}_{bc} = 0$), and moreover since pure AdS is conformally flat, it follows from the Ricci equation~\eqref{eq:conformalRicci} that the curvature of the normal bundle of~$\Hcal$ vanishes.  There therefore exists a null Fermi-Walker frame; that is, there exists a null basis~$\{k^a, \ell^a\}$ of the normal bundle such that~$D_a k^b = 0 = D_a \ell^b$, which besides the usual normalization~$k \cdot \ell = 1$ we also take to be outward-pointing in the sense that~$k^a$ and~$\ell^a$ both point towards~$W_E[R]$\footnote{In the coordinates of~\eqref{eq:AdS}, we may take, for instance,
\be
k^a = \frac{\rho\sech\chi}{\sqrt{2} \, l} \left((\partial_t)^a - (\partial_\rho)^a\right), \qquad \ell^a = -\frac{\rho\sech\chi}{\sqrt{2} \, l} \left((\partial_t)^a + (\partial_\rho)^a\right).
\ee}.  In this frame,~$\chi_a$ vanishes.  Moreover, using the expressions~\eqref{eq:Qkk} for~$Q_{kk}$ and~$Q_{k\ell}$, as well as the fact that the Riemann tensor of pure AdS can be written as
\be
\label{eq:AdSRiemann}
R_{abcd} = \frac{2}{l^2} g_{a[d} g_{c]b}
\ee
with~$l$ the AdS scale, we find that~$Q_{kk} = 0 = Q_{\ell\ell}$ and~$Q_{k\ell} = -(d-2)/l^2$.  Consequently, in this frame the components of the equation of extremal deviation~\eqref{eq:quantumcorrection} reduce to the decoupled equations
\begin{subequations}
\label{eq:pureAdSJacobi}
\begin{align}
D^2_{\Hcal} \alpha_E -\frac{d-2}{l^2} \, \alpha_E &= -s_\ell + 4 G_N \hbar \, \ell^a \frac{\Dcal \delta S_\mathrm{out}}{\Dcal \Sigma^a}, \\ 
D^2_\Hcal \beta_E - \frac{d-2}{l^2} \, \beta_E &= -s_k + 4 G_N \hbar \, k^a \frac{\Dcal \delta S_\mathrm{out}}{\Dcal \Sigma^a},
\end{align}
\end{subequations}
where the subscript~$E$ denotes the fact that~$\alpha_E$ and~$\beta_E$ are the components of the deviation vector~$\eta^a_E = \alpha_E k^a + \beta_E \ell^a$ describing the perturbation to the HRT surface, and~$D^2_{\Hcal}$ denotes the Laplacian on the Rindler horizon~$\Hcal$.

Since the perturbed surface~$X[R]$ and unperturbed horizon~$\Hcal$ are both anchored to the same boundary region, the perturbations~$\alpha_E$,~$\beta_E$ must vanish at~$\partial \Hcal$.  Thus introducing the Dirichlet Green's function~$G(p,p')$ which obeys
\be
\label{eq:Greens}
D^2_{\Hcal} G(p,p') - \frac{d-2}{l^2} \, G(p,p') = -\delta(p,p')
\ee
and vanishes as either~$p$ or~$p'$ approach~$\partial \Hcal$, the system~\eqref{eq:pureAdSJacobi} can be solved immediately:
\begin{subequations}
\label{eq:extremalalphabeta}
\begin{align}
\alpha_E &= \int_\Hcal G(p,p') \left[s_\ell(p') - 4 G_N \hbar \, \ell^a \frac{\Dcal \delta S_\mathrm{out}}{\Dcal \Sigma^a(p')} \right]\bm{\eps}(p'),\\
\beta_E &= \int_\Hcal G(p,p') \left[s_k(p') - 4 G_N \hbar \, k^a \frac{\Dcal \delta S_\mathrm{out}}{\Dcal \Sigma^a(p')}\right]\bm{\eps}(p').
\end{align}
\end{subequations}
The form of~$G(p,p')$ is given both in general dimension and explicitly for~$d = 3, \ldots, 7$ in Appendix~\ref{app:Greens}, though the results are not particularly illuminating.

\subsection{Perturbation of $C_R$}

Next, let us compute the deformation of the causal surface~$C_R$.  By its definition,~$C_R$ is the intersection of the future and past causal horizons~$\partial J^-[R]$,~$\partial J^+[R]$, which are generated by null geodesics fired from the boundary.  In pure AdS, when~$R$ is (the causal development of) a ball-shaped region these causal horizons are just Poincar\'e horizons~$\Pcal^+$,~$\Pcal^-$ fired from the future and past tips of the boundary causal diamond~$R$, and their intersection defines the surface~$\Hcal$.  The deformation of~$\Pcal^\pm$, and thus of the causal rim~$C_R$, under a perturbation of the metric can be found by computing the deformation of the individual null generators of~$\Pcal^\pm$, which obey the sourced equation of \textit{geodesic} deviation~\eqref{eq:sourcedJacobi}.

Consider therefore the null basis~$\{k^a, \ell^a\}$ of the normal bundle of~$\Hcal$ introduced above, and extend it to the entirety of the Poincar\'e horizons~$\Pcal^\pm$ via parallel transport along their null generators; then~$k^a$ generates~$\Pcal^+$ and~$\ell^a$ generates~$\Pcal^-$.  Next, note every point~$p$ on~$\Hcal$ lies at the intersection of two null geodesics~$\gamma^+ \subset \Pcal^+$ and~$\gamma^- \subset \Pcal^-$.  If~$\xi_\pm^a$ are the deviation vectors along these geodesics describing their response to a metric perturbation, then the perturbation of the point~$p$ is governed by the components of~$\xi^a_\pm$ normal to~$\Hcal$.  In particular, let~$\eta_C^a = \alpha_C k^a + \beta_C \ell^a$ be the deviation vector field describing the perturbation of~$\Hcal$ to~$C_R$; then the components~$\alpha_C$ and~$\beta_C$ at~$p$ are determined by the components~$k \cdot \xi_+$ (which determines how much~$\gamma^+$ is perturbed in the~$\ell^a$ direction) and~$\ell \cdot \xi_-$ (which determines how much~$\gamma_-$ is perturbed in the~$k^a$ direction), as shown in Figure~\ref{fig:perturbedgenerators}:
\be
\alpha_C(p) = \ell \cdot \xi_-(p), \qquad \beta_C(p) = k \cdot \xi_+(p).
\ee

\begin{figure}[t]
\centering
\includegraphics[page=14,width=0.25\textwidth]{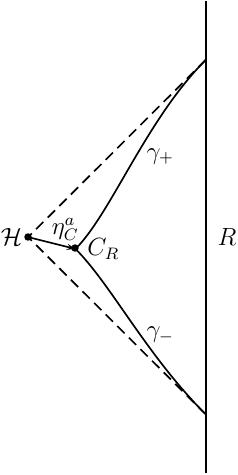}
\caption{Under a deformation of the spacetime, the Rindler horizon~$\Hcal$ is perturbed to the causal rim~$C_R$.  At each point~$p \in \Hcal$, the components of the corresponding deviation vector field~$\eta^a_C(p)$ are determined by the perturbation of the null geodesics~$\gamma^\pm$ which intersect at~$p$.  }
\label{fig:perturbedgenerators}
\end{figure}

The components~$\ell \cdot \xi_-$ and~$k \cdot \xi_+$ can be computed easily by contracting~\eqref{eq:sourcedJacobi} with either~$k^a$ or~$\ell^a$ and using the expression~\eqref{eq:AdSRiemann} for the Riemann tensor in pure AdS; a straightforward integration yields
\be
\label{eq:causalalphabeta}
\alpha_C = \xi_- \cdot \ell = \frac{1}{2} \int_{\gamma_-} \delta g_{ab} \ell^a \ell^b d\sigma_-, \qquad \beta_C = \xi_+ \cdot k = \frac{1}{2} \int_{\gamma_+} \delta g_{ab} k^a k^b d\sigma_+,
\ee
where the integrals are taken over the (incomplete) null geodesics~$\gamma_\pm$ connecting~$p$ to the past and future tips of the boundary causal diamond~$R$, and the affine parameters are associated to the particular normalization of~$k^a$ and~$\ell^a$ as~$k^a = (\partial_{\sigma_+})^a$,~$\ell^a = (\partial_{\sigma_-})^a$ (we have also set the boundary conditions so that~$\gamma_\pm$ end at the future and past tips of the causal diamond~$R$; see~\cite{EngFis16} for more details).  The deviation vector field~$\eta_C^a$ is therefore obtained everywhere on~$\Hcal$ by computing these integrals along \textit{all} the null generators of~$\Pcal^\pm$.

\subsection{A Constraint from CWI}

We may now assemble these results.  CWI requires that~$C_R$ lie on the side of~$X_R$ closer to~$R$; since the deviation vector between the perturbed HRT surface and the perturbed causal rim is just the difference~$\eta_C^a - \eta_E^a$, the infinitesimal statement of CWI is that the components of~$\eta_C^a - \eta_E^a$ in the~$\{k^a, \ell^a\}$ basis must both be positive.  From~\eqref{eq:extremalalphabeta} and~\eqref{eq:causalalphabeta}, we thus conclude that for each~$p \in \Hcal$, we must have
\begin{subequations}
\label{eqs:CWI}
\begin{align}
\frac{1}{2} \int_{\gamma_+(p)} \delta g_{ab} k^a k^b \, d\sigma_+ &\geq \int_\Hcal G(p,p') k^a \left[h^{bc} \delta \Gamma_{abc}(p') - 4 G_N \hbar \, \frac{\Dcal \delta S_\mathrm{out}}{\Dcal \Sigma^a(p')} \right] \bm{\eps}(p'), \\
\frac{1}{2} \int_{\gamma_-(p)} \delta g_{ab} \ell^a \ell^b \, d\sigma_- &\geq \int_\Hcal G(p,p') \ell^a \left[h^{bc} \delta \Gamma_{abc}(p') - 4 G_N \hbar \, \frac{\Dcal \delta S_\mathrm{out}}{\Dcal \Sigma^a(p')} \right] \bm{\eps}(p'),
\end{align}
\end{subequations}
where we have replaced~$s_k$ and~$s_\ell$ with their explicit expressions using~\eqref{subeq:sdef}, and the integrals on the left-hand side are taken over the null geodesics~$\gamma_\pm(p)$ which intersect at~$p$.

In order for a linearized perturbation of the vacuum to satisfy CWI, these constraints must be obeyed on \textit{every} Rindler horizon~$\Hcal$: they thus constrain such a perturbation everywhere.  It is worth noting that this constraint is physical, since~\eqref{eqs:CWI} are gauge-invariant in the sense that they are unaffected by an infinitesimal diffeomorphism~$\delta g_{ab} = 2\grad_{(a} \zeta_{b)}$ (with~$\zeta_a$ falling off sufficiently fast asymptotically).  They are, however, gauge-fixed in the sense that they require choosing a basis~$\{k^a, \ell^a\}$ in which the twist~$\chi_a$ vanishes (though they do not depend on the overall normalization of this basis, i.e.~taking~$k^a \to c k^a$ for some number~$c$ leaves them unchanged).

Finally, let us highlight two special cases.  First, consider taking the region~$R$ to be very large, so that it wraps almost all the way around the AdS boundary.  Then the corresponding Rindler horizon~$\Hcal$ moves out towards the boundary, and (assuming sufficient falloff of~$\delta g_{ab}$ and~$\Dcal \delta S_\mathrm{out}/\Dcal \Sigma^a$) the right-hand sides of~\eqref{eqs:CWI} vanish, yielding only the condition~$\int_\gamma \delta g_{ab} k^a k^b d\sigma \geq 0$ along every \textit{complete} null geodesic~$\gamma$ of the background pure AdS.  This is precisely the so-called boundary causality condition (BCC) of~\cite{EngFis16}, which is dual to microcausality in the bulk CFT, and was shown to be intimately connected to the chaos bound in~\cite{AfkHar17}.   The condition~\eqref{eqs:CWI} can therefore be interpreted as a ``fine-grained'' CFT chaos bound, potentially involving some entanglement entropy contribution\footnote{Also of interest is the fact that the near-boundary version of CWI implies the quantum half averaged null energy condition~\cite{AkeKoe16}.}.

Second, if we were to work in the regime of quantum field theory on a \textit{fixed} pure AdS spacetime, the quantum extremal surface would still be different from the classical extremal surface (which still coincides with the causal wedge), since it is a stationary point of a different functional.  While fixing the background geometry is not consistent with the full equations of motion (which enforce backreaction that comes in at the same order as the $S_{\mathrm{out}}$ contribution), the GSL nevertheless \textit{still} implies that the causal wedge is contained inside the (quantum) entanglement wedge~\cite{EngWal14}.  We can prove a partial converse to this: if causal wedge inclusion holds perturbatively around pure AdS, then we get a smeared GSL on the Rindler horizon. Explicitly, from~\eqref{eqs:CWI} we obtain the constraint
\be
\label{eq:smearedGSL}
\int_\Hcal G(p,p') \, \delta \Theta^{(k)}(p') \geq 0
\ee
and likewise for~$\Theta^{(\ell)}$ (where the quantum expansion~$\Theta^{(k)}$ is as defined in~\eqref{eq:Theta}).  Now, it is straightforward to check (either from the definition~\eqref{eq:Greens} or from the explicit expressions in Appendix~\ref{app:Greens}) that~$G(p,p') \geq 0$ for all~$p$,~$p'$.  Since the GSL requires~$\Theta^{(k)}$ to be non-negative on slices of any causal horizon with (future-pointing) generator~$k^a$, we indeed get a smeared GSL for quantum fields on Poincar\'e horizons of pure AdS.  While this result is already known (since the GSL is known to hold along Killing horizons~\cite{Wal11}), this derivation provides some clarification on a potential converse to the results that assume the GSL or QFC and derive bulk consistency conditions.

\section{Entanglement Wedge Nesting}
\label{sec:EWN}

To further illustrate the power of this formalism at work, we now exploit it to explore the implications of entanglement wedge nesting for the bulk geometry.  As a warmup, we first proceed in the converse direction by deriving (a local version of) EWN from the NCC.  The proof is obtained by using maximum principles for so-called cooperative elliptic systems, and it is free of many of the subtleties that feature in the geometric proofs~\cite{Wal12,EngWal13} due to complications arising from caustics and non-local intersections of null geodesics; we are also able to relax the assumptions used in~\cite{EngWal13}.  We then proceed in the desired direction: assuming EWN, we derive constraints on the bulk geometry.  For pedagodical clarity, we first consider the purely classical case~($\hbar = 0$) before including quantum corrections.

\subsection{Warmup: EWN from the NCC}
\label{subsec:EWNfromNCC}

There are currently two versions of the proof that the NCC implies EWN, which use two different sets of assumptions.  The first, from~\cite{Wal12}, assumes minimality and homology of the surfaces as well as a notion of ``stability'' for HRT surfaces (which is related, but different, to the notions of stability discussed in Section~\ref{sec:stability}); the second, from~\cite{EngWal13}, proves nesting for extremal surfaces that are part of a smoothly deformable family satisfying some global conditions.  The result we will present makes no global assumptions like those made in~\cite{EngWal13}.  We also do not assume minimality or homology, but consequently we are only able to prove local nesting -- we make no claims about the situation in which the HRT surface jumps due to a phase transition realized by the minimality constraint.  Our proof is quite expedient thanks to the fact that the extremal deviation equation~\eqref{eq:unsourcedJacobi} is an eigenvalue equation for an elliptic differential operator, and such operators have a number of useful properties (see e.g.~\cite{evans10, GilbargTrudinger, ProtterWeinberger} for introductory textbooks on scalar elliptic operator theory).  Most useful to our present purposes are the so-called \textit{maximum and minimum principles} for elliptic PDEs, which we now review.

To develop some intuition, we remind the reader of the standard minimum principle for ordinary differential equations (ODEs).  Consider an ordinary (single variable) differential inequality on some open connected interval $U$ of $\mathbb{R}$:
\be
\label{eq:ODE}
-u''(x) + b(x) u'(x)+ c(x) u(x) \geq 0,
\ee
where $u(x)$ is twice differentiable on $U$ and~$b(x)$ and~$c(x)$ are known functions on~$U$.  If $c(x) > 0$ everywhere in~$U$, then $u(x)$ must be nonnegative at a local minimum in~$U$ (since there~$u'(x) = 0$ and~$u''(x) \geq 0$).  It follows that if~$u(x)$ is non\textit{positive} at a local minimum in~$U$, it must in fact be constant and vanishing everywhere in~$U$.  This latter statement extends to the case where the strict inequality on~$c(x)$ is relaxed: if~$c(x) \geq 0$ everywhere in~$U$ and~$u(x)$ is nonpositive at a local minimum in~$U$, it must be a constant function.  Consequently, the sign of~$u(x)$ at the boundary~$\partial U$ fixes the sign of~$u(x)$ everywhere in~$U$: if~$u(x)|_{\partial U} \geq 0$, then~$u(x) \geq 0$ \textit{everywhere} on~$U$ (or else it would necessarily have a negative local minimum, violating~\eqref{eq:ODE}).  This is the content of the minimum principle for ODEs; a similar maximum principle is obtained by flipping the sign of the ineqality in~\eqref{eq:ODE}.

The relevance of the minimum principle to the problem of EWN is clear: we would like to constrain the sign of the components of~$\eta^a$ in the null basis~$\{k^a, \ell^a\}$ without actually solving for them.  The system of equations~\eqref{eqs:scalarcodimtwo} governing these components, however, has three features that must be dealt with: it is \textit{(i)} a \textit{system} of \textit{(ii)} \textit{partial} differential equations, and \textit{(iii)} not all the zero-derivative terms need have definite sign on an arbitrary surface~$\Sigma$.  That the minimum principle for ODEs extends to scalar PDEs is very well-established, addressing item~\textit{(ii)}, and it turns out that the condition~$c \geq 0$ on the zero-derivative term can be replaced with the existence of a supersolution, addressing item~\textit{(iii)}.  Surprisingly, even item~\textit{(i)} is addressed by the existence of minimum principles for certain \textit{systems} of elliptic differential equations, thus allowing us to apply the minimum principle to the system~\eqref{eqs:scalarcodimtwo}.  We will present these extensions in steps, first addressing~\textit{(ii)} and~\textit{(iii)} by providing the minimum (and maximum) principle for (scalar) PDEs, followed by the minimum principle for so-called cooperative elliptic systems of PDEs.  It will turn out that whenever the NCC is satisfied, the system~\eqref{eqs:scalarcodimtwo} is cooperative, thereby allowing us to prove a version of EWN from the NCC.

To state the minimum (and maximum) principle for scalar PDEs, let~$U$ be some open connected domain of~$\mathbb{R}^n$ and consider a linear scalar differential operator
\be
\label{eq:ellipticop}
L = -\sum_{\alpha,\beta = 1}^n h^{\alpha\beta} \partial_\alpha \partial_\beta + \sum_{\alpha = 1}^n b^\alpha \partial_\alpha + c,
\ee
where~$h^{\alpha\beta}$, $b^\alpha$, and $c$ are all at least twice-differentiable on~$U$ and once-differentiable on~$\partial U$.  If the~$h^{\alpha\beta}$ are the components of a positive-definite matrix, then~$L$ is elliptic, and the following statement holds:
\begin{thm}
\label{thm:scalarmax}
\textbf{Maximum and minimum principle for scalar PDEs.}  Let~$L$ be an elliptic operator as in~\eqref{eq:ellipticop}.  Assume one of the following two statements is true:
\begin{itemize}
	\item $c \geq 0$ on~$U$.
	\item There exists a positive strict supersolution~$u^+$ on~$U \cup \partial U$, i.e.~a function which is twice-differentiable on~$U$ and once-differentiable on~$\partial U$ such that~$u^+ \geq 0$ everywhere,~$u^+$ is nonzero somewhere, and~$L u^+ \geq 0$ everywhere.
\end{itemize}
Let~$u$ be any function which is twice-differentiable on~$U$ and once-differentiable on $\partial U$.  If $Lu \geq 0$ ($Lu \leq 0$) and $u$ has a nonpositive minimum (nonnegative maximum) on $U$ (at an interior point, since~$U$ is open), then $u$ is a constant function on $U$.
\end{thm}
The~$c \geq 0$ version of this theorem is the standard maximum and minimum principle for PDEs; the alternative version invoking the existence of the supersolution~$u^+$ was formulated indirectly when~$U$ is a compact manifold by~\cite{AndMar05,AndMar07} (it essentially follows by combining Definition~5.1, Proposition~5.1, and Lemma~4.2 of~\cite{AndMar07}), though as we will see it is a special case of a more general result which applies to systems of elliptic PDEs.  To introduce this more general result -- and thereby finish addressing all three items~\textit{(i)},~\textit{(ii)}, and~\textit{(iii)} listed above -- we now introduce \textit{cooperative elliptic systems}~\cite{Swe92}:
\begin{defn}
\label{def:cooperative}
\textbf{Cooperative elliptic systems.} Consider some open domain~$U$ of~$\mathbb{R}^n$.  A linear system of differential equations on~$U$ for the~$m$ functions~$u_i$,~$i = 1, \ldots, m$ is said to be a \textit{cooperative elliptic system} if it can be written in the form
\be
\label{eq:cooperative}
\left[
\begin{pmatrix}
L_1 & 0 & \cdots & 0 \\ 0 & L_2 & \cdots & 0 \\ \vdots & \vdots & \ddots & \vdots \\ 0 & 0 & \cdots & L_m
\end{pmatrix} -
\begin{pmatrix}
0 & H_{12} & \cdots & H_{1m} \\ H_{21} & 0 & \cdots & H_{2m} \\ \vdots & \vdots & \ddots & \vdots \\ H_{m1} & H_{m2} & \cdots & 0
\end{pmatrix}
\right]
\begin{pmatrix}
u_1 \\ u_2 \\ \vdots \\ u_m
\end{pmatrix} = 
\begin{pmatrix}
f_1 \\ f_2 \\ \vdots \\ f_m
\end{pmatrix},
\ee
or~$(L - H) \vec{u} = \vec{f}$ for shorthand, where the~$L_i$ are elliptic operators defined as in~\eqref{eq:ellipticop} (i.e.~the coefficients~$(h_i)^{\alpha\beta}$,~$(b_i)^\alpha$, and~$c_i$ are allowed to be different for each~$i$) and the coefficients~$H_{ij}$ are all non-negative on~$U \cup \partial U$.  Moreover, a cooperative elliptic system is said to be \textit{fully coupled} if~$\{1,\cdots, m\}$ cannot be split into two disjoint nonempty sets~$A_1$ and~$A_2$ such that~$H_{ij} = 0$ everywhere in~$U$ for all~$i \in A_1$, $j \in A_2$.
\end{defn}
The system of equations~\eqref{eqs:scalarcodimtwo} for the components of the deviation vector in the~$\{k^a, \ell^a\}$ basis takes the form~\eqref{eq:cooperative}, while the NCC implies that~$Q_{kk}$ and~$Q_{\ell\ell}$ are non-negative, so the system~\eqref{eqs:scalarcodimtwo} is a cooperative elliptic system under the assumption of the NCC.  The promised minimum principle for such systems is as follows:
\begin{thm}
\label{thm:maximumcooperative}
\textbf{Minimum principle for cooperative systems}~\cite{Swe92}.  Consider a fully-coupled cooperative elliptic system as defined above which additionally obeys the following properties:
\begin{itemize}
	\item There exists a positive strict supersolution of the homogeneous version of~\eqref{eq:cooperative}; that is, there exists a	vector~$\vec{u}\,^+$ of (sufficiently smooth) functions such that for all~$i$,~$u^+_i \geq 0$ and~$(L\vec{u}\,^+-H\vec{u}\,^+)_i \geq 0$ everywhere on~$U$, and either~$\vec{u}\,^+$ is nonzero somewhere on $\partial U$ or $(L-H)\vec{u}\,^+$ is nonzero somewhere in~$U$;
	\item For all~$i$,~$f_i \geq 0$ on~$U$.
\end{itemize}
Then for any (sufficiently smooth)~$u_i$ which solve~\eqref{eq:cooperative} with~$u_i \geq 0$ on~$\partial U$, the~$u_i$ are either all positive everywhere on~$U$ or they vanish everywhere on~$U$.
\end{thm}

Theorem~\ref{thm:maximumcooperative} is the minimum principle we need in order to prove EWN from the NCC using the system of equations~\eqref{eqs:scalarcodimtwo}.  We may therefore finally give the advertised proof of continuous nesting of codimension-two spacelike extremal surfaces:
\begin{prop}
\label{prop:NCCnesting}
\textbf{Continuous Nesting of Extremal Surfaces}.  Let~$\Sigma(\lambda)$ be a continuous one-parameter family of codimension-two connected extremal surfaces anchored to a family of causal diamonds~$B(\lambda)$ on the boundary~$\partial M$ of an asymptotically locally AdS spacetime~$(M, g_{ab})$.  Assume that~$g_{ab}$ obeys the null curvature condition~$R_{ab} k^a k^b \geq 0$ for all null~$k^a$, and moreover assume that for each~$\lambda$, there exists an arbitrarily small deformation of~$\Sigma(\lambda)$ in a spacelike direction towards~$B(\lambda)$ which is nonvanishing somewhere on the boundary~$\partial \Sigma(\lambda)$ and whose outgoing null expansions are everywhere nonnegative.  Then if the family~$B(\lambda)$ is nested in the sense that~$B(\lambda_1)  \subset B(\lambda_2)$ whenever~$\lambda_1 > \lambda_2$, so are the~$\Sigma(\lambda)$ in the sense that for each~$\lambda$, the deviation vector field~$\eta^a$ on~$\Sigma(\lambda)$ is everywhere achronal, nonvanishing, and pointing towards~$B(\lambda)$.
\end{prop}

\begin{proof}
We proceed by contradiction: assume that there is some critical~$\lambda_*$ and surface~$\Sigma(\lambda_*)$, which without loss of generality we take to be~$\lambda_* = 0$ and~$\Sigma \equiv \Sigma(\lambda = 0)$, on which~$\eta^a$ is not everywhere achronal, nonvanishing, and pointing towards~$B(\lambda = 0)$.  Decomposing~$\eta^a$ in the usual outwards-pointing null basis~$\{k^a, \ell^a\}$ of the normal bundle of~$\Sigma$ as~$\eta^a = \alpha k^a + \beta \ell^a$, this condition implies that somewhere on~$\Sigma$, at least one of~$\alpha$,~$\beta$ is negative or both are zero.  We take the basis~$\{k^a, \ell^a\}$ to be nonzero at~$\partial M$ with respect to any conformal completion that renders~$\partial M$ finite (this implies that~$k^a$ and~$\ell^a$ are divergent at~$\partial \Sigma$ with respect to~$g_{ab}$).  Consequently (since the components of~$\eta^a$ in some coordinate chart are the coordinate displacement of~$\Sigma$ under the perturbation generated by~$\eta^a$), the condition that the diamonds~$B(\lambda)$ are nested implies that neither~$\alpha$ nor~$\beta$ is anywhere negative at~$\partial \Sigma$ and that they cannot both be zero everywhere there.  In particular, it follows that they cannot both be zero everywhere on~$\Sigma$, and thus at least one must become negative somewhere.

Since~$\eta^a$ is a deviation vector along a family of extremal surfaces, on~$\Sigma$~$\alpha$ and~$\beta$ obey~\eqref{eqs:scalarcodimtwo} with sources turned off, which we write here as
\begin{subequations}
\be
\label{eq:codimtwocooperative}
\left[\begin{pmatrix} J_+ & 0 \\ 0 & J_- \end{pmatrix} - \begin{pmatrix} 0 & Q_{\ell\ell} \\ Q_{kk} & 0 \end{pmatrix}\right] \begin{pmatrix} \alpha \\ \beta \end{pmatrix} = \begin{pmatrix} 0 \\ 0 \end{pmatrix}
\ee
with
\be
J_\pm \equiv -D^2 \mp 2\chi^a D_a - \left(|\chi|^2 \pm D_a \chi^a + Q_{k\ell} \right).
\ee
\end{subequations}
As mentioned above, the NCC implies that~$Q_{kk}$ and~$Q_{\ell\ell}$ are non-negative, and therefore this is a cooperative elliptic system.  Finally, the assumption that there exists a deformation of~$\Sigma(\lambda)$ in an outwards direction that renders its outgoing null expansions nonnegative can be reinterpreted using~\eqref{eq:thetadot}: it is equivalent to the statement that there exists some~$\nu^a$ on~$\Sigma$, with~$\nu^a$ everywhere spacelike, pointing towards~$B$, and nonvanishing somewhere on~$\partial \Sigma(\lambda)$, such that~$k^a J \nu_a \geq 0$ and~$\ell^a J \nu_a \geq 0$.  Decomposing~$\nu^a = \tilde{\alpha} k^a + \tilde{\beta} \ell^a$, we have that~$\tilde{\alpha}$ and~$\tilde{\beta}$ are everywhere non-negative, that both are nonzero somewhere on~$\partial \Sigma(\lambda)$, and that they satisfy
\be
J_+ \tilde{\alpha} - Q_{\ell\ell} \tilde{\beta} \geq 0, \qquad J_- \tilde{\beta} - Q_{kk} \tilde{\alpha} \geq 0.
\ee
In other words, the vector~$\vec{\nu} = (\tilde{\alpha}, \tilde{\beta})$ is a positive strict supersolution of~\eqref{eq:codimtwocooperative}.

To complete the proof, we must deal with two separate cases.

\paragraph{Case 1:} One of~$Q_{kk}$ or~$Q_{\ell\ell}$ is everywhere zero on~$\Sigma$.  Taking without loss of generality~$Q_{kk} = 0$ everywhere, this implies that~$\beta$ decouples from~$\alpha$, as we have simply~$J_- \beta = 0$.  Since~$\beta \geq 0$ on~$\partial \Sigma$, and since~$\tilde{\beta}$ is a positive strict supersolution (since for~$Q_{kk} = 0$ it satisfies~$J_i \tilde{\beta} \geq 0$), we may apply the minimum principle for cooperative elliptic systems in the case~$m = 1$ to conclude that~$\beta$ must be everywhere positive or vanish everywhere.  Consequently,~$\alpha$ obeys~$J_+ \alpha = Q_{\ell\ell} \beta \geq 0$, and~$\tilde{\alpha}$ obeys~$J_+\tilde{\alpha} \geq Q_{\ell\ell} \tilde{\beta} \geq 0$, so it is a positive strict supersolution.  Again applying the~$m = 1$ case of the minimum principle for cooperative elliptic systems, we conclude that~$\alpha$ must be everywhere positive or vanish everywhere.  But by assumption, one of~$\alpha$ or~$\beta$ must become negative somewhere, and therefore we have a contradiction.

\paragraph{Case 2:} Neither of~$Q_{kk}$ or~$Q_{\ell\ell}$ is everywhere zero.  Consequently, the system~\eqref{eq:codimtwocooperative} is a fully coupled cooperative elliptic system, for which we have already established the existence of a positive strict supersolution.  Thus we may apply the~$m = 2$ case of the minimum principle for cooperative elliptic systems to conclude that either~$\alpha$ and~$\beta$ are both everywhere positive or both vanish everywhere.  But since by assumption at least one must be negative somewhere, this is a contradiction.
\end{proof}

We emphasize that besides the NCC, the assumptions of this proof are very weak, the main one being the existence of a spacelike perturbation of each~$\Sigma(\lambda)$ which is nonvanishing somewhere on~$\partial \Sigma(\lambda)$ and whose expansions are nonnegative; this can be interpreted as a weaker, local notion of the smooth deformability criterion of~\cite{EngWal13} (we of course make no assumptions about global minimality).  To see this explicitly, note that one (but not the only) way of achieving such a perturbation is to find an extremal surface attached to a ``shrunken'' boundary domain~$B(\lambda)$ which is everywhere spacelike to~$\Sigma(\lambda)$ on the side of~$B(\lambda)$.  In other words, if there exists an extremal perturbation of an extremal surface~$\Sigma$ which is nested, then all such extremal perturbations must also be nested; this is precisely a local version of the ``deformable family'' of extremal surfaces invoked in~\cite{EngWal13}.

Furthermore, let us note that we needed this deformation to be nonvanishing somewhere on~$\partial \Sigma(\lambda)$ only to ensure that~$\tilde{\alpha}$ and~$\tilde{\beta}$ are independently supersolutions to the equations~$J_-\beta = 0$ and~$J_+ \alpha = 0$ even when one (or both) of~$Q_{kk}$ and~$Q_{\ell\ell}$ are everywhere-vanishing.  However, generically~$Q_{kk}$ and~$Q_{\ell\ell}$ should not vanish everywhere, in which case we could also invoke Theorem~\ref{thm:maximumcooperative} with a deformation of~$\Sigma$ that \textit{fixes} the boundary~$\partial \Sigma$ but renders both outgoing expansions nonnegative, with at least one strictly positive somewhere.  In other words, if both~$Q_{kk}$ and~$Q_{\ell\ell}$ are not everywhere-vanishing, then EWN is guaranteed by the existence of perturbations of the~$\Sigma(\lambda)$ that render them ``normal'' in the sense of having positive outwards expansions.

\subsection{A Constraint from EWN}
\label{subsec:NCC}

We have re-established, using this formalism, the fact that the bulk NCC enforces a type of EWN.  But since EWN is enforced from a fundamental principle of the boundary field theory, is would be desirable to proceed in the converse direction: that is, does EWN tell us anything about the bulk geometry?  Here we show that the answer is yes.  To obtain the result, note that since we will only consider small perturbations of extremal surfaces, our result only assumes \textit{extremal} wedge nesting (rather than entanglement wedge nesting, which requires the extremal surfaces in question to be HRT surfaces).  With this caveat in mind, we first show the following:

\begin{prop}
Consider two deviation vector fields~$\eta_1^a$,~$\eta_2^a$ on a boundary-anchored extremal surface~$\Sigma$ along independent one-parameter families of extremal surfaces.  Let~$k^a$ be any null normal to~$\Sigma$ which is nonvanishing at~$\partial \Sigma$\footnote{More precisely, when we say a vector field~$k^a$ is nonvanishing at~$\partial \Sigma$, we mean nonvanishing in any conformal compactification; that is, for any~$\Omega$ which vanishes at~$\partial \Sigma$ such that~$\Omega^2 g_{ab}$ is smoothly extendable to~$\partial \Sigma$ and nondegenerate there,~$\Omega^{-1} k^a$ should also be smoothly extendable to~$\partial \Sigma$ and nowhere-vanishing there.}, and choose~$\eta_{1,2}^a$ such that at~$\partial \Sigma$, the following hold:
\begin{itemize}
	\item $\eta_1^a|_{\partial \Sigma}$ is everywhere nonvanishing, achronal, and points into the extremal wedge of~$\Sigma$, and is nowhere proportional to~$k^a$;
	\item $\eta_2^a|_{\partial \Sigma}$ is everywhere nonvanishing and proportional to~$k^a$.
\end{itemize}
Without loss of generality, also take~$\eta_{1,2}^a$ normal to~$\Sigma$.  Then if extremal wedge nesting holds,~$\eta_2^a$ may be decomposed as
\be
\eta_2^a = w k^a + v \eta_1^a,
\ee
and either~$w Q_{kk} > 0$ somewhere or~$w Q_{kk} = 0$ everywhere.
\end{prop}

\begin{proof}
Since~$\eta_1^a$ is a deviation vector along a family of extremal surfaces, we have
\be
J (\eta_1)_a = 0 \Rightarrow J_{k,\eta_1} \mathbb{I} = 0,
\ee
where~$\mathbb{I}$ denotes the function which is everywhere unity and the differential operator~$J_{k,\eta_1}$ is defined as in~\eqref{eq:Lij}.  This implies that the zero-derivative term in~$J_{k,\eta_1}$ vanishes.

Next, since both~$\eta_1^a$ and~$\eta_2^a$ are achronal and nonvanishing on the boundary and point into the extremal wedge of~$\Sigma$, extremal wedge nesting implies that this must be true in the bulk as well.  Moreover, since~$\eta_1^a$ is never proportional to~$k^a$ on the boundary, this must be true in the bulk as well, and thus~$k^a$ and~$\eta_1^a$ are everywhere linearly independent; this guarantees that~$\eta_2^a$ can be decomposed as stated.  When this decomposition is inserted into the Jacobi equation~$J(\eta_2)_a = 0$, we obtain (after contracting with~$k^a$)
\be
\label{eq:Letapsi}
J_{k,\eta_1} v = w Q_{kk}.
\ee
Since~$\eta_2^a$ is proportional to~$k^a$ on the boundary, we have~$v|_{\partial \Sigma} = 0$.  Extremal wedge nesting also implies that~$v \geq 0$ everywhere (otherwise~$\eta_2^a$ would point out of the extremal wedge).  Thus~$v$ must have a non-negative maximum somewhere in the interior of~$\Sigma$.

Now we proceed by contradiction: assume that~$w Q_{kk} \leq 0$ everywhere with the inequality holding strictly somewhere.  From~\eqref{eq:Letapsi}, this implies that~$v$ cannot be everywhere zero.  But since~$J_{k,\eta_1}$ is a uniformly elliptic operator and its zeroth-order piece vanishes, we may apply the maximum principle (Theorem~\ref{thm:scalarmax}): the fact that~$J_{k,\eta_1} v = w Q_{kk} \leq 0$, that the zero-derivative term in~$J_{k,\eta_1}$ vanishes, and that~$v$ has a non-negative maximum in~$\Sigma$ implies that~$v$ must be constant.  But as we already established, if~$v$ were constant it would have to vanish everywhere, which is not permitted by assumption.  Thus we have reached a contradiction. 
\end{proof}

We now wish to argue that under some appropriate genericity condition, it should be possible to find a choice of deviation vectors~$\eta_{1,2}^a$ as defined above such that~$w > 0$ everywhere.  To see this, in Figure~\ref{fig:normal} we illustrate the two-dimensional normal space~$T^\perp_p M$ at some point~$p \in \Sigma$.  Entanglement wedge nesting requires that both~$\eta_1^a$ and~$\eta_2^a$ point into the right wedge of the figure, and the sign of~$w$ at~$p$ is determined by whether~$\eta_1^a$ falls above or below the dotted line spanned by~$\eta_2^a$.  But since~$\eta_2^a$ is proportional to~$k^a$ at~$\partial \Sigma$, we expect that, say, taking~$\eta_1^a$ proportional to~$\ell^a$ at~$\partial \Sigma$ should keep~$\eta_1^a$ below the dotted line, and hence~$w > 0$, everywhere on~$\Sigma$.  This is certainly true for spacetimes that are sufficiently small (but nonperturbative) deformations of pure AdS, but it seems reasonable to expect that it should be true much more broadly as well.

\begin{figure}[t]
\centering
\includegraphics[page=15,width=0.23\textwidth]{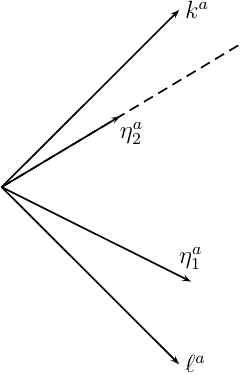}
\caption{Here we illustrate the normal space~$T^\perp_p M$ at some point~$p \in \Sigma$.  EWN requires that both~$\eta_1^a$ and~$\eta_2^a$ point into the right wedge, while the sign of~$w$ is determined by whether~$\eta_1^a$ falls above or below the dashed line marking the span of~$\eta_2^a$.}
\label{fig:normal}
\end{figure}

Consequently, we conclude (under appropriate genericity assumptions) that extremal wedge nesting implies that on any extremal surface, we must have~$Q_{kk} > 0$ somewhere or~$Q_{kk} = 0$ everywhere.  Recall from~\eqref{eq:Qkk} that~$Q_{kk}$ is the ``right-hand side'' of the Raychaudhuri equation, which determines the focusing of null geodesics; we have therefore found that extremal wedge nesting imposes that any defocusing of null geodesics fired off of some region of an HRT surface~$\Sigma$ must be accompanied by focusing of null geodesics fired elsewhere off of~$\Sigma$, as shown in Figure~\ref{fig:QEI}.  Interpreting focusing and defocusing in some rough sense as due to a local null energy, we might heuristically rephrase this statement as enforcing that any negative local null energy on~$\Sigma$ must be accompanied by some positive local null energy elsewhere; this is reminiscent of certain spacelike quantum energy inequalities (and some more recent results from modular theory~\cite{BlaCas17}) which require any negative energy on a Cauchy slice to be accompanied by positive energy elsewhere~\cite{Fla97,FewHol04,Few12}.  These results are restricted to quantum field theory on fixed Minkowski spacetime in two dimensions (and in fact, a generalization to higher dimensions is known \textit{not} to exist, at least for the massless minimally coupled scalar field on four-dimensional Minkowski spacetime~\cite{ForHel02}).  On the other hand, our result applies to general bulk spacetimes which are consistently coupled to (classical) matter fields.

\begin{figure}[t]
\centering
\includegraphics[page=16]{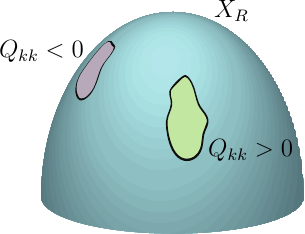}
\caption{In a general spacetime, we argue that the quantity~$Q_{kk} = R_{kk} + \sigma_k^2$ cannot be everywhere-negative on an HRT surface~$X_R$.  Thus if any null geodesics fired off of~$X_R$ defocus due to a region of negative~$Q_{kk}$ (the purple region), there must be null geodesics fired from elsewhere on~$X_R$ (the green region) which strictly focus due to positive~$Q_{kk}$.}
\label{fig:QEI}
\end{figure}

Finally, the result as it currently stands does not immediately generalize to quantum extremal surfaces due to the nonlocal contributions from the integral of the second derivative of~$S_{\mathrm{out}}$ in the equation of quantum extremal deviation~\eqref{eq:unsourcedquantumextremaldeviation}. However, there is some literature in mathematics on such integro-differential elliptic equations (see e.g.~\cite{garroni95}), and it is reasonable to expect an analogous statement that permits the QFC to be false on an HRT only if it is satisfied somewhere else on the same HRT surface.

\subsection{Special Case: Perturbations of Pure AdS}
\label{subsec:pureAdS}

We have therefore argued that~$Q_{kk}$ cannot be negative everywhere on an extremal surface, but is there a more explicit quantitative statement we can make?  Presumably for perturbations around any sufficiently nice solution (e.g. Schwarzschild-AdS), it is possible to explicitly obtain constraints from our formalism.  Here for simplicity we will focus on pure AdS, showing that spacetimes that are a small perturbation thereof, we can indeed obtain a quantitative constraint.  We begin with the classical statement:

\begin{prop}
\label{prop:classicalAdS}
Consider a spacetime~$(M,g_{ab})$ whose metric is a linear perturbation~$\delta g_{ab}$ of the metric~$\bar{g}_{ab}$ of pure AdS.  Let~$\Hcal$ be the bifurcation surface of any Rindler horizon of the pure AdS background, and let~$\bar{k}^a$ be any null normal to~$\Hcal$ (with respect to~$\bar{g}_{ab}$) which is parallel-transported along~$\Hcal$ (such a vector field necessarily exists per the discussion above~\eqref{eq:AdSRiemann}).    Also let~$\Sigma$ be the surface anchored to~$\partial \Hcal$ which is extremal with respect to~$g_{ab}$, with~$k^a = \bar{k}^a + \delta k^a$ its null normal.  Then if extremal wedge nesting holds,
\be
\label{eq:Rkkbound}
\int_\Sigma \delta R_{kk} \, \bm{\eps} \geq 0 + \Ocal(\delta^2)
\ee
whenever the integral is finite (with~$\delta R_{kk} = \delta(R_{ab} k^a k^b)$ the null-null component of the Ricci tensor of~$g_{ab}$ and~$\bm{\eps}$ the usual natural volume form on~$\Sigma$).
\end{prop}

In fact, before providing the proof of this statement let us briefly remark that the bound~\eqref{eq:Rkkbound} holds even when the integral is divergent, but such a case is not particularly insteresting since then the only nontrivial contribution to~\eqref{eq:Rkkbound} comes from the asymptotic behavior of the Ricci tensor, which in turn is simply related to the boundary metric and stress tensor (e.g.~by the Fefferman-Graham expansion).  It is also worth noting that~\eqref{eq:Rkkbound} is related to the positive energy theorems of~\cite{LasRab14,LasLin16,NeuSar18}, obtained there via entangement entropy inequalities rather than via subregion/subregion duality as we do here.

\begin{proof}
Consider solutions to the equation of extremal deviation~\eqref{eqs:scalarcodimtwo} with~$\eta^a$ corresponding to a deformation of~$\partial \Sigma$ along~$k^a$.  In the pure AdS background, we found in Section~\ref{sec:causal} that the components of~\eqref{eqs:scalarcodimtwo} for perturbations of a Rindler horizon~$\Hcal$ reduce to~\eqref{eq:pureAdSJacobi}, which for no metric perturbation become
\be
\label{eq:alphabetaAdS}
-D^2_\Hcal \bar{\alpha} + \frac{d-2}{l^2} \, \bar{\alpha} = 0, \qquad -D^2_\Hcal \bar{\beta} + \frac{d-2}{l^2} \, \bar{\beta} = 0,
\ee
where we continue to use overlines to denote objects evaluated in the pure AdS metric~$\bar{g}_{ab}$.  Requiring that~$\bar{\eta}^a$ correspond to a perturbation of the asymptotic boundary~$\partial \Hcal$ along~$\bar{k}^a$ implies that~$\bar{\beta} = 0$ everywhere and that~$\bar{\eta}^a$ must be finite and nonzero at~$\partial \Hcal$ in any conformal compactification of~$(M,\bar{g}_{ab})$.  This implies, in particular, that~$1/\bar{\alpha}$ is a defining function for a conformal compactification of~$\Hcal$ (that is, the conformally compactified metric~$\bar{g}_{ab}/\bar{\alpha}^2$ on~$\Hcal$ can be smoothly extended to~$\partial \Hcal$); it follows that~$\bar{\alpha}$ diverges at~$\Hcal$.

Now, in the perturbed spacetime~$g_{ab}$, we must linearize~\eqref{eqs:scalarcodimtwo} in~$\delta g_{ab}$\footnote{To avoid potential confusion, let us emphasize that here we are interested in the \textit{unsourced} equation of extremal deviation~\eqref{eq:unsourcedJacobi} which governs the deviation vector between extremal surfaces in the \textit{same} geometry~$\bar{g}_{ab} + \delta g_{ab}$; this is different from the \textit{sourced} equation of extremal deviation~\eqref{eq:sourcedgeneralJacobi}, which controls the deviation vector between an extremal surface in one geometry~$g_{ab}$ and an extremal surface in a perturbed geometry~$g_{ab} + \delta g_{ab}$.}.  To do so, first note that for a deviation vector which is proportional to~$k^a$ at~$\partial \Sigma$, we must have~$\alpha = \bar{\alpha} + \delta \alpha + \Ocal(\delta^2)$ and~$\beta = \delta \beta + \Ocal(\delta^2)$, with~$\delta \beta$ finite at~$\partial \Sigma$.  Since we also have that the extrinsic curvature~${\overline{K}^a}_{bc}$ of~$\Hcal$ vanishes, Simons' tensor for the perturbed surface~$\Sigma$ must be~$S_{ab} = K_{acd} {K_b}^{cd} = \Ocal(\delta^2)$, and hence~$Q_{kk} = \delta R_{kk} + \Ocal(\delta^2)$.  Thus from~\eqref{eqs:scalarcodimtwo}, we have to linear order in~$\delta g_{ab}$ that the perturbation~$\delta \beta$ obeys
\be
\label{eq:perturbedbeta}
-D^2_\Hcal \delta \beta + \frac{d-2}{l^2} \, \delta \beta = \bar{\alpha} \, \delta R_{kk},
\ee
where we emphasize that the Laplacian~$D^2_\Hcal$ is still the Laplacian on the \textit{unperturbed} Rindler horizon~$\Hcal$, and~$\bar{\alpha}$ is any solution to the unperturbed deviation equation~\eqref{eq:alphabetaAdS}.  (Also note that it doesn't matter whether we evaluate linearized equations like~\eqref{eq:perturbedbeta} on~$\Hcal$ or~$\Sigma$, since the difference will introduce subleading~$\Ocal(\delta^2)$ corrections.)

Next, divide~\eqref{eq:perturbedbeta} through by~$\bar{\alpha}$ to obtain
\be
\label{eq:deltabeta}
2\widetilde{\delta\beta} D^2_\Hcal \ln \bar{\alpha} - \overline{D}_a \left[\overline{D}^a \widetilde{\delta \beta} + 2 \widetilde{\delta \beta} \, \overline{D}^a \ln \bar{\alpha} \right] = \delta R_{kk},
\ee
where~$\widetilde{\delta \beta} \equiv \delta \beta/\bar{\alpha}$ and the second term on the left-hand side is a divergence taken on the unperturbed surface~$\Hcal$.  Integrating this equation over~$\Hcal$, we obtain
\be
\label{eq:intRkk}
\int_\Hcal \delta R_{kk} \, \bm{\eps} = 2\int_\Hcal \widetilde{\delta\beta} D^2_\Hcal \ln \bar{\alpha} \, \bm{\eps} - \int_{\partial \Hcal} \overline{N}^a \left[D_a \widetilde{\delta \beta} + 2 \widetilde{\delta \beta} \, D_a \ln \bar{\alpha} \right] \, ^\partial \! \bm{\eps},
\ee
where~$\overline{N}^a$ is the outward-pointing unit normal to~$\partial \Hcal$ in~$\Hcal$ and~$\, ^\partial \! \bm{\eps}$ is the natural volume element on~$\partial \Hcal$ (and the second integral should be interpreted in an appropriate limiting sense, since~$\partial \Hcal$ is an asymptotic boundary).  Now, since at~$\partial \Sigma$~$\eta^a$ is nowhere chronal and points everywhere into the same boundary causal diamond, extremal wedge nesting requires that this be true everywhere on~$\Sigma$ as well.  In particular, this requires that both~$\alpha$ and~$\beta$ be everywhere non-negative, enforcing that~$\widetilde{\delta \beta} \geq 0$.  Moreover, it is also easy to check that in the coordinates~\eqref{eq:AdS}, the function~$\bar{\alpha} = \cosh\chi$ solves~\eqref{eq:alphabetaAdS} and also satisfies~$D^2_\Hcal \ln \bar{\alpha} \geq 0$ (for~$d \geq 3$).  For this choice of~$\bar{\alpha}$, then, extremal wedge nesting implies that the first term on the right-hand side of~\eqref{eq:intRkk} is non-negative.

To get more control over the second integral in~\eqref{eq:intRkk}, note that since~$\bar{\alpha}$ diverges at~$\partial \Hcal$ and~$\delta \beta$ does not,~$\widetilde{\delta \beta}$ must vanish at~$\partial \Hcal$.  Since~$1/\bar{\alpha}$ vanishes at~$\partial \Hcal$, we must be able to write 
\be
\widetilde{\delta \beta} = b/\bar{\alpha}^p
\ee
for some~$p > 0$ and where we impose that at~$\partial \Hcal$,~$b$ is not identically zero but its normal derivative~$\widetilde{N}^a D_a b$ is, with~$\widetilde{N}^a \equiv \bar{\alpha} \overline{N}^a$ the unit normal vector to~$\partial \Hcal$ in~$\Hcal$ with respect to the compactified metric~$\tilde{g}_{ab} \equiv \bar{g}_{ab}/\bar{\alpha}^2$ (and for simplicity we assume~$p$ is just a number, though it is straightforward to generalize to the case where~$p$ varies along~$\partial \Hcal$).  Also note that since~$\widetilde{\delta \beta} \geq 0$, we have~$b \geq 0$.  Using this decomposition, the last term in~\eqref{eq:intRkk} becomes
\be
\label{eq:int1}
\lim_{\vareps \to 0} \int_{\partial \Hcal_\vareps} \frac{\widetilde{N}^a}{\bar{\alpha}^p} \left((2-p) b \, D_a\left(\frac{1}{\bar{\alpha}}\right) - \frac{1}{\bar{\alpha}} \, D_a b\right)\, ^\partial \! \bm{\eps},
\ee
where~$\partial \Hcal_\vareps$ is a cutoff surface characterized by some parameter~$\vareps$ such that as~$\vareps \to 0$,~$\partial \Hcal_\vareps \to \partial \Hcal$.  Now, because~$1/\bar{\alpha}$ vanishes at~$\partial \Hcal$ and~$\widetilde{N}^a$ is outward-pointing, we have~$\widetilde{N}^a D_a (1/\bar{\alpha}) < 0$ at~$\partial \Hcal$, and thus the first term above will be non-negative as long as~$p > 2$.  Moreover, since~$\widetilde{N}^a D_a b$ vanishes at~$\partial \Hcal$, the second term above is subleading as the cutoff is removed, so we can neglect it (as long as~$p \neq 2$).  Finally, by conformally compactifying the implied volume element in~\eqref{eq:int1} by rescaling it by a factor of~$\bar{\alpha}^{d-3}$, it straightforward to show that the entire integral will vanish in the limit~$\vareps \to 0$ if~$p > d-3$.  Thus we conclude that the entire integral will be non-negative as long as~$p > \min(2, d-3)$.

Thus for~$p > \min(2, d-3)$, we find that extremal wedge nesting implies that both integrals on the right-hand side of~\eqref{eq:intRkk} are non-negative, and hence the bound~\eqref{eq:Rkkbound} holds (where as mentioned above, we switch from integrating over~$\Hcal$ in~\eqref{eq:intRkk} to~$\Sigma$ in~\eqref{eq:Rkkbound} since the difference introduces subleading corrections).  In fact, it is straightforward to check that the first integral on the right-hand side of~\eqref{eq:intRkk} is \textit{finite} as long as~$p > d-3$, which in turn implies that the second integral vanishes.  Thus the integral in~\eqref{eq:Rkkbound} is finite if and only if~$p > d-3$; since this value lies in the range~$p > \min(2, d-3)$ for which~\eqref{eq:Rkkbound} holds, we have shown that the bound holds whenever the integral is finite.
\end{proof}

It's worth understanding in somewhat more physical terms what conditions are required to ensure the sufficient falloff of~$\delta R_{kk}$ to render the integral finite.  To that end, consider some Fefferman-Graham expansion of the near-boundary metric:
\be
ds^2 = \frac{l^2}{z^2}\left[\sum_{\mu,\nu = 1}^{d-1} \left(g^{(0)}_{\mu\nu} + z^q g^{(q)}_{\mu\nu} + \Ocal(z^{q+1})\right) dx^\mu \, dx^\nu + dz^2\right],
\ee
where~$q > 0$.  One can then show that if the above metric is asymptotically AdS (so that~$g^{(0)}_{\mu\nu}$ is conformally flat), the asymptotic falloff of~$R_{kk}$ is~$\Ocal(z^q)$, and therefore the integral in~\eqref{eq:Rkkbound} is finite as long as~$q > d-3$.  Now, for pure metric perturbations, the Einstein equation implies that~$q = d-1$, which is compatible with the constraint~$q > d-3$, and thus the integral in~\eqref{eq:Rkkbound} is finite for any purely gravitational perturbation.  In the presence of matter, however, finiteness of~\eqref{eq:Rkkbound} imposes nontrivial constraints: for instance, in the case of a massive bulk scalar field with leading behavior~$\phi = \Ocal(z^\Delta)$, the relevant part of the stress tensor goes like~$T_{kk} = \Ocal(z^{2\Delta})$, and therefore sufficient falloff requires (see e.g.~\cite{FisKel12})
\be
\Delta > \frac{d-3}{2}.
\ee
Now, the possible values of~$\Delta$ for a scalar field of mass~$m$ are~$\Delta_\pm = (d-1)/2 \pm \nu$, where~$\nu = \sqrt{(d-1)^2/4 + (ml)^2}$ is required to be non-negative by the Breitenlohner-Freedman bound; note therefore that~$\Delta_+ > (d-3)/2$ always, while~$\Delta_- > (d-3)/2$ requires~$\nu < 1$.  Thus we find that finiteness of the bound~\eqref{eq:Rkkbound} requires any operators with~$\nu \geq 1$ (when only standard quantization is permitted, in which the coefficient of the~$z^{\Delta_-}$ term is the source and the coefficient of the~$z^{\Delta_+}$ term is the response) must have no source turned on, while sources may be turned on for operators with~$0 \leq \nu < 1$ (where alternate quantization is allowed).

The generalization of Propostion~\ref{prop:classicalAdS} to include quantum effects is then quite straightforward:
\begin{prop}
\label{prop:quantumAdS}
Consider a spacetime~$(M,g_{ab})$ whose metric is a linear perturbation~$\delta g_{ab}$ of the metric~$\bar{g}_{ab}$ of pure AdS, with~$\delta g_{ab}$ of~$\Ocal(\hbar)$.  As above, let~$\Hcal$ be the bifurcation surface of any Rindler horizon of the pure AdS background and let~$\bar{\alpha}$ be any solution to~\eqref{eq:alphabetaAdS} with~$D^2_\Hcal \ln \bar{\alpha} \geq 0$.  Also let~$\Sigma$ be the surface anchored to~$\partial \Hcal$ which is quantum extremal with respect to the generalized entropy~$S_\mathrm{gen}$ in the perturbed spacetime~$(M,g_{ab})$, and again let~$k^a = \bar{k}^a + \delta k^a$ be its null normal, with~$\bar{k}^a$ as above.  Then if quantum extremal wedge nesting holds,
\be
\label{eq:smearedQFC}
\int_\Sigma \int_\Sigma \frac{\bar{\alpha}(p)}{\bar{\alpha}(p')} \frac{\Dcal^2 S_\mathrm{gen}}{\Dcal \Sigma^a(p) \Dcal \Sigma^b(p')} k^a(p) k^b(p') \, \epsb(p) \epsb(p') \leq 0 + \Ocal(\hbar).
\ee
\end{prop}

\begin{proof}
As for the proof of Proposition~\ref{prop:classicalAdS}, consider a deviation vector between quantum extremal surfaces which corresponds to a deformation of~$\partial \Sigma$ along~$k^a$.  In the pure AdS background and with~$\Ocal(\hbar)$ corrections turned off, we again have~$\Sigma = \Hcal$, with~$\bar{\beta} = 0$ and~$\bar{\alpha}$ a solution to~\eqref{eq:alphabetaAdS}.  Thus including corrections linear in~$\hbar$ we have~$\alpha = \bar{\alpha} + \delta \alpha + \Ocal(\hbar^2)$,~$\beta = \delta \beta + \Ocal(\hbar^2)$.

Now, decomposing~$\eta^a = \alpha k^a + \beta \ell^a$ and contracting the unsourced equation of quantum extremal deviation~\eqref{eq:unsourcedquantumextremaldeviation} with~$k^a$, we obtain
\begin{multline}
0 = -D^2_\Hcal \delta \beta(p) + \frac{d-2}{l^2} \, \delta \beta(p) -\bar{\alpha}(p) \, \delta R_{kk}(p) + \\ 4G_N \hbar \int_\Hcal \bar{\alpha}(p') \bar{k}^a(p) \bar{k}^b(p') \frac{\Dcal^2 S_\mathrm{out}}{\Dcal \Sigma^b(p') \Dcal \Sigma^a(p)} \, \epsb(p') + \Ocal(\hbar^2),
\end{multline}
which is just a quantum-corrected version of~\eqref{eq:perturbedbeta}.  Quantum extremal wedge nesting still requires~$\delta \beta \geq 0$; then dividing through by~$\bar{\alpha}(p)$ and integrating over~$\Hcal$, from the same arguments as in the proof of Proposition~\ref{prop:classicalAdS} we conclude that the first two terms are non-negative, and we obtain
\be
\label{eq:int2}
\int_\Hcal \left[-\delta R_{kk}(p) + 4G_N \hbar \int_\Hcal \frac{\bar{\alpha}(p')}{\bar{\alpha}(p)} k^a(p) k^b(p') \frac{\Dcal^2 S_\mathrm{out}}{\Dcal \Sigma^b(p') \Dcal \Sigma^a(p)} \, \epsb(p')\right] \epsb(p) \leq 0 + \Ocal(\hbar^2).
\ee
But~$\delta R_{kk}$ is related to the second functional derivative of the area: from the Raychaudhuri equation and properties~\ref{cond:scalar} and~\ref{cond:compatible} of the definition of the covariant functional derivative, we have that the functional derivative of the \textit{classical} expansion~$\theta^{(k)} = K_a k^a$ is\footnote{The letter~$\delta$ is playing double duty as a Dirac delta function~$\delta(p,p')$ and as a variation like~$\delta R_{kk}$; we assume context is sufficient to distinguish these two roles.}
\be
k^a(p) \frac{\Dcal \theta^{(k)}(p')}{\Dcal \Sigma^a(p)} = \delta(p,p') k^a \grad_a \theta^{(k)} = - \delta(p,p') \delta R_{kk} + \Ocal(\hbar^2),
\ee
where we noted that the square of the shear and expansions contribute at~$\Ocal(\hbar^2)$ in this perturbative setup.  But since~$K_a = \Dcal A/\Dcal \Sigma^a$, then taking~$k^a$ to be affinely-parametrized (so that~$k^a(p) \Dcal k^b(p')/\Dcal \Sigma^b(p) = \delta(p,p') k^a \grad_a k^b = 0$) we find
\be
k^a(p) k^b(p') \frac{\Dcal^2 A}{\Dcal \Sigma^a(p) \Dcal \Sigma^b(p')} = - \delta(p,p') \delta R_{kk} + \Ocal(\hbar^2),
\ee
which allows us to write
\be
-\int_\Hcal \delta R_{kk}(p) = \int_\Hcal \int_\Hcal \frac{\bar{\alpha}(p')}{\bar{\alpha}(p)} k^a(p) k^b(p') \frac{\Dcal^2 A}{\Dcal \Sigma^b(p') \Dcal \Sigma^a(p)} \, \epsb(p')\, \epsb(p).
\ee
Inserting this expression into~\eqref{eq:int2} and dividing through by an overall factor of~$4G_N \hbar$, we obtain the bound~\eqref{eq:smearedQFC}.
\end{proof}

Some comments are in order.  First,~\eqref{eq:smearedQFC} is uninteresting unless it is finite; the necessary conditions on the geometry were already examined in the classical context of Proposition~\ref{prop:classicalAdS}, but we must in addition ensure that the matter entropy~$S_\mathrm{out}$ is sufficiently well-behaved asymptotically.  Second, since~$k^a \Dcal S_\mathrm{gen}/\Dcal \Sigma^a$ is just the quantum expansion~$\Theta^{(k)}$ associated to~$k^a$, the bound~\eqref{eq:smearedQFC} is a smeared version of the quantum focusing conjecture (QFC), which states that~$k^a(p) \Dcal \Theta^{(k)}(p')/\Dcal \Sigma^a(p) \leq 0$ for all~$p$,~$p'$.  Indeed, as pointed out by~\cite{Lei17}, it is natural to define the QFC in terms of some smearing over~$\Sigma$; this picture is also natural from the perspective that functional derivatives are distributional and thus should always be interpreted as being smeared against test functions.  Note that under assumption of the usual semiclassical Einstein equation~$G_{ab} = 8\pi G_N \langle T_{ab}\rangle$ (recently derived in~\cite{HaeMin19} for certain classes of states from consistency of holographic entanglement entropy), the bound~\eqref{eq:smearedQFC} can, in fact, be derived from the quantum null energy condition~\cite{BouFis15}, proven to hold for free fields on Killing horizons by~\cite{BouFis15b}\footnote{We thank Don Marolf for bringing this point to our attention.}.  An aspect of the novelty of Proposition~\ref{prop:quantumAdS}, however, is in showing that the QFC is intimately tied to the consistency of subregion/subregion duality.  Moreover, since our derivation made no assumptions about the explicit form of the dynamics, it may be viewed as evidence in favor of this form of the semiclassical Einstein equation for a general class of states.

\section*{Acknowledgments}

It is a pleasure to thank R. Bousso, J. Camps, V. Chandrasekaran, M. Dafermos, X. Dong, T. Faulkner, Z. Fisher, G. Horowitz, T. Jacobson, N. Kamran, C. Keeler, N. Lashkari, A. Levine, A. Maloney, D. Marolf, H. Maxfield, R. Myers, F. Pretorius, I. Rodnianski, A. Speranza, M. van Raamsdonk, and H. Verlinde for helpful discussions.   NE is supported by the Princeton University Gravity Initiative and by NSF grant No. PHY-1620059. SF acknowledges the support of the Natural Sciences and Engineering Research Council of Canada (NSERC), funding reference number SAPIN/00032-2015, and of a grant from the Simons Foundation (385602, AM).

\appendix

\section{Variation Formulas}
\label{app:variations}

\subsection{Area}
\label{subsec:areavar}

Here we derive the expressions for first variation of the area of an arbitrary surface and the second variation of the area of an extremal surface.  A more general treatment, which applies to general geometric functionals, can be found in Appendix~C of~\cite{FisWis16} (which is a more formal version of the so-called calculus of moving surfaces~\cite{Gri13}).

Consider a one-parameter family of surfaces~$\Sigma(\lambda)$ in an ambient geometry with metric~$g_{ab}$ (which for now we take to be~$\lambda$-independent).  The area of these surfaces is given by the functional
\be
A(\lambda) \equiv A[\Sigma(\lambda)] = \int_{\Sigma(\lambda)} \bm{\eps},
\ee
where~$\bm{\eps}$ is the natural volume form on~$\Sigma(\lambda)$.  As in Section~\ref{sec:classical}, we may extend~$\bm{\eps}$ to a field on the surface swept out by the~$\Sigma(\lambda)$ as~$\lambda$ is varied; then converting to the passive picture, we may equivalently express~$A(\lambda)$ as
\be
A(\lambda) = \int_\Sigma \phi_{-\lambda}^* \bm{\eps},
\ee
where~$\phi_\lambda$ is a one-parameter group of diffeomorphisms that map~$\Sigma$ to~$\Sigma(\lambda)$ and $\phi^*_{-\lambda}$ is the pullback to~$\Sigma$.  The derivative of~$A(\lambda)$ can then be evaluated by a Lie derivative:
\be
\left. \frac{dA}{d\lambda} \right|_{\lambda = 0} = \int_\Sigma \pounds_\eta \bm{\eps} = \int_\Sigma \iota_\eta d \bm{\eps} + \int_{\partial \Sigma} \iota_\eta \bm{\eps},
\ee
where we used Cartan's formula~$\pounds_\eta \bm{f} = \iota_\eta d \bm{f} + d \iota_\eta \bm{f}$ for any form~$\bm{f}$ (with~$d$ the exterior derivative and~$\iota_\eta$ the interior derivative, i.e.~the contraction of~$\eta^a$ with the first index of~$f_{a_1 \cdots a_s}$) followed by an application of Stokes' theorem.  This expression can be simplified by noting that since any~$n$-form tangent to~$\Sigma$ must be proportional to~$\bm{\eps}$, we may write~$\iota_\eta d\bm{\eps} = \alpha \bm{\eps} + \cdots$ for some~$\alpha$, where the ellipsis denotes terms whose projection tangent to~$\Sigma$ vanishes.  Contracting this expression with~$\eps^{a_1 \cdots a_s}$, we obtain~$\alpha = \eta^a K_a$.  Likewise, if~$\partial \Sigma$ is nongenerate (which we shall assume), there is also a unique volume form~$^\partial \! \bm{\eps}$ on~$\partial \Sigma$.  $\iota_\eta \bm{\eps}$ must be proportional to this volume form, and indeed it is straightforward to show that~$\iota_\eta \bm{\eps} = \eta^a N_a \, ^\partial \! \bm{\eps}$, where~$N^a$ is the unit normal to~$\partial \Sigma$ in~$\Sigma$, taken to be outward- (inward-)pointing if~$N^a$ is spacelike (timelike).  Put together, these results yield the first area variation formula
\be
\label{eq:firstareavar}
\left. \frac{dA}{d\lambda} \right|_{\lambda = 0} = \int_\Sigma \eta^a K_a \, \bm{\eps} + \int_{\partial \Sigma} \eta^a N_a \, ^\partial \! \bm{\eps}.
\ee
Here we can immediately recover the fact that surfaces which are stationary points of the area functional have vanishing mean curvature:~$dA/d\lambda = 0$ for any perturbation~$\eta^a$ (obeying appropriate boundary conditions, if~$\Sigma$ has a boundary) if and only if~$K_a = 0$ everywhere.

To obtain the second area variation formula for extremal surfaces, consider an arbitrary two-parameter family of surfaces~$\Sigma(\lambda_1, \lambda_2)$ in a fixed background~$g_{ab}$ with the property that~$\Sigma \equiv \Sigma(0,0)$ is extremal with respect to~$g_{ab}$ (it is quite simple to generalize this result to the case where~$g_{ab}$ varies as well).  Thus first fixing~$\lambda_1$ and varying~$\lambda_2$, we have from~\eqref{eq:firstareavar} that
\be
\frac{\partial A(\lambda_1, 0)}{\partial \lambda_2} = \int_\Sigma \eta_2^a K_a \, \bm{\eps} + \int_{\partial \Sigma} \eta_2^a N_a \, ^\partial \! \epsb.
\ee
Now we take another derivative in~$\lambda_1$: using the fact that~$\Sigma$ is extremal, we have simply
\be
\frac{\partial^2 A(0,0)}{\partial \lambda_1 \partial \lambda_2} = \int_\Sigma \eta_2^a \left. \frac{\partial K_a}{\partial \lambda_1} \right|_{\lambda_1 = 0} \bm{\eps} + \mathrm{b.t.},
\ee
where we will discuss the boundary term b.t.~momentarily.  Thus using~\eqref{eq:Kdotgeneral} for the derivative~$\partial K_a/\partial \lambda_1$ (and again the fact that~$K_a = 0$), we obtain
\be
\frac{\partial^2 A(0,0)}{\partial \lambda_1 \partial \lambda_2} = \int_\Sigma \eta_{2,\perp}^a J (\eta_{1,\perp})_a \, \bm{\eps} + \mathrm{b.t.}
\ee
The boundary term is not needed anywhere in this paper, and since its computation is rather cumbersome we will not show it here.  For completeness, we simply state the result, which can be found in~\cite{BaoCao19}: taking both~$\eta_1^a$ and~$\eta_2^a$ to be normal to~$\partial \Sigma$ (though not necessarily to~$\Sigma$), the general second variation formula for extremal surfaces is
\begin{multline}
\label{eq:boundarysecondareavar}
\frac{\partial^2 A(0,0)}{\partial \lambda_1 \partial \lambda_2} = \int_\Sigma \eta_{2,\perp}^a J (\eta_{1,\perp})_a \, \bm{\eps} + \int_{\partial \Sigma} \left[\eta_{2,\perp}^a N^b D_b (\eta_{1,\perp})_a + N_a \eta_1^b \grad_b \eta_2^a \right. \\ \left. + \, ^\partial \! K_a N_b \left(\eta_1^a \, \eta_2^b + \eta_2^a \, \eta_1^b - g^{ab} p_{cd} \, \eta_1^c \, \eta_2^d\right)\right] \, ^\partial \! \bm{\eps},
\end{multline}
where~$^\partial \! K_a$ is the mean curvature of~$\partial \Sigma$ and~$p_{ab} = N^2 N_a N_b$ is the normal projector to~$\partial \Sigma$ in~$\Sigma$ (here~$N^2 = \pm 1$ is just a sign).

Ignoring the boundary term, it is natural to interpret the second variation formula~\eqref{eq:extremalsecondareavar} as an inner product: for any tensors~$u^{a_1 \cdots a_k}$,~$v^{a_1 \cdots a_k}$ in~$T_\Sigma M$, we define
\be
\label{eq:inprod}
\inprod{u}{v} \equiv \int_\Sigma u_{a_1 \cdots a_k} v^{a_1 \cdots a_k} \, \bm{\eps},
\ee
so we have that for perturbations of extremal surfaces,
\be
\label{eq:extremalsecondareavar}
\frac{\partial^2 A(0, 0)}{\partial \lambda_1 \partial \lambda_2} = \inprod{\eta_{2,\perp}}{J \eta_{1,\perp}}.
\ee
Note that the commutativity of partial derivatives~$\partial_{\lambda_1} \partial_{\lambda_2} A = \partial_{\lambda_2} \partial_{\lambda_1} A$ implies that~$J$ must be formally self-adjoint (that is, self-adjoint up to boundary terms) under this inner product.  This property can be seen directly from the definition~\eqref{subeq:Ldef}: the tensors~$S_{ab}$ and~$h^{cd} {P_a}^e {P_b}^f R_{cedf}$ are clearly symmetric, while for any~$u, v \in T^\perp_\Sigma M$ the Laplacian on the normal bundle obeys
\be
\inprod{u}{D^2 v} = \inprod{D^2 u}{v} + \int_\Sigma D_a(u_b D^a v^b - v_b D^a u^b) \, \bm{\eps} = \inprod{D^2 u}{v},
\ee
since~$u_b D^a v^b - v_b D^a u^b$ is tangent to~$\Sigma$ and thus the integrand is a divergence on~$\Sigma$.  Indeed, using this expression -- including the boundary terms -- and the fact that~$\eta_1^a$ and~$\eta_2^a$ commute since they are coordinate basis vectors, it is easy to see that the right-hand side of~\eqref{eq:boundarysecondareavar} is symmetric under the exchange~$\eta_1^a \leftrightarrow \eta_2^a$, as is required by the commutativity of the partial derivatives~$\partial/\partial \lambda_1$ and~$\partial/\partial \lambda_2$.

\subsection{Generalized Entropy}

The variation of the quantum analog of the second area variation formula~\eqref{eq:extremalsecondareavar} is very straightforward to obain: again, consider a two-parameter family of surfaces~$\Sigma(\lambda_1, \lambda_2)$ (in a fixed background geometry and state) such that~$\Sigma \equiv \Sigma(0,0)$ is quantum extremal.  Because~$S_\mathrm{gen}$ is only defined for Cauchy-splitting surfaces, which have no (finite) boundary, we will not bother keeping track of boundary terms.  Taking a first derivative of the generalized entropy of course yields
\be
\frac{\partial S_\mathrm{gen}(\lambda_1, 0)}{\partial \lambda_2} = \int_\Sigma \left(K_a + 4G_N \hbar \, \frac{\Dcal S_\mathrm{out}^{(\lambda_1)}}{\Dcal \Sigma^a} \right) \eta_2^a \, \bm{\eps};
\ee
another derivative (and using the fact that~$\Sigma$ is quantum extremal) gives
\be
\frac{\partial^2 S_\mathrm{gen}(0, 0)}{\partial \lambda_1 \partial \lambda_2} = \int_\Sigma \frac{\partial}{\partial \lambda_1} \left.\left(K_a + 4G_N \hbar \, \frac{\Dcal S_\mathrm{out}^{(\lambda_1)}}{\Dcal \Sigma^a} \right)\right|_{\lambda_1 = 0} \eta_2^a \, \bm{\eps}.
\ee
The derivative can be evaluated using~\eqref{eq:Kdotgeneral} and~\eqref{eq:DSderiv} (with~\eqref{eq:LieVdown}) with no sources turned on; again using the fact that~$\Sigma$ is quantum extremal, we obtain the formula for the second variation of the generalized entropy:
\be
\frac{\partial^2 S_\mathrm{gen}(0,0)}{\partial \lambda_1 \partial \lambda_2} = \int_\Sigma \eta_{2,\perp}^a \left[ J(\eta_{1,\perp})_a + 4 G_N \hbar \left(\int_\Sigma \frac{\Dcal^2 S_\mathrm{out}}{\Dcal \Sigma^b(p') \Dcal \Sigma^a} \, \eta_{1,\perp}^b(p') \bm{\eps}(p')\right) \right] \bm{\eps}.
\ee

\section{Functional Derivatives}
\label{app:functionalderiv}

\subsection{Functional Covariant Derivative}
\label{subapp:functionalcovariant}

Here we show that functional derivative operators obeying the conditions~\ref{cond:linearity}-\ref{cond:scalar} given in Section~\ref{subsec:functionalderiv} exist, and that the compatibility condition~\ref{cond:compatible} picks out a unique such derivative operator.  Just as one can show the existence of ordinary covariant derivative operators~$\grad_a$ by working with coordinate derivatives, here we show the existence of the covariant functional derivative by working with coordinate \textit{functional} derivatives.  To that end, consider an arbitrary coordinate system~$\{y^\alpha\}$,~$\alpha = 1, \ldots, n$ on~$\Sigma$ and an arbitrary coordinate system~$\{x^\mu\}$,~$\mu = 1, \ldots, d$ on (at least a portion of)~$M$; then the map~$\psi: \Sigma \to M$ which embeds~$\Sigma$ in~$M$ is described by the~$d$ embedding functions~$X^\mu(y)$\footnote{For simplicity here we assume that all of~$\Sigma$ can be covered with the single coordinate chart~$\{y^\alpha\}$; if this is not the case, the discussion generalizes straightforwardly by instead considering an atlas over~$\Sigma$.}.  Any tensor field~${V_{a_1 \cdots a_k}}^{b_1 \cdots b_l}$ which is a functional of~$\Sigma$ can therefore be expressed in this coordinate system as a functional of the~$X^\mu(y)$:
\be
{V_{a_1 \cdots a_k}}^{b_1 \cdots b_l}[\Sigma] \equiv {V_{a_1 \cdots a_k}}^{b_1 \cdots b_l}[X^\mu(y)].
\ee
We now define the covariant functional derivative of~${V_{a_1 \cdots a_k}}^{b_1 \cdots b_l}$ associated to this coordinate system, denoted by~$\Dcal {V_{a_1 \cdots a_k}}^{b_1 \cdots b_l}/\Dcal X^b$, as the tensor on~$\Sigma$ whose components in this coordinate system are given by
\be
\label{eq:coordfuncderiv}
\frac{\Dcal {V_{\mu_1 \cdots \mu_k}}^{\nu_1 \cdots \nu_l}}{\Dcal X^\sigma(y)} = \frac{1}{\sqrt{h(y)}} \sum_{\rho = 1}^d {P_\sigma}^\rho \, \frac{\delta {V_{\mu_1 \cdots \mu_k}}^{\nu_1 \cdots \nu_l}[X^\lambda(y)]}{\delta X^\rho(y)},
\ee
where~$h$ is the determinant of the components~$h_{\alpha\beta}$ of the induced metric in the coordinate system~$\{y^\alpha\}$ and~${P_\sigma}^\rho$ are the components of the normal projector~${P_a}^b$.  Note that the object on the right-hand side is now a sum of ordinary functional derivatives of the (scalar) components of~${V_{a_1 \cdots a_k}}^{b_1 \cdots b_l}$.  (For an ordinary functional~$F[\Sigma]$, this definition of~$\Dcal F/\Dcal X^a$ is essentially a vector version of the vector \textit{density}~$\delta F/\delta X^a$ defined in~\cite{EngWal14}.)

The object~\eqref{eq:coordfuncderiv} is manifestly normal to~$\Sigma$ in the index~$\sigma$, and therefore property~\ref{cond:normal} is immediately satisfied.  Since ordinary functional derivatives are linear and obey the Leibnitz rule, properties~\ref{cond:linearity} and~\ref{cond:leibnitz} are also satisfied by this functional derivative.  It is also easy to see that this definition satisfies property~\ref{cond:contraction}, commutativity with contraction.  To check property~\ref{cond:scalar}, i.e.~the functional variation of scalars and ordinary functionals, note that a one-parameter family of surfaces~$\Sigma(\lambda)$ is encoded in this coordinate system as a one-parameter family of embedding functions~$X^\mu(\lambda; y)$; the components of the deviation vector along this family are given by~$\eta^\mu = dX^\mu/d\lambda|_{\lambda = 0}$.  Moreover, the points~$p \in \Sigma$ and~$\phi_\lambda(p) \in \Sigma(\lambda)$ have the same~$y$ coordinate values; thus by the chain rule for ordinary functional derivatives, we have that for any scalar~$V[\Sigma](p_i) = V[X^\mu(y)](y_i)$,
\bea
\left. \frac{dV[X^\mu(\lambda;y)](y_i)}{d\lambda} \right|_{\lambda = 0} &= \int \sum_{\nu = 1}^d \frac{\delta V[X^\mu(y)](y_i)}{\delta X^\nu(y')}  \left. \frac{dX^\nu(\lambda; y')}{d\lambda}\right|_{\lambda = 0} \, d^n y', \\
		&= \int \sum_{\nu = 1}^d \frac{\Dcal V(y_i)}{\Dcal X^\nu(y')} \, \eta^\nu(y') \sqrt{h(y')} \, d^n y',
\eea
with the last equality holding because~$\eta^a$ is taken normal to~$\Sigma$.  Hence using the fact that the natural volume element on~$\Sigma$ can be written as~$\epsb = \sqrt{h} \, dy^1 \wedge \cdots \wedge dy^n$, we verify that property~\ref{cond:scalar} holds; thus we have shown the existence of covariant functional derivatives satisfying the four properties~\ref{cond:linearity}-\ref{cond:scalar}.

To show that the compatibility condition~\eqref{eq:compatible} is sufficient to uniquely fix a preferred functional covariant derivative, let us first note that for any (local) scalar~$F(p)$ and any dual vector~$V_a(p)$, and for any two covariant functional derivatives~$\Dcal/\Dcal \Sigma^a$ and~$\widetilde{\Dcal}/\widetilde{\Dcal} \Sigma^a$, we have
\be
\label{eq:DDtilde}
\left(\frac{\Dcal}{\Dcal \Sigma^a(p')} - \frac{\widetilde{\Dcal}}{\widetilde{\Dcal} \Sigma^a(p')}\right)(F(p) V_b(p)) = F(p) \left(\frac{\Dcal}{\Dcal \Sigma^a(p')} - \frac{\widetilde{\Dcal}}{\widetilde{\Dcal} \Sigma^a(p')}\right)V_b(p),
\ee
which follows from the Leibitz rule and the fact that by~\eqref{eq:scalarvar},~$\Dcal F/\Dcal \Sigma^a = \widetilde{\Dcal} F/\widetilde{\Dcal} \Sigma^a$.  This property shows that the difference between two covariant functional derivatives acting on a dual vector depends only locally on the value of that dual vector on~$\Sigma$.  There must therefore exist a tensor~${\Ccal^a}_{bc}(p,p',p'')$ such that
\be
\label{eq:Cdef}
\frac{\Dcal V_b(p)}{\Dcal \Sigma^a(p')} = \frac{\widetilde{\Dcal} V_b(p)}{\widetilde{\Dcal} \Sigma^a(p')} - \int_\Sigma {\Ccal^c}_{ba}(p'',p,p') V_c(p'') \, \epsb(p''),
\ee
where~\eqref{eq:DDtilde} requires that~${\Ccal^a}_{bc}(p,p',p'') = 0$ when~$p \neq p'$.  Note that the relationship between~$\Dcal/\Dcal \Sigma^a$ and~$\widetilde{\Dcal}/\widetilde{\Dcal} \Sigma^a$ when acting on higher-rank multilocal tensors can be inferred from~\eqref{eq:Cdef} by using the Leibnitz rule and the fact that~$\Dcal/\Dcal \Sigma^a$ and~$\widetilde{\Dcal}/\widetilde{\Dcal} \Sigma^a$ act the same on scalars; the result is analogous to the relationship between two different ordinary covariant derivatives~$\grad_a$ and~$\widetilde{\grad}_a$.  Explicitly, for a given multlilocal tensor~${V_{a_1 \cdots a_k}}^{b_1 \cdots b_l}$ we have
\begin{multline}
\frac{\Dcal {V_{a_1 \cdots a_k}}^{b_1 \cdots b_l}(p_i,q_i)}{\Dcal \Sigma^c(p')} = \frac{\Dcal {V_{a_1 \cdots a_k}}^{b_1 \cdots b_l}(p_i,q_i)}{\Dcal X^c(p')} \\ - \sum_{i = 1}^k \int_\Sigma {V_{a_1 \cdots d \cdots a_k}}^{b_1 \cdots b_l}(p_1, \ldots, p'', \ldots, p_k, q_i) \, {\Ccal^d}_{a_i c}(p'', p_i, p') \, \epsb(p'') \\ + \sum_{i = 1}^l \int_\Sigma {V_{a_1 \cdots a_k}}^{b_1 \cdots d \cdots b_l}(p_i, q_1, \ldots, q'', \ldots, q_l) \, {\Ccal^{b_i}}_{d c}(q_i, q'', p') \, \epsb(q''),
\end{multline}
where the~$p_i$ and~$q_i$ schematically label the points on whose tangent spaces the lower and upper indices act, respectively.

Now take~$\widetilde{\Dcal}/\widetilde{\Dcal} \Sigma^a$ to be the coordinate functional derivative~$\Dcal/\Dcal X^a$ associated to some coordinate system and consider a dual vector field~$v_a$ on~$M$; its restriction to any surface~$\Sigma$ is obtained by just evaluating~$v_a$ on~$\Sigma$.  It then follows that
\be
\frac{\Dcal v_\mu(y)}{\Dcal X^\nu(y')} = \frac{1}{\sqrt{h'}} \sum_{\sigma = 1}^d {P_\nu}^\sigma \frac{\delta v_\mu(X^\lambda(y))}{\delta X^\sigma(y')} = \delta(y,y') \sum_{\sigma = 1}^d {P_\nu}^\sigma \partial_\sigma v_\mu.
\ee
But~\eqref{eq:compatible} and~\eqref{eq:Cdef} imply
\begin{subequations}
\begin{multline}
\delta(y,y') \sum_{\sigma = 1}^d {P_\nu}^\sigma \grad_\sigma v_\mu = \frac{\Dcal v_\mu(y)}{\Dcal \Sigma^\nu(y')} \\ = \frac{\Dcal v_\mu(y)}{\Dcal X^\nu(y')} - \int_\Sigma {\Ccal^\sigma}_{\mu\nu}(y'', y,y') v_\sigma(y'') \sqrt{h(y'')} \, d^n y'',
\end{multline}
and thus
\begin{multline}
\delta(y,y') \sum_{\sigma = 1}^d {P_\nu}^\sigma \left( \partial_\sigma v_\mu - \sum_{\lambda = 1}^d {\Gamma^\lambda}_{\mu\sigma} v_\lambda\right) \\= \delta(y,y') \sum_{\sigma = 1}^d {P_\nu}^\sigma \partial_\sigma v_\mu - \int_\Sigma {\Ccal^\sigma}_{\mu\nu}(y'', y,y')  v_\sigma(y'') \sqrt{h(y'')} \, d^n y'',
\end{multline}
\end{subequations}
where we expressed the covariant derivative~$\grad_a$ in terms of the ordinary coordinate derivative and the Christoffel symbols of this coordinate system.  Requiring this expression to hold for all~$v_a$, we conclude that the connection~${\Ccal^a}_{bc}$ associated to some particular coordinate system is given by
\be
\label{eq:Cfunc}
{\Ccal^a}_{bc}(p,p',p'') = \delta(p,p') \delta(p',p'') {P_c}^d {\Gamma^a}_{bd}.
\ee
As promised, this fixes~${\Ccal^a}_{bc}$ -- and therefore~$\Dcal/\Dcal \Sigma^a$ -- uniquely.

\subsection{Functional Lie Derivative}
\label{subapp:functionalLie}

Now let us obtain a covariant expression for functional Lie derivatives in terms of the functional covariant derivative.  Such an expression was given in~\eqref{eq:Liescalar} for the Lie derivative of scalars; to determine the action of~$\pounds_\eta$ on general rank tensors, it is convenient to introduce a coordinate system adapted to the group of diffeomorphisms~$\phi_\lambda$ that defines~$\pounds_\eta$.  Thus assume temporarily that the family of surfaces~$\Sigma(\lambda)$ generated by~$\phi_\lambda$ do not intersect (for small~$\lambda$).  Then~$\eta^a$ must be nowhere-vanishing on~$\Sigma$ and we can introduce a coordinate system~$\{x^\mu\}$ on~$M$ in which~$\eta^a = (\partial_1)^a$.  The action of~$\phi_{-\lambda}$ thus corresponds to the coordinate transformation which sends~$x^1 \to x^1 + \lambda$ and leaves the other coordinates fixed, and thus the matrix of components~${(\phi^*_{-\lambda})^\mu}_\nu$ is just the identity.  If we introduce a coordinate system~$\{y^\alpha\}$ on~$\Sigma$, then the points~$p \in \Sigma$ and~$\phi_\lambda(p) \in \Sigma(\lambda)$ have the same~$y$ coordinate values, while the embedding functions~$X^\mu(\lambda; y)$ that define the family~$\Sigma(\lambda)$ are given by~$X^1(\lambda; y) = X^1(y) + \lambda$,~$X^{\mu \neq 1}(\lambda; y) = X^{\mu \neq 1}(y)$.  Thus the components of the pullback of any multilocal functional~${V_{a_1 \cdots a_k}}^{b_1 \cdots b_l}[\Sigma](p) \equiv {V_{a_1 \cdots a_k}}^{b_1 \cdots b_l}[X^\mu(y')](y)$ to~$\Sigma$ are
\be
\phi^*_{-\lambda}{V_{\mu_1 \cdots \mu_k}}^{\nu_1 \cdots \nu_l}[X^\sigma(y')](y_i) = {V_{\mu_1 \cdots \mu_k}}^{\nu_1 \cdots \nu_l}[X^1(y') + \lambda, X^{\sigma \neq 1}(y')](y_i).
\ee
Using the definition~\eqref{eq:functionalLie}, we thus find that in this coordinate system, the components of the Lie derivative of~${V_{a_1 \cdots a_k}}^{b_1 \cdots b_l}$ are
\be
\label{eq:Lietensorcoord}
\pounds_\eta {V_{\mu_1 \cdots \mu_k}}^{\nu_1 \cdots \nu_l} = \int \frac{\delta {V_{\mu_1 \cdots \mu_k}}^{\nu_1 \cdots \nu_l}}{\delta X^1(y')} \, d^n y'.
\ee
Since ordinary functional derivatives obey the Leibnitz rule, this guarantees that the functional Lie derivative does as well.

In the special case of a vector functional~$V^a[\Sigma](p) \equiv V^a[X^\mu(y)](y')$, we find
\be
\label{eq:LieVmu}
\pounds_\eta V^\mu = \int \frac{\delta V^\mu}{\delta X^1(y')} \, d^n y'.
\ee
On the other hand, consider the object
\be
\int_\Sigma \frac{\Dcal V^a}{\Dcal \Sigma^b(p')} \eta^b(p') \, \epsb(p') - V^b \grad_b \eta^a;
\ee
it is easy to see (using the connection~\eqref{eq:Cfunc} between the covariant functional derivative~$\Dcal/\Dcal \Sigma^a$ and the coordinate functional derivative~$\Dcal/\Dcal X^a$) that when~$\eta^a$ is normal to~$\Sigma$, the components of this object in this coordinate system are just equal to~\eqref{eq:LieVmu}.  Since both expressions are obtained from covariant definitions, we conclude that whenever~$\eta^a$ is normal to~$\Sigma$,
\be
\label{eq:LieVuppernormal}
\pounds_\eta V^a = \int_\Sigma \frac{\Dcal V^a}{\Dcal \Sigma^b(p')} \eta^b(p') \, \epsb(p') - V^b \grad_b \eta^a.
\ee
Let us now note that although this expression was obtained under the assumption that~$\eta^a$ be nonvanishing on~$\Sigma$, the case where~$\eta^a$ vanishes somewhere can be treated by introducing an appropriate atlas of coordinate systems corresponding to different regions of nonvanishing~$\eta^a$; then the coordinate expressions~\eqref{eq:Lietensorcoord} and~\eqref{eq:LieVmu} will consist of a sum of integrals over each chart, but the final covariant expression~\eqref{eq:LieVuppernormal} will remain unchanged.

If~$\eta^a$ is not normal to~$\Sigma$, we may decompose it into its normal and tangent pieces~$\eta_\perp^a$ and~$\eta_\parallel^a$ which we may interpret as generators of two different diffeomorphisms; then~$\pounds_\eta V^a = \pounds_{\eta_\perp} V^a + \pounds_{\eta_\parallel} V^a$, with~$\pounds_{\eta_\perp} V^a$ given by~\eqref{eq:LieVuppernormal}.  Since the diffeomorphism generated by~$\eta_\parallel^a$ doesn't change the image of~$\Sigma$ in~$M$, it must just act as a normal diffeomorphism of the vector field~$V^a$ on~$\Sigma$; we would therefore conclude that
\be
\label{eq:LieVupper}
\pounds_\eta V^a = \int_\Sigma \frac{\Dcal V^a}{\Dcal \Sigma^b(p')} \eta^b(p') \, \epsb(p') + \eta^b_\parallel \grad_b V^a - V^b \grad_b \eta^a.
\ee
Lie derivatives of higher-rank multilocal tensors are then fixed by~\eqref{eq:LieVupper},~\eqref{eq:Liescalar}, and the Leibnitz rule.  For instance, we must have for any~$V^a$ and~$U_a$ (and~$\eta^a$ not necessarily normal to~$\Sigma$)
\begin{subequations}
\be
\pounds_\eta(V^a U_a) = \int_\Sigma \frac{\Dcal (V_a U^a)}{\Dcal \Sigma^b(p')} \eta^b(p') \, \epsb(p') + \eta^b_\parallel \grad_b (V^a U_a),
\ee
\begin{multline}
U_a \pounds_\eta V^a + V^a \pounds_\eta U_a = \int_\Sigma \left(U_a \frac{\Dcal V^a}{\Dcal \Sigma^b(p')}+ V^a \frac{\Dcal U_a}{\Dcal \Sigma^b(p')} \right) \eta^b(p') \, \epsb(p') \\ + V^a \eta^b_\parallel \grad_b U_a + U_a \eta^b_\parallel \grad_b V^a,
\end{multline}
\end{subequations}
and thus using~\eqref{eq:LieVupper} we find
\be
\pounds_\eta U_a = \int_\Sigma \frac{\Dcal U^a}{\Dcal \Sigma^b(p')} \eta^b(p') \, \epsb(p') + \eta^b_\parallel \grad_b U_a + U_b \grad_a \eta^b.
\ee
The generalization to a multilocal tensor functional is straightforward: slightly schematically, if~${V_{a_1 \cdots a_k}}^{b_1 \cdots b_l}(p_1, \ldots, p_r)$ depends on~$r$ points on~$\Sigma$,
\begin{multline}
\pounds_\eta {V_{a_1 \cdots a_k}}^{b_1 \cdots b_l} = \int_\Sigma \frac{\Dcal {V_{a_1 \cdots a_k}}^{b_1 \cdots b_l}}{\Dcal \Sigma^c(p')} \eta^c(p') \, \epsb(p') + \sum_{i = 1}^r \eta^c_\parallel \grad_c^{(p_i)} {V_{a_1 \cdots a_k}}^{b_1 \cdots b_l} \\ + \sum_{i = 1}^k {V_{a_1 \cdots c \cdots a_k}}^{b_1 \cdots b_l} \grad_{a_i} \eta^c - \sum_{i = 1}^l {V_{a_1 \cdots a_k}}^{b_1 \cdots c \cdots b_l} \grad_c \eta^{b_i},
\end{multline}
where~$\eta^b_\parallel \grad_b^{(p_i)}$ denotes taking the directional derivative along~$\eta_\parallel^a$ of~${V_{a_1 \cdots a_k}}^{b_1 \cdots b_l}(p_1, \ldots, p_r)$ at the point~$p_i$ (so~$\eta^b_\parallel \grad_b^{(p_i)}$ ignores any indices that do not act on the tangent space~$T_{p_i} M$) and each of the objects~$\grad_a \eta^b$ is evaluated at the point corresponding to the index of~${V_{a_1 \cdots a_k}}^{b_1 \cdots c \cdots b_l}$ into which it is to be contracted.

\section{Green's Functions}
\label{app:Greens}

To compute the Green's function~$G(p,p')$ defined by~\eqref{eq:Greens} explicitly, first note that~$\Hcal$ is the hyperbolic ball, which is a maximally symmetric space; thus~$G(p,p')$ can depend only on the geodesic distance between its arguments.  In the coordinates of~\eqref{eq:hyperbolicball}, we note that the geodesic distance between the point~$\chi = 0$ and any other point~$(\chi, \Omega^i)$ is just~$l \chi$.  Thus~$G(0,\{\chi, \Omega^i\}) = G(\chi)$, which for~$\chi \neq 0$ solves
\be
D^2 G(\chi) - \frac{d-2}{l^2} G(\chi) = \frac{1}{l^2 \sinh^{d-3} \chi} \partial_\chi \left(\sinh^{d-3} \chi \, \partial_\chi G(\chi)\right) - \frac{d-2}{l^2} G(\chi) = 0.
\ee
The (unique) solution to this equation which vanishes as~$\chi \to \infty$ and is normalized to obey~\eqref{eq:Greens} is
\begin{multline}
G(\chi) = \frac{\Gamma(\frac{d-2}{2})}{2l^{d-4}(d-4)\pi^{(d-2)/2}} \cosh\chi \left[\frac{2\sqrt{\pi} \, \Gamma(\frac{6-d}{2}) \sec(\frac{d\pi}{2})}{\Gamma(\frac{3-d}{2})} \right. \\ \left. + \, _2 F_1\left(\frac{3-d}{2}, \frac{4-d}{2}, \frac{6-d}{2}; \tanh^2\chi\right)\tanh^{4-d} \chi \right].
\end{multline}
More explicitly, the Green's functions for the first few dimensions are
\bea
d = 3:& \quad G(\chi) = \frac{l}{2} \, e^{-|\chi|}, \\
d = 4:& \quad G(\chi) = -\frac{1}{2\pi} \left[1 + \cosh\chi \, \ln \left(\frac{\sinh\chi}{1+\cosh\chi} \right) \right], \\
d = 5:& \quad G(\chi) = \frac{1}{4\pi l} \, e^{-2\chi} \csch\chi, \\
d = 6:& \quad G(\chi) = \frac{1}{4\pi^2 l^2} \left[2 + \coth^2\chi + 3\cosh\chi \, \ln \left(\frac{\sinh\chi}{1+\cosh\chi} \right) \right], \\
d = 7:& \quad G(\chi) = \frac{1}{8\pi^2 l^3} \, e^{-3\chi} (3+\coth\chi)\csch^2\chi.
\eea

\bibliographystyle{jhep}
\bibliography{all}

\providecommand{\href}[2]{#2}\begingroup\raggedright\begin{thebibliography}{10}

\bibitem{Mal97}
J.~Maldacena, {\it The large {$N$} limit of superconformal field theories and
  supergravity},  {\em Adv. Theor. Math. Phys.} {\bf 2} (1998) 231,
  [\href{http://arxiv.org/abs/hep-th/9711200}{{\tt hep-th/9711200}}].

\bibitem{Wit98a}
E.~Witten, {\it {A}nti-de~{S}itter space and holography},  {\em Adv. Theor.
  Math. Phys.} {\bf 2} (1998) 253,
  [\href{http://arxiv.org/abs/hep-th/9802150}{{\tt hep-th/9802150}}].

\bibitem{GubKle98}
S.~S. Gubser, I.~R. Klebanov, and A.~M. Polyakov, {\it Gauge theory correlators
  from noncritical string theory},  {\em Phys. Lett.} {\bf B428} (1998) 105,
  [\href{http://arxiv.org/abs/hep-th/9802109}{{\tt hep-th/9802109}}].

\bibitem{DonHar16}
X.~Dong, D.~Harlow, and A.~C. Wall, {\it {Bulk Reconstruction in the
  Entanglement Wedge in AdS/CFT}},  \href{http://arxiv.org/abs/1601.05416}{{\tt
  arXiv:1601.05416}}.

\bibitem{FauLew17}
T.~Faulkner and A.~Lewkowycz, {\it {Bulk locality from modular flow}},
  \href{http://arxiv.org/abs/1704.05464}{{\tt arXiv:1704.05464}}.

\bibitem{RyuTak06}
S.~Ryu and T.~Takayanagi, {\it {Holographic derivation of entanglement entropy
  from AdS/CFT}},  {\em Phys.Rev.Lett.} {\bf 96} (2006) 181602,
  [\href{http://arxiv.org/abs/hep-th/0603001}{{\tt hep-th/0603001}}].

\bibitem{HubRan07}
V.~E. Hubeny, M.~Rangamani, and T.~Takayanagi, {\it {A Covariant holographic
  entanglement entropy proposal}},  {\em JHEP} {\bf 0707} (2007) 062,
  [\href{http://arxiv.org/abs/0705.0016}{{\tt arXiv:0705.0016}}].

\bibitem{FauLew13}
T.~Faulkner, A.~Lewkowycz, and J.~Maldacena, {\it {Quantum corrections to
  holographic entanglement entropy}},  {\em JHEP} {\bf 1311} (2013) 074,
  [\href{http://arxiv.org/abs/1307.2892}{{\tt arXiv:1307.2892}}].

\bibitem{EngWal14}
N.~Engelhardt and A.~C. Wall, {\it {Quantum Extremal Surfaces: Holographic
  Entanglement Entropy beyond the Classical Regime}},  {\em JHEP} {\bf 01}
  (2015) 073, [\href{http://arxiv.org/abs/1408.3203}{{\tt arXiv:1408.3203}}].

\bibitem{Bek72}
J.~D. Bekenstein, {\it Black holes and the second law},  {\em Nuovo Cim. Lett.}
  {\bf 4} (1972) 737--740.

\bibitem{ColMin}
T.~H. Colding and W.~P. Minicozzi~II, {\em A Course in Minimal Surfaces}.
\newblock American Mathematical Society, Providence, Rhode Island, 2011.

\bibitem{LarFro93}
A.~L. Larsen and V.~P. Frolov, {\it {Propagation of perturbations along
  strings}},  {\em Nucl. Phys.} {\bf B414} (1994) 129--146,
  [\href{http://arxiv.org/abs/hep-th/9303001}{{\tt hep-th/9303001}}].

\bibitem{Guv93}
J.~Guven, {\it Perturbations of a topological defect as a theory of coupled
  scalar fields in curved space interacting with an external vector potential},
   {\em Phys. Rev. D} {\bf 48} (Dec, 1993) 5562--5569.

\bibitem{VisPar96}
K.~S. Viswanathan and R.~Parthasarathy, {\it {String theory in curved
  space-time}},  {\em Phys. Rev.} {\bf D55} (1997) 3800--3810,
  [\href{http://arxiv.org/abs/hep-th/9605007}{{\tt hep-th/9605007}}].

\bibitem{Car92}
B.~Carter, {\it Basic brane theory},  {\em Classical and Quantum Gravity} {\bf
  9} (dec, 1992) S19--S33.

\bibitem{Car92b}
B.~Carter, {\it Outer curvature and conformal geometry of an imbedding},  {\em
  Journal of Geometry and Physics} {\bf 8} (1992), no.~1 53 -- 88.

\bibitem{Car93}
B.~Carter, {\it Perturbation dynamics for membranes and strings governed by the
  dirac-goto-nambu action in curved space},  {\em Phys. Rev. D} {\bf 48} (Nov,
  1993) 4835--4838.

\bibitem{BatCar95}
R.~A. Battye and B.~Carter, {\it {Gravitational perturbations of relativistic
  membranes and strings}},  {\em Phys. Lett.} {\bf B357} (1995) 29--35,
  [\href{http://arxiv.org/abs/hep-ph/9508300}{{\tt hep-ph/9508300}}].

\bibitem{BatCar00}
R.~A. Battye and B.~Carter, {\it {Second order Lagrangian and symplectic
  current for gravitationally perturbed Dirac-Goto-Nambu strings and branes}},
  {\em Class. Quant. Grav.} {\bf 17} (2000) 3325--3334,
  [\href{http://arxiv.org/abs/hep-th/9811075}{{\tt hep-th/9811075}}].

\bibitem{Mos17}
B.~Mosk, {\it {Metric Perturbations of Extremal Surfaces}},  {\em Class. Quant.
  Grav.} {\bf 35} (2018), no.~4 045013,
  [\href{http://arxiv.org/abs/1710.01316}{{\tt arXiv:1710.01316}}].

\bibitem{GhoMis17}
A.~Ghosh and R.~Mishra, {\it {Inhomogeneous Jacobi equation for minimal
  surfaces and perturbative change in holographic entanglement entropy}},  {\em
  Phys. Rev.} {\bf D97} (2018), no.~8 086012,
  [\href{http://arxiv.org/abs/1710.02088}{{\tt arXiv:1710.02088}}].

\bibitem{LewPar18}
A.~Lewkowycz and O.~Parrikar, {\it {The holographic shape of entanglement and
  Einstein's equations}},  {\em JHEP} {\bf 05} (2018) 147,
  [\href{http://arxiv.org/abs/1802.10103}{{\tt arXiv:1802.10103}}].

\bibitem{AkeKoe16}
C.~Akers, J.~Koeller, S.~Leichenauer, and A.~Levine, {\it {Geometric
  Constraints from Subregion Duality Beyond the Classical Regime}},
  \href{http://arxiv.org/abs/1610.08968}{{\tt arXiv:1610.08968}}.

\bibitem{EngFis16}
N.~Engelhardt and S.~Fischetti, {\it {The Gravity Dual of Boundary Causality}},
   {\em Class. Quant. Grav.} {\bf 33} (2016), no.~17 175004,
  [\href{http://arxiv.org/abs/1604.03944}{{\tt arXiv:1604.03944}}].

\bibitem{AfkHar17}
N.~Afkhami-Jeddi, T.~Hartman, S.~Kundu, and A.~Tajdini, {\it {Shockwaves from
  the Operator Product Expansion}},
  \href{http://arxiv.org/abs/1709.03597}{{\tt arXiv:1709.03597}}.

\bibitem{EngWal13}
N.~Engelhardt and A.~C. Wall, {\it {Extremal Surface Barriers}},  {\em JHEP}
  {\bf 1403} (2014) 068, [\href{http://arxiv.org/abs/1312.3699}{{\tt
  arXiv:1312.3699}}].

\bibitem{Wal12}
A.~C. Wall, {\it {Maximin Surfaces, and the Strong Subadditivity of the
  Covariant Holographic Entanglement Entropy}},  {\em Class.Quant.Grav.} {\bf
  31} (2014), no.~22 225007, [\href{http://arxiv.org/abs/1211.3494}{{\tt
  arXiv:1211.3494}}].

\bibitem{LasRab14}
N.~Lashkari, C.~Rabideau, P.~Sabella-Garnier, and M.~Van~Raamsdonk, {\it
  {Inviolable energy conditions from entanglement inequalities}},  {\em JHEP}
  {\bf 06} (2015) 067, [\href{http://arxiv.org/abs/1412.3514}{{\tt
  arXiv:1412.3514}}].

\bibitem{LasLin16}
N.~Lashkari, J.~Lin, H.~Ooguri, B.~Stoica, and M.~Van~Raamsdonk, {\it
  {Gravitational positive energy theorems from information inequalities}},
  {\em PTEP} {\bf 2016} (2016), no.~12 12C109,
  [\href{http://arxiv.org/abs/1605.01075}{{\tt arXiv:1605.01075}}].

\bibitem{NeuSar18}
D.~Neuenfeld, K.~Saraswat, and M.~Van~Raamsdonk, {\it {Positive gravitational
  subsystem energies from CFT cone relative entropies}},  {\em JHEP} {\bf 06}
  (2018) 050, [\href{http://arxiv.org/abs/1802.01585}{{\tt arXiv:1802.01585}}].

\bibitem{BouFis15}
R.~Bousso, Z.~Fisher, S.~Leichenauer, and A.~C. Wall, {\it {Quantum focusing
  conjecture}},  {\em Phys. Rev.} {\bf D93} (2016), no.~6 064044,
  [\href{http://arxiv.org/abs/1506.02669}{{\tt arXiv:1506.02669}}].

\bibitem{BouFis15b}
R.~Bousso, Z.~Fisher, J.~Koeller, S.~Leichenauer, and A.~C. Wall, {\it {Proof
  of the Quantum Null Energy Condition}},  {\em Phys. Rev.} {\bf D93} (2016),
  no.~2 024017, [\href{http://arxiv.org/abs/1509.02542}{{\tt
  arXiv:1509.02542}}].

\bibitem{EngWal17b}
N.~Engelhardt and A.~C. Wall, {\it {Decoding the Apparent Horizon: A
  Coarse-Grained Holographic Entropy}},
  \href{http://arxiv.org/abs/1706.02038}{{\tt arXiv:1706.02038}}.

\bibitem{BaoCao19}
N.~Bao, C.~Cao, S.~Fischetti, and C.~Keeler, {\it {Towards Bulk Metric
  Reconstruction from Extremal Area Variations}},
  \href{http://arxiv.org/abs/1904.04834}{{\tt arXiv:1904.04834}}.

\bibitem{KoeLei15}
J.~Koeller and S.~Leichenauer, {\it {Holographic Proof of the Quantum Null
  Energy Condition}},  {\em Phys. Rev.} {\bf D94} (2016), no.~2 024026,
  [\href{http://arxiv.org/abs/1512.06109}{{\tt arXiv:1512.06109}}].

\bibitem{KoeLei17}
J.~Koeller, S.~Leichenauer, A.~Levine, and A.~Shahbazi-Moghaddam, {\it {Local
  Modular Hamiltonians from the Quantum Null Energy Condition}},  {\em Phys.
  Rev.} {\bf D97} (2018), no.~6 065011,
  [\href{http://arxiv.org/abs/1702.00412}{{\tt arXiv:1702.00412}}].

\bibitem{AkeCha17}
C.~Akers, V.~Chandrasekaran, S.~Leichenauer, A.~Levine, and
  A.~Shahbazi~Moghaddam, {\it {The Quantum Null Energy Condition, Entanglement
  Wedge Nesting, and Quantum Focusing}},
  \href{http://arxiv.org/abs/1706.04183}{{\tt arXiv:1706.04183}}.

\bibitem{BalFau17}
S.~Balakrishnan, T.~Faulkner, Z.~U. Khandker, and H.~Wang, {\it {A General
  Proof of the Quantum Null Energy Condition}},
  \href{http://arxiv.org/abs/1706.09432}{{\tt arXiv:1706.09432}}.

\bibitem{Bou18}
R.~Bousso, {\it {Black hole entropy and the Bekenstein bound}},
\newblock 2018.
\newblock \href{http://arxiv.org/abs/1810.01880}{{\tt arXiv:1810.01880}}.

\bibitem{Mez14}
{Mezei, M{\'a}rk}, {\it {Entanglement entropy across a deformed sphere}},  {\em
  Phys. Rev.} {\bf D91} (2015), no.~4 045038,
  [\href{http://arxiv.org/abs/1411.7011}{{\tt arXiv:1411.7011}}].

\bibitem{AllMez14}
A.~Allais and M.~Mezei, {\it {Some results on the shape dependence of
  entanglement and Rényi entropies}},  {\em Phys. Rev.} {\bf D91} (2015), no.~4
  046002, [\href{http://arxiv.org/abs/1407.7249}{{\tt arXiv:1407.7249}}].

\bibitem{NozNum13}
M.~Nozaki, T.~Numasawa, A.~Prudenziati, and T.~Takayanagi, {\it {Dynamics of
  Entanglement Entropy from Einstein Equation}},  {\em Phys. Rev.} {\bf D88}
  (2013), no.~2 026012, [\href{http://arxiv.org/abs/1304.7100}{{\tt
  arXiv:1304.7100}}].

\bibitem{Hub12}
V.~E. Hubeny, {\it {Extremal surfaces as bulk probes in AdS/CFT}},  {\em JHEP}
  {\bf 1207} (2012) 093, [\href{http://arxiv.org/abs/1203.1044}{{\tt
  arXiv:1203.1044}}].

\bibitem{Wald}
R.~M. Wald, {\em General Relativity}.
\newblock The University of Chicago Press, Chicago, 1984.

\bibitem{PynBir93}
T.~Pyne and M.~Birkinshaw, {\it {Null geodesics in perturbed space-times}},
  {\em Astrophys. J.} {\bf 415} (1993) 459,
  [\href{http://arxiv.org/abs/astro-ph/9303020}{{\tt astro-ph/9303020}}].

\bibitem{AndMar05}
L.~Andersson, M.~Mars, and W.~Simon, {\it {Local existence of dynamical and
  trapping horizons}},  {\em Phys. Rev. Lett.} {\bf 95} (2005) 111102,
  [\href{http://arxiv.org/abs/gr-qc/0506013}{{\tt gr-qc/0506013}}].

\bibitem{Wal18}
A.~C. Wall, {\it {A Survey of Black Hole Thermodynamics}},
  \href{http://arxiv.org/abs/1804.10610}{{\tt arXiv:1804.10610}}.

\bibitem{Don13}
X.~Dong, {\it {Holographic Entanglement Entropy for General Higher Derivative
  Gravity}},  {\em JHEP} {\bf 01} (2014) 044,
  [\href{http://arxiv.org/abs/1310.5713}{{\tt arXiv:1310.5713}}].

\bibitem{Cam13}
J.~Camps, {\it {Generalized entropy and higher derivative Gravity}},  {\em
  JHEP} {\bf 03} (2014) 070, [\href{http://arxiv.org/abs/1310.6659}{{\tt
  arXiv:1310.6659}}].

\bibitem{JacMye93}
T.~Jacobson and R.~C. Myers, {\it Black hole entropy and higher curvature
  interactions},  {\em Phys. Rev. Lett.} {\bf 70} (1993) 3684--3687,
  [\href{http://arxiv.org/abs/http://arXiv.org/abs/hep-th/9305016}{{\tt
  http://arXiv.org/abs/hep-th/9305016}}].

\bibitem{Wal93}
R.~M. Wald, {\it Black hole entropy is the {N}oether charge},  {\em Phys. Rev.
  D} {\bf 48} (1993) 3427--3431,
  [\href{http://arxiv.org/abs/http://arXiv.org/abs/gr-qc/9307038}{{\tt
  http://arXiv.org/abs/gr-qc/9307038}}].

\bibitem{IyeWal94}
V.~Iyer and R.~M. Wald, {\it Some properties of {N}oether charge and a proposal
  for dynamical black hole entropy},  {\em Phys. Rev. D} {\bf 50} (1994)
  846--864,
  [\href{http://arxiv.org/abs/http://arXiv.org/abs/gr-qc/9403028}{{\tt
  http://arXiv.org/abs/gr-qc/9403028}}].

\bibitem{IyeWal95}
V.~Iyer and R.~M. Wald, {\it A comparison of {N}oether charge and {E}uclidean
  methods for computing the entropy of stationary black holes},  {\em Phys.
  Rev. D} {\bf 52} (1995) 4430--4439,
  [\href{http://arxiv.org/abs/http://arXiv.org/abs/gr-qc/9503052}{{\tt
  http://arXiv.org/abs/gr-qc/9503052}}].

\bibitem{HawEll}
S.~W. Hawking and G.~F.~R. Ellis, {\em The large scale stucture of space-time}.
\newblock Cambridge University Press, Cambridge, England, 1973.

\bibitem{Cam18}
J.~Camps, {\it {Superselection Sectors of Gravitational Subregions}},  {\em
  JHEP} {\bf 01} (2019) 182, [\href{http://arxiv.org/abs/1810.01802}{{\tt
  arXiv:1810.01802}}].

\bibitem{evans10}
L.~C. Evans, {\em Partial differential equations}.
\newblock American Mathematical Society, Providence, R.I., 2010.

\bibitem{Igor}
I.~Rodnianski. Private communication.

\bibitem{LewMal13}
A.~Lewkowycz and J.~Maldacena, {\it {Generalized gravitational entropy}},  {\em
  JHEP} {\bf 1308} (2013) 090, [\href{http://arxiv.org/abs/1304.4926}{{\tt
  arXiv:1304.4926}}].

\bibitem{MarWal19}
D.~Marolf, A.~C. Wall, and Z.~Wang, {\it {Restricted Maximin surfaces and HRT
  in generic black hole spacetimes}},
  \href{http://arxiv.org/abs/1901.03879}{{\tt arXiv:1901.03879}}.

\bibitem{Wal11}
A.~C. Wall, {\it {A proof of the generalized second law for rapidly changing
  fields and arbitrary horizon slices}},  {\em Phys.Rev.} {\bf D85} (2012),
  no.~6 104049, [\href{http://arxiv.org/abs/1105.3445}{{\tt arXiv:1105.3445}}].

\bibitem{GilbargTrudinger}
D.~Gilbarg and N.~Trudinger, {\em Elliptic Partial Differential Equations of
  Second Order}.
\newblock Classics in Mathematics. U.S. Government Printing Office, 2001.

\bibitem{ProtterWeinberger}
M.~Protter and H.~Weinberger, {\em Maximum Principles in Differential
  Equations}.
\newblock Partial differential equations. Springer New York, 1999.

\bibitem{AndMar07}
L.~Andersson, M.~Mars, and W.~Simon, {\it {Stability of marginally outer
  trapped surfaces and existence of marginally outer trapped tubes}},  {\em
  Adv. Theor. Math. Phys.} {\bf 12} (2008), no.~4 853--888,
  [\href{http://arxiv.org/abs/0704.2889}{{\tt arXiv:0704.2889}}].

\bibitem{Swe92}
G.~Sweers, {\it {Strong positivity in $C(\overline{\Omega})$ for elliptic
  systems}},  {\em Mathematische Zeitschrift} {\bf 209} (Jan, 1992) 251.

\bibitem{BlaCas17}
D.~Blanco, H.~Casini, M.~Leston, and F.~Rosso, {\it {Modular energy
  inequalities from relative entropy}},  {\em JHEP} {\bf 01} (2018) 154,
  [\href{http://arxiv.org/abs/1711.04816}{{\tt arXiv:1711.04816}}].

\bibitem{Fla97}
E.~E. Flanagan, {\it {Quantum inequalities in two-dimensional Minkowski
  space-time}},  {\em Phys. Rev.} {\bf D56} (1997) 4922--4926,
  [\href{http://arxiv.org/abs/gr-qc/9706006}{{\tt gr-qc/9706006}}].

\bibitem{FewHol04}
C.~J. Fewster and S.~Hollands, {\it {Quantum energy inequalities in
  two-dimensional conformal field theory}},  {\em Rev. Math. Phys.} {\bf 17}
  (2005) 577, [\href{http://arxiv.org/abs/math-ph/0412028}{{\tt
  math-ph/0412028}}].

\bibitem{Few12}
C.~J. Fewster, {\it {Lectures on quantum energy inequalities}},
  \href{http://arxiv.org/abs/1208.5399}{{\tt arXiv:1208.5399}}.

\bibitem{ForHel02}
L.~H. Ford, A.~D. Helfer, and T.~A. Roman, {\it {Spatially averaged quantum
  inequalities do not exist in four-dimensional space-time}},  {\em Phys. Rev.}
  {\bf D66} (2002) 124012, [\href{http://arxiv.org/abs/gr-qc/0208045}{{\tt
  gr-qc/0208045}}].

\bibitem{garroni95}
M.~G. Garroni and J.-L. Menaldi, {\it Maximum principles for
  integro-differential parabolic operators},  {\em Differential Integral
  Equations} {\bf 8} (1995), no.~1 161--182.

\bibitem{FisKel12}
S.~Fischetti, W.~Kelly, and D.~Marolf, {\it {Conserved Charges in
  Asymptotically (Locally) AdS Spacetimes}},
  \href{http://arxiv.org/abs/1211.6347}{{\tt arXiv:1211.6347}}.

\bibitem{Lei17}
S.~Leichenauer, {\it {The Quantum Focusing Conjecture Has Not Been Violated}},
  \href{http://arxiv.org/abs/1705.05469}{{\tt arXiv:1705.05469}}.

\bibitem{FisWis16}
S.~Fischetti and T.~Wiseman, {\it {A Bound on Holographic Entanglement Entropy
  from Inverse Mean Curvature Flow}},  {\em Class. Quant. Grav.} {\bf 34}
  (2017), no.~12 125005, [\href{http://arxiv.org/abs/1612.04373}{{\tt
  arXiv:1612.04373}}].

\bibitem{Gri13}
P.~Grinfeld, {\em Introduction to Tensor Analysis and the Calculus of Moving
  Surfaces}.
\newblock Springer, New York, 2013.

\end{thebibliography}\endgroup

\end{document}